\documentclass [12pt]{article}

\usepackage{amsmath,amsthm,amscd,amssymb}

\usepackage[small, centerlast, it]{caption}
\usepackage{epsfig}
\usepackage[titles]{tocloft}
\usepackage[latin1]{inputenc}
\usepackage{verbatim} 
\usepackage[active]{srcltx} 
\usepackage{fancyhdr}
\usepackage{caption}

\def\b1{{1\!\!1}}

\def\cA{{\ca A}}
\def\cB{{\ca B}}

\def\cI{{\ca I}}
\def\cJ{{\ca J}}
\def\cL{{\ca L}}

\def\cP{{\ca P}}

\def\cZ{{\ca Z}}

\def\sK{{\mathsf K}}
\def\sM{{\mathsf M}}
\def\sN{{\mathsf N}}
\def\sH{{\mathsf H}}

\def\sT{{\mathsf T}}

\def\bC{{\mathbb C}}           

\def\bK{{\mathbb K}}
\def\bH{{\mathbb H}}

\def\bD{{\mathbb D}}
\def\bN{{\mathbb N}}

\def\bR{{\mathbb R}}

\def\bZ{{\mathbb Z}}

\def\gA{{\mathfrak A}}       
\def\gB{{\mathfrak B}}

\def\gg{{\mathfrak g}}

\def\gM{{\mathfrak M}}

\def\gO{{\mathfrak O}}

\def\gR{{\mathfrak R}}

\def\gU{{\mathfrak U}}

\def\gZ{{\mathfrak Z}}
\def\gp{{\mathfrak p}}

\def\beq{\begin{eqnarray}}
\def\eeq{\end{eqnarray}}

\newcommand{\ca}[1]{{\cal #1}}         

\newcommand{\V}[1]{{\bf{#1}}}

\usepackage{sectsty}
\sectionfont{\normalsize}
\subsectionfont{\normalsize\normalfont\itshape}
\def\emptyline{\\[12pt]}

\newtheoremstyle{plain}
{5pt}
{9pt}
{\itshape}
{}
{\itshape\bfseries}
{}
{1em}
{}

\theoremstyle{thm}
\newtheorem{theorem}{\em Theorem}[section]
\newtheorem{lemma}[theorem]{\em Lemma}
\newtheorem{corollary}[theorem]{\em Corollary}
\newtheorem{proposition}[theorem]{\em Proposition}
\newtheorem{definition}[theorem]{\em Definition}
\newtheorem{remark}[theorem]{\em Remark}

\begin{document}


\par
\bigskip
\large
\noindent
{\bf  Quantum theory in real Hilbert space: How the complex Hilbert space structure emerges from Poincar\'e symmetry}
\bigskip
\par
\rm
\normalsize


\noindent  {\bf Valter Moretti$^{a}$}, {\bf Marco Oppio$^{b}$}\\
\par

\noindent 
 Department of  Mathematics University of Trento, and INFN-TIFPA \\
 via Sommarive 15, I-38123  Povo (Trento), Italy.\smallskip

\noindent $^a$valter.moretti@unitn.i, $^b$marco.oppio@unitn.itt\\
 \normalsize

\par

\rm\normalsize

\noindent {\small June, 1st  2017 - corrected version}

\rm\normalsize


\par
\bigskip

\noindent
\small
{\bf Abstract}.
As earlier conjectured by several authors  and much later established by  Sol\`er (relying on partial results by Piron, Maeda-Maeda and other authors), from the lattice-theory point of view,   Quantum Mechanics  may be formulated in real, complex or quaternionic  Hilbert spaces only. 
St\"uckelberg provided some physical, but not mathematically rigorous, reasons for ruling out  the real Hilbert space formulation, assuming that any formulation should encompass a statement of Heisenberg principle.    
 Focusing on this issue from another --in our opinion deeper-- viewpoint, we argue that there is a general fundamental  reason  why elementary  quantum systems are not described in real Hilbert spaces. It is their basic symmetry group. 
In the first part of the paper, we  consider  an  elementary relativistic system within  Wigner's approach  defined as a 
locally-faithful
 irreducible strongly-continuous unitary representation of the Poincar\'e group in a real Hilbert space. We prove that, if the squared-mass operator is non-negative, the system admits a natural, Poincar\'e invariant and unique  up to sign, complex structure which commutes with the whole algebra of observables generated by the representation itself.  This complex structure leads to a physically equivalent reformulation of the theory in a complex Hilbert space. Within this complex  formulation, differently from what happens in the real one,  all selfadjoint operators represent observables in accordance with  Sol\`er's thesis, and the standard quantum version of Noether theorem may be formulated.
In the second part of this work we focus on the physical hypotheses adopted to define a quantum elementary relativistic system relaxing them on the one hand, and making our model physically more general on the other hand.  We use a physically more accurate notion of irreducibility regarding the algebra of observables only, we describe the symmetries in terms of automorphisms of the restricted lattice of elementary propositions  of the  quantum system and we adopt a notion of continuity referred to the states viewed as probability measures on the elementary propositions.
 Also in this case, the final result proves that there exist a unique (up to sign) Poincar\'e invariant complex structure making the theory complex and completely fitting into Sol\`er's picture.  This complex structure reveals a nice interplay of Poincar\'e symmetry and the classification of the commutant of irreducible real von Neumann algebras.
\normalsize

\begin{flushright}
{\em In  memory of Rudolf Haag}
\end{flushright}
\newpage
\tableofcontents

\section{Introduction}\label{sec1} 
\subsection{The three Hilbert space formulations  permitted by Sol\`er's theorem }\label{SecI1}
Quantum theory can basically be formulated  in terms of  a non-Boolean probability theory over the partially ordered set of {\bf  elementary propositions} $\cL$ about the given physical quantum system \cite{BeCa,V2,Redei}.  {\bf Elementary propositions}, also called {\bf elementary observables}, are the experimentally testable propositions admitting the only possible outcomes $0$ and $1$. The partial order relation  $\leq$  in $\cL$ is the logical implication (many slightly different interpretations are possible actually \cite{Mackey,BeCa,V2,librone}). With some noticeable exceptions \cite{Mackey}, many authors assume that the partially ordered set $\cL$ is more strongly a {\bf lattice}. In other words, a  pair of elements $a,b \in \cL$ always admits  $\inf\{a,b\} \in \cL$ indicated by $a \wedge b$ and called 
  {\bf meet},  and   always admits  $\sup\{a,b\}\in \cL$  indicated by $a \vee b$ and called 
 {\bf join}. It is immediate to see that $a\le b$ if and only if $a=a\wedge b$. It turns out that   $\vee$ and $\wedge$ are separately {\em symmetric} and {\em associative} in every lattice.
The  $\cL$  is also requested to  be a {\bf bounded} lattice: A {\em minimal element} ${\bf 0}$, the always false proposition,  and a {\em maximal element} ${\bf 1}$, the always true proposition,  of $\cL$ are also assumed  to exist  in $\cL$.  
$\cL$  is also supposed  to be {\bf orthocomplemented}:  For every element $a\in \cL$, an {\bf orthogonal complement}  $a^\perp\in \cL$ is defined and interpreted as the logical negation of $a$.
The orthocomplement is defined by requiring  $a \vee a^\perp = {\bf 1}$,  $a \wedge a^\perp = {\bf 0}$,   $(a^\perp)^\perp = a$,
and  $a\leq b$ implies $b^\perp  \leq a^\perp$ for any $a,b \in \cL$. With these definitions, $a,b \in \cL$ are {\bf orthogonal}, written $a \perp b$, if $a \leq b^\perp$ (equivalently $b \leq a^\perp$. )\\
If $\cL_1$, $\cL_2$ are  orthocomplemented  lattices,  a map $h : \cL_1 \to \cL_2$ is a {\bf lattice homomorphism} if $f(a\vee_1 b) = h(a) \vee_2 h(b)$,  $f(a\wedge_1 b) = h(a) \wedge_2 h(b)$, $h(a)^{\perp_2} = h(a^{\perp_1})$ if $a,b \in \cL_1$, $h({\bf 0}_1)= {\bf 0}_2$, $h({\bf 1}_1)= {\bf 1}_2$.
When  the lattices are complete, resp. $\sigma$-complete, (see Appendix A (ii)) the first pair of conditions are made stronger to
 $h(\sup_{a\in A} a) = \sup_{a\in A} h(a)$ and $h(\inf_{a\in A} a) = \inf_{a\in A}  h(a)$ for every 
 infinite, resp. countably infinite, subset $A \subset \cL_1$. 
A straightforward calculation shows that $a\le_1 b$ implies $h(a)\le_2 h(b)$.
A bijective lattice homomorphism is a {\bf lattice isomorphism}. The inverse map of a lattice isomorphism is a lattice isomorphism as well. {\bf Lattice automorphisms} are isomorphisms with $\cL_1=\cL_2$; they give rise to a group, denoted by $\mbox{Aut}(\cL_1)$.\\
A pair of mutually {\em compatible} elementary propositions (those which are simultaneously testable by means of experiment)  is assumed to be represented by {\bf commuting} elements $p,q \in \cL$  in the sense of abstract orthocomplemented lattices. By definition \cite{BeCa} it means that the sublattice {\bf generated} by $\{p, q\}$, namely the intersection of all orthocomplemented sublattices of  $\cL$ which include $\{p,q\}$  is  {\bf  Boolean}: $\vee$ and $\wedge$ are mutually {\em distributive}.
If restricting ourselves to a maximal set of pairwise compatible propositions, we have a  complete Boolean sublattice and an interpretation in terms of {\em classical logic} turns out to be appropriate. Since compatibility of pair of propositions is not a transitive relation, the structure of maximal boolean subsets of $\cL$ is very complex.
The whole lattice ${\cal L}$ of elementary propositions of a quantum system  is however  {\em non-Boolean}, in particular $\wedge$ and $\vee$  are  not mutually distributive. This obstruction to distributivity is physically  due  to  the existence of 
pairwise {\em incompatible} elementary propositions (e.g., see \cite{BeCa,M}).
The  non-Boolean nature of $\cL$  was and still is nowadays the starting point for  interpretations of the formalism in terms of  {\em quantum logics} instead of classical logics \cite{librone}.
Generally speaking, the quantum lattice $\cL$ seems to enjoy a list of specific  features 
which one may try to justify from the known quantum phenomenology (e.g., see \cite{BeCa}) even if some deep interpretative problems remain \cite{librone}.
We merely  list these properties  without entering into the details  \cite{BeCa} (see Appendix \ref{Alattices} for a brief illustration of these definitions): 
(i) {\bf orthomodularity},  (ii) {\bf $\sigma$-completeness},  (iii) {\bf  atomicity}, (iii)' {\bf atomisticity},  (iv) {\bf covering property}, (v) {\bf separability}, (vi) {\bf irreducibility}.\\
A  long standing problem, the so-called {\em coordination problem} \cite{BeCa},  was to prove that an abstract bounded orthocomplemented lattice ${\cal L}$ fulfilling the  properties (i)-(vi) and possibly further technical requirements,
is necessarily isomorphic to the lattice ${\cal L}(\sH)$ of the orthogonal projectors/closed subspaces of a {\em complex} Hilbert space $\sH$. This was done in order to recover the standard Hilbert-space formulation of Quantum Theory.
Some intermediate, but fundamental,  results  due to Piron \cite{Piron} and next to  Maeda-Maeda \cite{MM} 
demonstrated that such ${\cal L}$, if contains at least four orthogonal atoms, must be  isomorphic to the lattice of  the {\em orthoclosed} subspaces ($K=K^{\perp\perp}$)  of a structure generalizing a vector space over a division ring $\bD$ equipped with a suitable involution operation, and admitting a generalized non-singular $\bD$-valued Hermitian scalar product (giving rise to the above mentioned notion of orthogonal $^\perp$).
The order relation of  this concrete lattice is the standard inclusion of orthoclosed subspaces.  In 1995 Sol\`er \cite{Soler} achieved the perhaps conclusive result. Consider an orthocomplemented bounded lattice $\cL$ satisfying (i)-(vi), such that (vii) it
contains at least four orthogonal atoms (so that the above generalized Hermitian scalar product exists)
and (viii) $\cL$ includes an infinite orthogonal sets of atoms with unit (generalized) norm. With these hypotheses (for alternative equivalent requirements see \cite{Holland} and \cite{XY}), the thesis of Soler's theorem reads: \\

\noindent [{\bf Sth}]\:\:  {\em  The lattice  $\cL$ of quantum elementary propositions is  isomorphic to the lattice ${\cal L}(\sH)$ of (topologically) closed subspaces  of a separable Hilbert space $\sH$ with set of scalars given by either the fields $\bR$, $\bC$ or the real division algebra of quaternions $\bH$}.  \\

\noindent  The quaternionic Hilbert space structure is defined in Appendix \ref{QHS}. In all three cases,  the partial order relation of the lattice is again the standard inclusion of closed subspaces and 
$\sM\vee\sN$ corresponds to the closed span of the union of the closed subspaces $\sM$ and $\sN$, whereas  $\sM\wedge\sN:=\sM \cap \sN$. The minimal element  is the trivial subspace $\{0\}$ and the maximal element is  $\sH$ itself. Finally, the  orthocomplement of $\sM \in \cL(\sH)$ is described by the standard orthogonal $\sM^\perp$ in $\sH$. All the structure can equivalently be rephrased in terms of orthogonal projectors $P$ in $\sH$, since they are one-to-one associated with 
the closed subspaces   of $\sH$ identified with their images $P(\sH)$. In particular $P\leq Q$ (namely $P(\sH)\subset Q(\sH)$) corresponds to the logical implication $P \Rightarrow Q$, for $P,Q\in \cL(\sH)$.
Relaxing the irreducibility requirement on $\cL$, requirement physically corresponding to the absence of {\em superselection rules}, an orthogonal  direct sum of many such Hilbert spaces (even over different set of scalars)  replaces the single Hilbert space $\sH$. \\
Sol\`er's theorem relies upon a list of  rigid postulates on the lattice $\cL$ and the arising  picture stated in  {\em Sth} turns out to be  equally rigid. Regarding the hypotheses, in particular, no reference to physically fundamental symmetries, like Galileo or Poincar\'e ones are included.
Looking at the thesis {\em Sth} in complex Hilbert spaces, we see that  only type-$I$ factors are admitted to represent the algebra of observables $\gR$  and no gauge group may enter the game excluding, for instance, systems of {\em quarks}  where internal
 symmetries (color $SU(3)$) play a crucal r\^ole.  Sol\`er's   picture is evidently not appropriate also  to describe   non-elementary quantum systems like  {\em pure phases} of extended quantum thermodynamical systems. There, always referring to 
complex Hilbert space description, the algebra of observables is still a factor, but the 
 type-$I$ is not admitted in general due to the presence of a non-trivial commutant $\gR'$.  Also {\em localized} algebras of observables in QFT are not encompassed by Sol\`er's framework. 
 As a matter of fact, {\em elementary relativistic systems} like elementary particles in Wigner's view are however in agreement with {\em Sth}  when we confine ourselves to deal with a {\em complex} Hilbert space $\sH$. Since these elementary relativistic systems are characterized by {\em irreducible} unitary  representations of Poincar\'e group and assuming that the von Neumann algebra of observables is that generated by the representation,  Schur's lemma  implies that  the algebra of observables coincides with the whole  $\gB(\sH)$. Therefore the lattice of  
elementary propositions is the entire $\cL(\sH)$, just  as stated in {\em Sth}. What happens when changing the set of scalars of the Hilbert space, passing from $\bC$ to $\bR$ or $\bH$ is not obvious.

\subsection{Quantum notions common to the three formulations}\label{seclist}
The following theoretical notions used to axiomatize quantum mechanics  are defined in the afore-mentioned  separable Hilbert space $\sH$, with scalar product $( \cdot|\cdot)$, over $\bR$, $\bC$ or $\bH$ respectively and referring to the quantum lattice $\cL(\sH)$.
However these notions are defined also replacing $\cL(\sH)$ for a smaller  lattice $\cL_1(\sH) \subset \cL(\sH)$, provided it is still  orthocomplemented and $\sigma$-complete (and therefore also orthomodular and separable). For future convenience, we shall list these notion below in this generalized case.

(1) {\bf Elementary  observables} are represented by the orthogonal projectors in $\cL_1(H)$. Two such projectors are said to be {\bf compatible} if they commute as operators.
Indeed the abstract commutativity notion  of  elementary observables turns out to be equivalent to the standard  commutativity of associated orthogonal projectors. 

(2) {\bf Observables} are the  Projector-Valued Measures  (PVMs) over the real Borel sets  (see Def.\ref{defPVM}) taking values in ${\cal L}_1(\sH)$ $$ {\cal B}(\mathbb R) \ni E \mapsto P^{(A)}(E)\in {\cal L}_1(\sH)\:.$$ 
Equivalently, \cite{V2} an observable  is a selfadjoint operator $A : D(A) \to \sH$ with $D(A)\subset \sH$ a dense subspace
such that the associated projector-valued measure is made by elements of ${\cal L}_1(\sH)$. The link with the previous notion is the  statement of the spectral theorem for selfadjoint operators
 $A = \int_{\sigma(A)} \lambda dP^{(A)}(\lambda)$ (Theorem \ref{st} in appendix for the real and complex case, for the quaternionic case see \cite{V2}). Obviously the meaning of each  elementary proposition $P^{(A)}(E)$ is {\em the outcome of the measurement of $A$ belongs  to the real Borel set $E$}.
Evidently, $\cL_1(\sH)= \cL(\sH)$ if and only if  {\em every} selfadjoint operator in $\sH$ represents an observable.
A selfadjoint operator, in particular an observables, $A$ is said to be {\bf compatible} with another  selfadjoint operator, in particular an observables, $B$ when the respective PVMs are made of pairwise commuting projectors.

(3) {\bf Quantum states} are defined  as {\bf $\sigma$-additive  probability measures} over ${\cal L}_1(\sH)$, that is maps
$\mu : {\cal L}_1(\sH) \to [0,1]$
such that $\mu(I)=1$ and $$\mu\left(s\mbox{-}\sum_k P_k\right) = \sum_k \mu(P_k)\quad \mbox{if $\{P_k\}_{k \in \bN}\subset \cL_1(\sH)$ with $P_kP_h=0$ for $h \neq k$,}$$ 
$s\mbox{-}\sum_k$ denoting the series in the strong operator topology.
$\mu(P)$ has the meaning of {\em the probability that the outcome of $P$ is $1$ if the proposition is tested when the state is $\mu$}. \\
If $\cL_1(\sH)= \cL(\sH)$ for $\sH$ separable with  $+\infty \geq \dim (\sH) \neq  2$ (always assumed henceforth),
these measures are in one-to-one correspondence with all of the selfadjoint positive, unit-trace, trace class operators $T_\mu: \sH\to \sH$ according to
$$\mu(P)= tr(T_\mu P)\quad \forall P \in  {\cal L}(H)\:.$$
This correspondence exists  for  the three cases as demonstrated by the celebrated {\em Gleason's theorem} valid for $\bR$ and $\bC$ \cite{G}, and finally extended by Varadarajan to the $\bH$ case \cite{V2}. The result holds true (but the correspondence ceases to be one-to-one) for separable complex Hilbert spaces when $\cL_1(\sH)\subsetneq \cL(\sH)$ and $\cL_1(\sH)$ is the projector lattice of a von Neumann algebra whose canonical decomposition into definite-type von Neumann algebras does not contains type-$I_2$ algebras \cite{libroGleason}.

 (4) {\bf Pure states} are  extremal elements of the convex body of the afore-mentioned probability measures. If $\cL_1(\sH)= \cL(\sH)$ pure states  are one-to-one with unit vectors of $\sH$ up to {\bf (generalized) phases} $\eta$, i.e., numbers of $\bR, \bC, \bH$  respectively,  with $|\eta|=1$. In this case, the notion of {\bf  probability transition} $|( \psi|\phi)|^2$ of a pair of pure states defined by unit vectors $\psi,\phi$
can be introduced.
$|(\psi|\phi)|^2 = \mu_\psi(P_\phi)$ is the probability that $P_\phi$ is true when the state is $\mu_\psi$, where  $P_\phi = (\phi| \cdot )\phi$ and  $\mu_\psi := (\psi| \cdot \psi)$.

(5) {\bf L\"uders-von Neumann's post measurement axiom} can be formulated in the standard way in the three cases:
{\em If the outcome of the ideal measurement of $P \in \cL_1(\sH)$ in the state $\mu$ is $1$, the post measurement state is} 
$$\mu_P(\cdot) := \frac{\mu(P \cdot P)}{\mu(P)}\:.$$
If $\cL_1(\sH)= \cL(\sH)$, we may define states in terms of trace class operators and, with obvious notation,
$T_P = \frac{1}{tr(PT)}PTP$. In terms of probability measures over $\cL(\sH)$, this  is equivalent to say that the post measurement measure $\mu_P$, when the state before the measurement of $P$ is $\mu$,  is the {\em unique} probability-measure over $\cL(\sH)$  satisfying  the natural requirement of conditional probability
$\mu_P(Q) = \frac{\mu(Q)}{{\mu(P)}}$, for every $Q\in \cL(\sH)$ with $Q\leq P$.

(6) {\bf Symmetries} are naturally defined  as {\em automorphisms} $h : \cL_1(\sH) \to \cL_1(\sH)$ of the lattice of elementary propositions. A {\em subclass} of symmetries $h_U$ are those induced by 
unitary (or also anti unitary in the complex case) operators $U\in \gB(\sH)$ by means of 
 $h_U(P) := UPU^{-1}$ for every $P \in \cL_i(\sH)$. 
Alternatively, another definition of symmetry is as {\em automorphism} of the {\em Jordan algebra of observables} constructed out of $\cL_1(\sH)$. 
If $\cL_1(\sH)=\cL(\sH)$, following Wigner, symmetries can be defined as {\em bijective}   {\em probability-transition preserving} transformations of pure states to pure states. \\
With the maximality hypothesis on the lattice, the three  notions of symmetry coincide. In this case, {\em all} symmetries turn out to be  described  by  unitary  (or anti unitary in the complex case) operators, up to constant phases of $\bR$, $\bC$, $\bH$, respectively  due to well known theorems by Kadison, 
Wigner and Varadarajan \cite{Simon,V2}. 

(7) {\bf Continuous symmetries} are one-parameter groups of lattice automorphisms $\bR \ni s \mapsto h_s$, such that 
$\bR \ni s \mapsto \mu(h_s(P))$ is continuous for every $P\in \cL_1(\sH)$ and every quantum state $\mu$ ($\bR$ may be replaced for a topological group but we  stick here to the simplest
case). The {\bf time evolution} of the system $\bR \ni s \mapsto \tau_s$ is a preferred  continuous symmetry parametrized over $\bR$. 

(8) A {\bf dynamical symmetry} is a continuous symmetry $h$ which  commutes with the time evolution, $h_s \circ \tau_t = \tau_t \circ h_s$ for $s,t \in \bR$.\\
If $\cL_1(\sH)= \cL(\sH)$, every  continuous symmetry $\bR \ni s \mapsto h_s$ is represented by a strongly continuous
one-parameter group of unitary operators $\bR \ni s \mapsto U_s$ such that $h_s(P)= U_sPU_s^{-1}$ for all $P\in \cL(\sH)$ \cite{V2}.
Versions of {\em Stone theorem}  hold in the three considered cases $\bR$, $\bC$ and $\bH$ (the  validity in the quaternionic case easily arises 
form the theory developed in \cite{GMP2}), proving that   $U_s = e^{sA}$ for some {\em anti}-selfadjoint operator $A$, uniquely determined by $U$.
 In the complex case, if $\bR \ni s \mapsto e^{sA}$ is also a  dynamical symmetry, the {\em selfadjoint} operator $-iA$, which is an observable the lattice being maximal, is invariant  under the natural adjoint action of time evolution $\tau$ unitarily represented by $\bR \ni t \mapsto V_t$,
and thus $-iA$ it is a {\bf constant of motion}, $V_t^{-1}(-iA) V_t =-iA$ for every $t\in \bR$. This is the celebrated quantum version of {\em Noether theorem}. In the real Hilbert space case, no such simple result exists, since
we have no general way to construct a selfadjoint operator out of an anti selfadjoint operator $A$ in  absence of $i$. There is no unitary operator  $J$ corresponding to the imaginary unit $iI$ which commutes with  the  anti selfadjoint generator $A$ of every possible continuous symmetry (the time evolution in particular), thus producing an associated observable $JA$ which is a constant of motion.  Such an operator  however may exist for  one or groups of observables.
 In the quaternionic case, contrarily, there are many, pairwise  non-commuting, imaginary unities as recently established 
 \cite{GMP2}.
An interesting physical discussion on these partially open  issues for the quaternionic formulation appears in \cite{Adler}.

\subsection{Fake real Hilbert space formulation and St\"uckelberg's analysis}\label{SecI3} Focusing on the description of quantum theories in real Hilbert space and complex Hilbert space, a crucial fact which makes a sharp distinction between these two descriptions,  regards the correspondence of (pure) states and unit  vectors of $\sH$. 
Assuming that the quantum lattice is the whole $\cL(\sH)$, while in a complex theory pure states are one-to-one with unit vectors up to {\em phases}, in a real theory pure states are one-to-one with unit vectors up to {\em signs}.
Therefore a quantum theory formulated in a real Hilbert space is {\em not} a theory formulated in a complex Hilbert space 
where states are decomposed in real and imaginary parts and where $i$ is simply hidden in the (fake) real formalism.
Suppose that $\sH= L^2(\bR, dx)$ viewed as space of {\em complex} wavefunctions. A complex  wavefunction $\psi$  can always be 
decomposed into a pair of real wavefunctions $\psi_1 = Re \psi$ and $\psi_2 = Im \psi$ and all the theory can be recast into the real Hilbert space $L^2_\bR(\bR, dx)\oplus L^2_\bR(\bR, dx)$, where $L^2_\bR(\bR, dx)$ indicates the real Hilbert space of real-valued square-integrable functions. A $\bC$-linear operator in $L^2(\bR, dx)$ admits a decomplexificated corresponding $\bR$-linear operator in $L^2_\bR(\bR, dx)\oplus L^2_\bR(\bR, dx)$ accordingly, and a selfadjoint operator in $L^2(\bR, dx)$ induces a selfadjoint operator in $L^2_\bR(\bR, dx)\oplus L^2_\bR(\bR, dx)$ this way.
 This real representation has nothing to do with the thesis of Sol\`er's  theorem in the real Hilbert space case, since

(a) pure states turn out to be one-to-one with unit  vectors  of  the real space $L^2_\bR(\bR, dx)\oplus L^2_\bR(\bR, dx)$ {\em up to a rotation of $SO(2)$} (arising from the decomposition in real and imaginary part of $e^{i\theta}\psi$) and not {\em up to a sign};

(b) not all selfadjoint operators of the real Hilbert space represent observables here, since not all operators in $L^2_\bR(\bR, dx)\oplus L^2_\bR(\bR, dx)$ descend from operators in $L^2(\bR, dx)$.\\

\noindent Let us stick a while to the analysis of this fake real model for further observations. The operator $iI$ of the complex Hilbert space induces  a non-diagonal decomplexified operator $J$ in the real Hilbert space, with the properties $JJ=-I$ and $J^*=-J$. These types of operators in real Hilbert spaces are called {\em complex structures}. 
$J$ permits to reconstruct back an isomorphic version of  the initial complex Hilbert space using the vectors of $L^2_\bR(\bR, dx)\oplus L^2_\bR(\bR, dx)$. This happens just defining the product of complex numbers and vectors like this
$$(a+ib) \Psi := (a+ bJ) \Psi\quad \mbox{where $a,b \in \bR$ and $\Psi \in L^2_\bR(\bR, dx)\oplus L^2_\bR(\bR, dx)$,}$$
also complexifying the natural scalar product of $L^2_\bR(\bR, dx)\oplus L^2_\bR(\bR, dx)$ as we shall discuss  into a very general fashion later.
The crucial property of $J$, in relation with (b) above, is that it permits us to distinguish between selfadjoint operators in  $L^2_\bR(\bR, dx)\oplus L^2_\bR(\bR, dx)$ constructed out of selfadjoint operators in $L^2(\bR, dx)$  through the decomplexification procedure and the remaining unphysical selfadjoint operators in  $L^2_\bR(\bR, dx)\oplus L^2_\bR(\bR, dx)$ not representing observables. In fact, also dropping the selfadjointness requirement, a $\bR$-linear operator  $A$ in $L^2_\bR(\bR, dx)\oplus L^2_\bR(\bR, dx)$  arises form a corresponding $\bC$-linear operator in  $L^2(\bR, dx)$ if and only if $AJ=JA$.\\
 {\em From this remark we conclude that a quantum theory apparently formulated in a real Hilbert space $\sH$ may actually be a standard theory, formulated in a corresponding  complex Hilbert space $\sH_J$. It happens  if there is a complex structure $J$ which commutes with every observable of the theory}.

\noindent If such a $J$ exists, also the correspondence of states and vectors {\em up to signs} (as in (a) above) fails in $\sH$. Indeed the algebra of observables, just due to the presence of $J$ cannot coincide with the whole class of (real) selfadjoint operators and vectors $\Psi$ and $e^{\theta J}\Psi$ for every $\theta \in \bR$ cannot be distinguished by means of physical measurements. For instance, if $A$ is an observable and $(\cdot|\cdot)$ is the real scalar product in $\sH$,
$$(e^{\theta J}\Psi|Ae^{\theta J}\Psi) = (\Psi|e^{-\theta J}Ae^{\theta J}\Psi) = 
 (\Psi|Ae^{-\theta J}e^{\theta J}\Psi) = (\Psi|A\Psi)\:.$$
Even passing from a fake real formulation to a corresponding  complex formulation by 
means of a complex structure commuting with all the observables, it is still possible that the found class of observables in
 the final complex Hilbert space is however smaller that the whole set of $\bC$-linear selfadjoint operators.
Nevertheless the  final complex Hilbert space formulation has less redundancy than the initial real formulation, since the 
selfadjoint-operators/observables ratio has increased.\\
Independently form the result by Sol\`er, the theoretical possibility of formulating quantum theories in 
Hilbert spaces over either $\bR$ or $\bH$  (or other division rings of scalars) \cite{BeCa}  was matter of investigation since 
the early mathematical formulations of Quantum Mechanics. 
However, differently from quaternionic quantum mechanics \cite{foudationofquaternionicmechanics,Adler} which 
still deserves some theoretical interest, real quantum mechanics was  not considered as physically interesting almost
 immediately especially in view of well-known St\"uckelberg's  analysis \cite{S1,S2} in the early seventies.
As a matter of fact, St\"uckelberg \cite{S1,S2} provided some  physical reasons for getting rid of the real Hilbert space
 formulation relying on the demand  that every conceivable formulation of Quantum Mechanics  should include the 
statement of Heisenberg principle.  
He argued that  the statement of Heisenberg principle requires the existence of a  natural complex structure $J$ commuting with all physical observables and thus
 making the theory complex as observed above. His
 analysis  is definitely  physically interesting,  but very poor from a mathematical viewpoint as it assumes that all observables 
have pure point spectrum and some of them are bounded,  in contradiction with the nature of position of momentum observables necessary to state 
Heisenberg principle. No discussion about domains  appears. Many inferences are just heuristically  justified (including the universality of $J$)  even if they all are physically plausible.
Moreover, in  St\"uckelberg's analysis,  the existence of $J$ seems to be more a {\em sufficient} condition
 to guarantee the validity of Heisenberg inequalities rather than a {\em necessary} requirement,
 since everything  is based on an {\em a priori} and arbitrary (though physically very plausible)  model of any version of uncertainty principle as described  
in Sect.2 of \cite{S1}.  
 Finally,  the validity of
 Heisenberg principle cannot be viewed as a fundamental {\em a priori}  condition  nowadays:  it needs  the existence of 
the {\em position observable} which is a very delicate issue, both theoretically and mathematically (it is based on Mackey's imprimitivity
machinery) in case of relativistic
 elementary systems \cite{V2}.  For massless particles like photons, the position observable simply does not exist \cite{V2}. The analysis 
of this work covers  also that case instead.

\subsection{Main results and structure of this work}\label{SecI4}

The overall goal of this work is to rigorously  investigate if there are cogent physical reasons to abandon any  real Hilbert space  
formulation. Reasons deeper  than, and independent from,  the request of validity of Heisenberg principle. Obviously 
we are thinking here of  {\em elementary} quantum systems  different from  the ones which already admit descriptions 
in complex Hilbert spaces. Simultaneously we want to check how solid the final picture arising from Sol\`er's analysis stated in {\em Sth} is.
 We therefore assume that quantum theories can be formulated in a real, complex or quaternionic Hilbert space, focusing  
on the first case. The core of our analysis and the corresponding results are contained in the sections \ref{secMAIN1}
and \ref{secMAIN2}. 
We initially suppose in Section \ref{secMAIN1} that, in accordance with Wigner's view,  an elementary 
relativistic  physical system is described in a {\em  real} Hilbert $\sH$ space admitting a strongly-continuous unitary irreducible 
representation $U$ of Poincar\'e group and that the algebra of observables $\gR$ coincides with the von Neumann
 algebra $\gR_U$ generated by the said representation. This idea is encapsulated in Definition \ref{defERS}.
In this sense the group representation completely fixes the physical system.  We therefore confine ourselves to deal with elementary  systems, described in {\em real} Hilbert spaces, whose maximal group of symmetry is {\em Poincar\'e group}  (so that more complicated systems like quarks are not encompassed  by  our study). However, we do not assume that the lattice of orthogonal projectors in $\gR$ coincides with the whole $\cL(\sH)$ or is isomorphic to some $\cL(\sH')$ for some other Hilbert space (also with a different set of scalars) as in the thesis of Sol\`er's theorem {\em Sth}. We would like to either find it as a consequence of our hypotheses or to disprove it.\\
With our  hypotheses, we shall find in Theorem\ref{poinccomplexstructure}  that, remarkably, 
there must exist a unique (up to  the sign) complex structure $J$ commuting with both the group representation $U$
 and algebra of observables $\gR$. As a consequence the theory can be reformulated in a complex Hilbert space $\sH_J$
 where both the representation (which remains strongly continuous and irreducible) and the von Neumann algebra of
 observables are well-defined  and the theory admits the standard formulation. In particular $\gR$ coincides with the whole $\gB(\sH_J)$ and consequently the lattice of orthogonal  projectors coincides (and thus is isomorphic) to  $\cL(\sH_J)$ {\em in agreement  with Sol\`er's thesis, even if different hypotheses are assumed}.
This way, also the standard 
formulation of the quantum Noether theorem takes place, because we can  associate anti selfadjoint generators $A$ of Poincar\'e continuous symmetries to observables $JA$ and $J$ commutes with the time evolution. \\
 In Section \ref{secMAIN2} we will  deal with a more sophisticated theoretical  idea of an elementary relativistic system, since some issues remain open in our first formulation when dealing with real Hilbert spaces. 
In particular, the irreducibility assumption is not well motivated and should be formulated into a more physical framework regarding only observables.  
As a consequence, there is no deep reason to assume
that symmetries are represented by unitary operators and also proving it for each
element of the representation separately, there is no {\em a priori} cogent reason to suppose that 
  the representation is unitary instead of (real) projective unitary. These issues will be fixed taking advantage of 
a result (Theorem \ref{threecommutant}) about the commutant of irreducible von Neumann algebras in real Hilbert spaces. 
With the improved version of elementary relativistic system stated in Definition \ref{RRES}, we will however find  the same result already established  with the previous simpler definition. In fact, Theorem \ref{main2} proves again that
the theory can be reformulated into a complex setting in agreement with Sol\`er's thesis, exploiting a complex structure $J$ which, again, is unique up to a sign and Poincar\'e invariant.  Actually we also prove that the improved definition of relativistic elementary system though physically finer is actually mathematically equivalent to Wigner's one also in the real case. \\
\noindent The rest of this paper is organized as follows. The next two sections, Section \ref{seccomplex} and Section \ref{represent},  are devoted to collect, and in some cases autonomously prove, several results on real spectral theory and the theory of Lie group representation
in real Hilbert spaces. Section \ref{secMAIN1} and Section \ref{secMAIN2} 
discuss the notion of relativistic elementary system and present the two versions of the afore-mentioned  main result of this paper  (Theorem\ref{poinccomplexstructure}  and Theorem \ref{main2}).
Conclusions are discussed in the last section. A final appendix includes several results and proofs of intermediate propositions.
%
%
%
%
%

\section{Complexification procedures and technical results for real Hilbert spaces}\label{seccomplex}
We hereafter assume  that the reader is familiar with some standard definitions and results of  the theory of operators in either real and complex Hilbert spaces. A summary of these notions appears in Appendix \ref{secstatic}.  \\
It is possible to extend back to {\em real} Hilbert spaces some technical results valid  for {\em complex} Hilbert 
spaces like {\em Stone's theorem} or {\em Schur's lemma} and the {\em polar decomposition theorem}. These extensions 
 take advantage of a certain complexification procedure which produces  a complex Hilbert space when a real Hilbert space is given.  
Another procedure to build up  a complex Hilbert space from a real one exploits  the existence of a so called {\em complex structure}. This section is devoted to introduces these procedures and to  prove some technical results about real Hilbert spaces, in comparison with corresponding well known results in complex Hilbert space theory presumably more familiar to the reader.

\subsection{External complexified structures}
Let $\sH$ be a {\em real} Hilbert space (Definition \ref{defHRC}) whose {\em real} scalar product will be henceforth denoted by $(\cdot|\cdot)$. It is possible to define an  associate  complex Hilbert space  \cite{MV} by means of an elementary {\em external complexification procedure}. The elements of  this associate complex vector space are couples $x+iy := (x,y) \in \sH \times \sH$ and the complex linear space structure is defined by assuming that  \beq (\alpha + i\beta) (x+iy) := \alpha x - \beta y + i (\beta x + \alpha y)\quad \mbox{ $\forall x,y \in \sH$ and $\forall \alpha,\beta \in \mathbb R$}\:.\label{ls}\eeq The  scalar product  of $\sH\times \sH$ is, by definition, 
\beq ( x+iy| u+iv )_\bC := ( x|u ) +( y|v ) +i [( x|v )- (y|u)] \:,\quad \forall x,y,u,v \in \sH\label{sp}\:.\eeq
This scalar product is Hermitian in agreement with Definition \ref{defSP} with associated norm
\beq || x+iy||_\bC^2 :=  ( x+iy| x+iy )_\bC  =
||x||^2 + ||y||^2 \:,\quad \forall x,y \in \sH \label{cs}\:.\eeq

\begin{proposition}\label{propcompleX}
Let $\sH$ be a real Hilbert space with real scalar product $(\cdot|\cdot)$, the following facts hold.\\
{\bf (a)} The complex vector space over $\sH \times \sH$,
with the complex linear structure (\ref{ls}) and the Hermitian scalar product $(\cdot|\cdot)_\bC$ defined in  (\ref{sp})
is a complex Hilbert space, henceforth  denoted by $\sH_\bC$ and called {\bf external complexification} of $\sH$.\\
{\bf (b)} $N\subset \sH$ is a Hilbert basis (Def.\ref{defHB}) of the real Hilbert space $\sH$  if and only if 
$N$ is a Hilbert basis of $\sH_\bC$. Thus $\sH_\bC$ is separable if and only if $\sH$ is.\\
{\bf (c)}  If $\sK \subset \sH$ is a subspace,
$\sK_\bC := \sK \times \sK \subset \sH_\bC\:,$
turns out a to be (complex)  subspace of  $\sH_\bC$ and $\overline{\sK_\bC}=\overline{\sK}_\bC$.
\end{proposition}

\begin{proof}
From (\ref{cs}), Cauchy sequences in $\sH_\bC$ define pairs of Cauchy sequences  in $\sH$.
 For this reason $\sH_\bC$ is complete in view of the completeness of   $\sH$. The second statement is true because $N$ is maximal orthonormal  in $\sH$ iff it is maximal orthonormal  in $\sH_\bC$. The proof of (c) is immediate.
\end{proof}
\begin{remark}{\em
Notice that, thanks to (c), a subspace $\sK\subset \sH$ turns out to be closed or dense in $\sH$ if and only if $\sK_\bC$ is, respectively, closed or dense in $\sH_\bC$.}
\end{remark}
\begin{definition}\label{defconj2}
{\em If $\sH$ is a {\em complex} Hilbert space, a {\bf conjugation} is an anti linear  (Def.\ref{defopant}) norm-preserving operator $C: \sH \to \sH$ such that $CC=I$.}
\end{definition}
 \noindent A conjugation $C: \sH \to \sH$ is bijective (since  $C=C^{-1}$) and satisfies  $(Cx|Cy)= \overline{(x|y)}$ for $x,y \in \sH$
due to  the polarization identity of a complex scalar product.
\noindent Conjugations always exist.  If $N\subset \sK$ is a Hilbert basis of the complex Hilbert space $\sK$,  an associated   conjugation is 
$\sK \ni x =  \sum_{z\in N} (z|x)z \mapsto  \sum_{z\in N} \overline{(z|x)}z\in \sK\:.$
\begin{proposition} \label{remcomp}   Let $\sK$ be a {\em complex} Hilbert space and  $C_\sK : \sK \to \sK$ a conjugation. $\sK$   is isomorphic (Definition \ref{defHRC}) to $\sH_\bC$ for a certain  {\em real} Hilbert space $\sH$ associated to $C_\sK$.
\end{proposition}

\begin{proof} Define the closed real vector subspace
 $\sH = \frac{1}{2}(I+C_\sK)(\sK)$ equipped with the real scalar product given by the restriction to $\sH$ of the Hermitian scalar product  of $\sK$. $\sH$ is a real Hilbert space because is closed. The identity map $\sK \ni x \mapsto x \in \sK$ rearranged as follows:
$$\sK \ni x \mapsto \frac{1}{2}(I+C_\sK) x + i \frac{1}{2}(I+C_\sK) \frac{1}{i} x \in \sH + i \sH = \sH_\bC\:,$$
turns out to be  a complex Hilbert space isomorphism from $\sK$ to $\sH_\bC$.
\end{proof}
\noindent If $\sH_\bC$ is constructed out of the real Hilbert space $\sH$, the real Hilbert space associated to $\sK = \sH_\bC$ through Prop.\ref{remcomp}  is just $\sH$
if employing the natural conjugation $C_\sK=C$ with 
\begin{equation}\label{conj}
C : \sH_\bC \ni x+iy \mapsto x-iy \in \sH_\bC\:.
\end{equation}

\noindent Let us pass to operators  extending Remark 20.18 in  \cite{MV}.
\begin{definition}\label{defcomplx0}
{\em If $A: D(A) \to \sH$  is an $\bR$-linear operator
in the real Hilbert space $\sH$ 
 we define the ($\bC$-linear) {\bf associated complexificated operator} 
\begin{equation}\label{defcomplx}
A_\bC := A+iA : D(A)+iD(A)\ni x+iy\mapsto Ax+iAy\in \sH_\bC\:.
\end{equation}
It follows immediately that $KerA_\bC=(KerA)_\bC$ and $RanA_\bC=(RanA)_\bC$.
}
\end{definition}

\begin{remark} {\em From now on, unless differently explicitly stated, a {\em subspace} of a {\em complex} Hilbert space is a {\em complex} subspace. Similarly, an {\em operator} in a complex Hilbert space is  a {\em complex}-linear operator.}
\end{remark}
\noindent The notion of {\em adjoint} operator and its elementary properties are given in Def.\ref{defagg} and Remark \ref{sumselfadj}.
The definitions of the various types of operators we use below are listed in Def.\ref{defclosop}  and  \ref{defop} including Remark
\ref{essentclos} for their basic properties. The notion {\em spectrum}, {\em PVM} and {\em spectral integral} appear in 
Def.\ref{defspec}, Def.\ref{defPVM}, Prop.\ref{propint} and Thm \ref{st}.  Finally we henceforth adopt 
Def.\ref{defdomain} concerning the domain of composed operators.

\begin{proposition}\label{prop2}
The following facts are valid  referring to Def.\ref{defcomplx0} for  a real Hilbert space $\sH$ and the associated complexified Hilbert space $\sH_\bC$.\\
{\bf (1)} An  operator   $B : D(B) \to \sH_\bC$ with $D(B) \subset \sH_\bC$ satisfies $B=A_\bC$ for some operator  $A: D(A) \to \sH$ and  $D(A) \subset \sH$ if and only if \beq CB\subset BC \:,\label{commC}\eeq
where $C$ is the conjugation in $\sH_\bC$ defined  in (\ref{conj}).  If (\ref{commC}) holds, then $CB=BC$ and $A$ is uniquely defined by
$$Ax+i0 = B(x+i0) \ \mbox{ on }  \ D(A)=\{x\in \sH,\ x+i0\in D(B)\}$$
In the following $A : D(A) \to \sH$ is an operator in the real Hilbert space $\sH$.\\
{\bf (2)}  If $D(A)$ is dense, then $(A_\bC)^*= (A^*)_\bC$, in  particular $D((A_\bC)^*) = D(A^*)+iD(A^*)$.\\
{\bf (3)}  $A$ is either closed or closable if and only if $A_\bC$ is, respectively, closed or closable. In the second case, $\overline{A_\bC}= (\overline{A})_\bC$.\\
{\bf (4)} A  subspace $S \subset D(A)$ is a core for $A$ if and only if $S_\bC$ is a core for $A_\bC$.\\
{\bf (5)}  $A_\bC$ is symmetric, selfadjoint, anti symmetric, anti selfadjoint, essentially selfadjoint, unitary, normal, an orthogonal projector,  if and only if $A$, respectively,  is symmetric, selfadjoint, anti symmetric, anti selfadjoint, essentially selfadjoint, unitary, normal, an orthogonal projector.\\
{\bf (6)}  If $A$ is self-adjoint and $P^{(A)}$ is the associated PVM, the PVM $P^{(A_\bC)}$ of $A_\bC$ satisfies 
 $$P^{(A_\bC)}= (P^{(A)})_\bC\mbox{  so that, in particular  } P^{(A)}= P^{(A_\bC)}|_{\sH}$$ and, regarding the spectrum, 
$$\sigma(A_\bC)= \sigma(A), \:\: \mbox{more precisely}\:\: \sigma_p(A_\bC)= \sigma_p(A)\:,\: \:
\sigma_c(A_\bC)= \sigma_c(A)\:.$$
If $f:\bR\rightarrow\bR$ is measurable, referring to Prop.\ref{propint}, we have
$$
f(A_\bC)=\int_{\sigma(A_\bC)}f(\lambda)dP^{(A_\bC)}(\lambda)=\left(\int_{\sigma(A)}f(\lambda)dP^{(A)}(\lambda)\right)_\bC=(f(A))_\bC
$$
{\bf (7)} If  $A': D(A') \to \sH$  is another  operator in $\sH$ then $$A\subset A' \mbox{ iff } A_\bC\subset A'_\bC\quad \mbox{ and }\quad  (AA')_{\bC} = A_\bC A'_\bC\:.$$
{\bf (8)} If $p=p(x)$ is a real polynomial of finite degree, it holds
$$p(A_\bC) = (p(A))_\bC\:.$$
{\bf (9)} Let $D(A)$ be dense. $A$ is symmetric and positive iff $A_\bC$ is positive ((8) Def.\ref{defop})
\end{proposition}
\noindent The proof of this proposition is given in Appendix \ref{appProof}.

\subsection{Stone's theorem for real (and complex) Hilbert spaces}
We are in a position to state and prove a version of famous Stone's theorem valid for real (and complex)  Hilbert spaces, exploiting the constructed formalism.
The difficulty with the real Hilbert space case relies on the fact that the spectral decomposition cannot be directly exploited because the generator of the group is {\em anti selfadjoint} and these operators do not admit a spectral decomposition in {\em real} Hilbert spaces.
\begin{definition} {\em A {\bf one-parameter group of bounded operators} over  an either real or complex Hilbert space   $\sH$ is a map
 $U : \bR \to \gB(\sH)$, such that $U_0=I$ and $U_tU_s = U_{t+s}$ for $t,s \in \bR$. }
\end{definition}
\noindent We are now interested in the case where this map is strongly continuous (Def.\ref{defcontinuity}) with respect to the standard topology of $\bR$ and every $U_t$ is unitary
(Def.\ref{defop}(6)  and  Remark \ref{essentclos}(b)).
 To introduce the problem, we observe that  if $A$ is an anti selfdjoint operator in the complex Hilbert space, then  $\bR \ni t \mapsto e^{tA} = e^{-i t (iA)}$ (adopting the notation (\ref{fA}) for a 
function of a {\em selfadjoint} operator $iA$) is a  one-parameter group of {\em unitary operators} which is also {\em strongly continuous}. The proof is easy (e.g., see \cite{M}). If $\sH$ is real and $A$ is anti-selfadjoint, 
we can consider the strongly-continuous  one-parameter group of unitary operators $\bR \ni t \mapsto e^{tA_\bC}$ in $\sH_\bC$. 
Now, let $C$ the natural conjugation defined in (\ref{conj}),  then an easy application of complex Stone's theorem proves that $Ce^{tA_\bC}C=e^{tCA_\bC C}=e^{tA_\bC}$. Hence, thanks to (1) and (5) of Prop.\ref{prop2}, the map 
$\bR \ni t \mapsto \left. e^{tA_\bC}\right|_\sH$ turns out to be a strongly-continuous  one-parameter group of unitary operators in the {\em real} Hilbert space $\sH$.  The theorem we go to state, reversing the argument, focuses 
on the existence of an anti-selfadjoint generator $A$
for a given strongly-continuous  one-parameter group of unitary operators $\bR \ni t \mapsto U_t$ in $\sH$.  For the sake of completeness we will state the theorem into a way which is 
valid  for both real and complex Hilbert spaces.

\begin{theorem}[Stone's theorem]\label{ST}
Let $\sH$ be an either real or complex Hilbert space  and consider a strongly-continuous one-parameter group of unitary operators $U : \bR \to \gB(\sH)$. Define the subspace \beq D(A) := \left\{ x \in \sH \:\left|\: \exists y_x \in \sH\:, \:\: y_x := \lim_{h\to 0} h^{-1} (U_hx -x)\right. \right\} \label{DA}\eeq
and the operator
\beq A : D(A) \ni x \mapsto y_x \in \sH\:. \label{Gen}\eeq
It turns out that

 (i) $D(A)$ is dense in $\sH$,

 (ii) $AU_t=U_tA$, so that $U_t(D(A))= D(A)$, for all $t \in \bR$,

(iii) $A$ is anti-selfadjoint:  $A= -A^*$,

 (iv)  If $\sH$ is complex,  $A$ is the unique anti-selfadjoint operator satisfying (adopting the notation (\ref{fA})) 
$$U_t =  e^{tA}\:.$$

 (v) If $\sH$ is real, $A$ is the unique anti-selfadjoint operator satisfying
$$U_t =  e^{tA_\bC}\:|_{\sH}\:.$$

\end{theorem}

\begin{proof}
The statement for the complex case is a trivial re-adaptation of the standard statement of celebrated Stone's theorem 
(e.g., see \cite{M}), let us therefore pass to focus on the real Hilbert space case.
Consider the class of operators $V_t := (U_t)_\bC$. $\bR \ni t \mapsto V_t$ is  a  one-parameter group of unitary operators on $\sH_\bC$ as one immediately prove from (5) and (7) of Prop.\ref{prop2}. Strong continuity of  $V$ immediately arises from 
 $V_t := (U_t)_\bC$
and strong continuity of  $U$.
 Notice also that  $V_tC=CV_t$ holds from (1) of Prop.\ref{prop2} where $C$ is the natural conjugation of $\sH_\bC$. 
In view of the complex version of Stone's theorem,  we have  $V_t =e^{tB}$ for a unique anti-selfadjoint operator $B: D(B) \to \sH_\bC$, $D(B) \subset \sH_\bC$. $D(B)$ is dense,  $V_t(D(B)) = D(B)$, $BV_t = V_tB$, and  it holds
$$D(B) := \left\{ x+iy \in \sH_\bC \:\left|\: \exists  x'+iy' \in \sH_\bC\:, \:\:  x'+iy' := \lim_{h\to 0} h^{-1} (V_h(x+iy) -(x+iy))\right. \right\}$$
and 
$B : D(B) \ni x+iy \mapsto x'+iy'$.
From the definition of $D(B)$ and the fact that $V_tC=CV_t$ it immediately arises  that 
$C(D(B)) \subset D(B)$ and $CB\subset BC$. (1) of Prop.\ref{prop2} entails that $B= A_\bC$ for some anti-selfadjoint operator on $\sH$ whose domain is $D(A) = D(B) \cap \sH$ which, by construction coincides with
$$D(A)= D(B) \cap \sH = \left\{ x \in \sH \:\left|\: \exists y_x \in \sH\:, \:\: y_x := \lim_{h\to 0} h^{-1} (U_hx -x)\right. \right\}\:.$$
$D(A)$ is dense and $U_t(D(A))=D(A)$, from the analogous properties of $D(A_\bC)= D(B)$, and $U_tA=AU_t$ from the complex-case analogue making use of (7) in Prop.\ref{prop2}.\\
Let us finally come to the uniqueness issue. Suppose that there is an anti-selfadjoint operator $A': D(A') \to \sH$, in principle different from $A$, such that $U_t =  e^{tA'_\bC}\:|_{\sH}$. Consequently, $V_t = (e^{tA'_\bC}\:|_{\sH})_\bC = e^{tA'_\bC}$. The uniqueness part of Stone's theorem for complex Hilbert space  implies $A'_\bC = A_\bC$ so that  $A' =A'_\bC|_\sH =  A_\bC|_\sH = A$.
\end{proof}

\begin{definition}\label{defgen}
{\em Consider a  strongly-continuous one-parameter group of unitary operators $U : \bR \to \gB(\sH)$ with $\sH$ either complex or real Hilbert space. The {\em anti selfadjoint operator} $A : D(A) \to \sH$ associated to $U$ and  defined  by (\ref{DA})-(\ref{Gen}) is called the {\bf generator} of $U$. In both the real and complex Hilbert space case,  we  write $$U_t=e^{tA}\quad t \in \bR\:,$$} 
\end{definition}

\subsection{Schur's lemma for real (and complex) Hilbert spaces}
Another important issue is the formulation of the so-called {\em Schur's lemma} which has different statements for real and complex Hilbert spaces.

\begin{definition}
{\em Let $\sH$ be  an either real or complex Hilbert space.
A family  of operators  $\gU\subset \gB(\sH)$ is said to be {\bf  irreducible} if $U(\sK)\subset \sK$ for all $U\in \gU$ and  
 a closed subspace $\sK \subset \sH$ implies $\sK=\{0\}$ or $\sK =\sH$. 
$\gU$ is said to be {\bf reducible} if it is not irreducible.}
\end{definition}
\noindent Since the definition refers to {\em closed} subspaces,  our notion of irreducibility is sometimes called {\em topological} irreducibility.
\begin{remark} \label{remirred}$\null$\\
{\em
{\bf (a)} if $\gU$ is irreducible, then it is easy to see that $\{P\in\cL(\sH) \:|\:  [P,U]=0\ \forall U\in\gU\}=\{0,I\}$, while the opposite implication holds true if $\gU$ is closed under Hermitian conjugation.\\
{\bf (b)} If $\sH$ is a real Hilbert space, every family $\gU \subset \gB(\sH)$ induces an associated  family $\gU_\bC := \{U_\bC \:|\: U \in \gU\}  \subset \gB(\sH_\bC)$. If $\gU_\bC$ is irreducible, $\gU$ must be irreducible as well because $U(\sK)\subset \sK$ implies $U_\bC(\sK_\bC) \subset \sK_\bC$. The opposite implication is not true in general.
}
\end{remark}
\noindent  We have a first result which is valid for both the real and the complex Hilbert space case.

\begin{proposition} [Schur's lemma  for essentially selfadjoint operators] \label{SL}
Let $\sH$  be an, either real or complex, Hilbert space and let $\gU \subset \gB(\sH)$ be irreducible.\\ If the operator $A : D(A) \to \sH$, with $D(A) \subset \sH$ dense, is essentially selfadjoint and 
\beq UA \subset  AU\quad \mbox{for all}\quad U \in\gU\label{commAU}\:,\eeq
then  $\overline{A} \in \gB(\sH)$ (the bar denoting the closure of $A$) and 
 $$\overline{A}= a I\:,\quad \mbox{for some $a \in \bR$.}$$
If $A$ satisfying (\ref{commAU}) is selfadjoint, we   have $A\in \gB(\sH)$ with  $A = a I$
for some $a \in \bR$.
\end{proposition}

\begin{proof} We prove the thesis for the real Hilbert space case, the complex Hilbert space case has an analogous proof with obvious changes.
Since the operators $U \in \gU$ are bounded, from Remark \ref{remarkclosure}, one  has $U \overline{A}\subset  \overline{A}U$. Theorem \ref{st} (b) (ii) now implies that  the spectral measure of $P^{(\overline{A})}$, of $\overline{A}$,  commutes with every $U\in \gU$.
Since $\gU$ is irreducible,  if $E \in \cB(\bR)$, then either  $P^{(\overline{A})}(E)=0$ (i.e. $P^{(\overline{A})}$ projects onto $\{0\}$) or $P^{(\overline{A})}(E)=I$
 (i.e. $P^{(\overline{A})}$ projects onto the whole $\sH$). 
If  $P^{(\overline{A})}( (a_0,b_0])=0$ for all  $a_0\leq b_0$ in $\bR$, we would have $P^{(\overline{A})}(\bR)=0$,
due to the $\sigma$-additivity, which is not possible.  Thus  $P^{(\overline{A})}( (a_0,b_0])=I$ for some  $a_0\leq b_0$ in $\bR$. Notice that
$P^{(\overline{A})}( \bR \setminus (a_0,b_0])=0$ as trivial consequence of the properties of a PVM. 
Now, define $\delta_0:=b_0-a_0$ and divide $(a_0,b_0]$ into the disjoint union of two equal-length contiguous half-open intervals. Reasoning as above we see that one and only one of them has vanishing measure, while the other satisfies $P^{(\overline{A})}((a_1,b_1])=I$. Clearly $\delta_1:=b_1-a_1=\frac{1}{2}\delta_0$. Iterating this procedure we find a couple  of sequences $a_0\le a_1\le \cdots \le a_n, b_n\le \cdots \le b_1\le b_0$ within $[a_0,b_0]$ such that $\delta_n:=b_n-a_n=2^{-n}\delta_0$ and $P^{(\overline{A})}((a_n,b_n])=I$.
Being $[a_0,b_0]$ compact and $\delta_n\rightarrow 0$ it easily follows that there must exist $\lambda_0\in [a_0,b_0]$ such that
 $a_n \to \lambda_0$ and 
$b_n \to \lambda_0$. From outer continuity  of the positive measure $( x| 
P^{(\overline{A})}(E)| x)$ we have
$( x |P^{\overline{A}}(\{\lambda_0\})| x)  = ( x |P^{(\overline{A})}(\cap_n (a_n,b_n]) x) = ( x| I x)$ for every $x\in \sH$.
Since $P^{(\overline{A})}(\{\lambda_0\})-I$ is selfadjoint,
 $2( x |(P^{(\overline{A})}(\{\lambda_0\}) -I) y) =( x+y |(P^{(\overline{A})}(\{\lambda_0\}) -I) (x+y))
-( x |(P^{(\overline{A})}(\{\lambda_0\}) -I) x) -( y |(P^{(\overline{A})}(\{\lambda_0\}) -I) y) =0$ so that 
$( x |(P^{(\overline{A})}(\{\lambda_0\})-I) y) = 0$.
Since $x,y \in \sH$ are arbitrary, we have obtained that 
$P^{(\overline{A})}(\{\lambda_0\})=I$ and therefore $P^{(\overline{A})}(\bR \setminus \{\lambda_0\})=0$.
Computing the spectral integral of $\overline{A}$, defining $a:= \lambda_0$,  this result immediately implies that
$\overline{A}= \int_{\bR} \lambda P^{(\overline{A})}(\lambda) = a I$.
If the initial $A$ is already selfadjoint, it is essentially selfadjoint, too and the proof applies to $\overline{A}$. However as $A^*$ is closed and $A=A^*$, we have $A= \overline{A}$ proving the last statement. 
\end{proof}

\noindent A better result  can be obtained when the class $\gU$ consists of  {\em an irreducible unitary  representation}

\begin{definition}\label{urep}
{\em Let $\sH$ be an, either real or complex, Hilbert space and $G$ a group with unit element $e$ and group multiplication 
$G\times G \ni (g,g') \mapsto gg' \in G$.  \\ A {\bf  unitary  representation} of $G$ over $\sH$
is a map $G \ni g \mapsto U_g \in \gB(\sH)$ where $U_g$ is unitary, $U_e=I$ and $U_gU_{g'}= U_{gg'}$ for every $g,g'\in G$.\\
The unitary representation is said to be {\bf irreducible} if $\gU := \{U_g\:|\: g \in G\}$ is irreducible. }
\end{definition}
 \noindent In this juncture, the difference from the real and complex Hilbert space cases is evident and concerns  a mathematical notion which will play a fundamental r\^ole in our work.
\begin{definition}\label{defCSJ} {\em If $\sH$ is an either  real or complex  Hilbert space, an operator $J \in \gB(\sH)$ such that $J^2=-I$ and $J^*=-J$ is called {\bf complex structure} on $\sH$. }
\end{definition}
\noindent  In complex Hilbert spaces, obviously, $\pm iI$ are  the most natural complex structures.

\begin{proposition}\label{SL2}
Let $\sH$ be an, either real or complex, Hilbert space and $G \ni g\mapsto U_g$ a  unitary representation on $\sH$ of the group $G$ and consider a densely-defined operator in $\sH$ $A:D(A)\rightarrow \sH$ such that
\begin{equation}\label{SL2'}
U_gA\subset AU_g\:, \quad \forall g\in G\:.
\end{equation}
Then (\ref{SL2'}) holds in the stronger form $$U_gA= AU_g\quad \mbox{and}\quad U_gA^*= A^*U_g\:, \quad \forall g\in G\:.$$ 
If $G \ni g\mapsto U_g$ is irreducible and $A$ is closed,  the following further facts hold

(i) if $\sH$ is real, then $A=aI+bJ$, with $a,b\in\bR$ and $J$ is a complex structure,

(ii) if $\sH$ is complex, then $A=cI$, where $c\in\bC$.

\noindent In particular  $D(A)=\sH$ and $A\in\gB(\sH)$ in both the cases.
\end{proposition}
\begin{proof}
We use conventions and properties of standard domains (Def.\ref{defdomain} and Remark \ref{remassociativity}) also in relation with the adjoint conjugation (Remark \ref{sumselfadj}).
From (\ref{SL2'}), applying $U_{g^{-1}}$ to both sides, we also have $AU_{g^{-1}}\subset U_{g^{-1}}A$. Since the inclusion holds for every $g\in G$, it must also be $AU_{g}\subset U_{g}A$ which together with (\ref{SL2'}), yields $AU_{g} = U_{g}A$ for every $g\in G$.  Taking the adjoint of both sides of  $AU_{g} = U_{g}A$ and observing that $U_g\in \gB(\sH)$, (c) and (d) of Remark \ref{sumselfadj} gives $U_g^*A^*\subset A^*U_g^*$. Now, since $U^*_g= U_{g^{-1}}$ and $g\in G$ is generic, this is equivalent to $U_gA^*\subset A^*U_g$ for every $g$. Reasoning as above we get also the opposite inclusion and the proof of the first statement is over. Let us now assume that $A$ is closed and the representation is irreducible.
From $AU_g= U_gA$, $A^*U_g= U_gA^*$ and taking eventually  Remark \ref{remassociativity} into account, we have
$A^*AU_g = A^*U_gA= U_gA^*A$.
Since  $A$ is closed, the operator $A^*A$ turns out to be densely defined and selfadjoint. This is a well known result if $\sH$ is complex \cite{R,S,M}. The validity of the same statement for the real case will be given in the proof of Thm \ref{Polar} (a) below.
 Using Proposition \ref{SL} for the selfadjoint operator 
$A^*A$ we find $A^*A= aI$ for some real $a$. In particular $D(A^*A)=D(aI)=\sH$ so that $D(A)=\sH$  and thus, since $A$ is closed, the closed graph theorem (Thm \ref{cgT}) gives $A\in \gB(\sH)$. \\
We can decompose the operator $A$ into $A=\frac{A+A^*}{2}+\frac{A^*-A}{2}$, where the two addends are, respectively, selfadjoint and anti selfadjoint. Let us denote them, respectively, by $A_S$ and $A_A$. Both of them commute with the representation $U$, in particular $U_g A_S=A_S U_g$ for any $g\in G$ gives $A_S=aI$ for some $a\in\bR$, thanks to Prop. \ref{SL}. 
Now, suppose that $\sH$ is complex, then the operator $iA_A$ is selfadjoint and commutes with the representation $U$. So, thanks again to Prop \ref{SL} we find $iA_A=cI$ for some $c\in \bR$, i.e. $A_A=-ci$ and the proof is complete. Now, suppose that $\sH$ is real. The operator $A_A^2$ is selfadjoint and commutes with the group representation, hence $A_A^2=cI$ for some $c\in\bR$, thanks again to Proposition \ref{SL}. Notice that $c\le 0$, indeed if we take a unit vector $v\in\sH$, it holds $c=(v|cv)=(v|A_AA_Av)=-(A_Av|A_Av)=-\|A_Av\|^2\le 0$. We also see that $c=0$ if and only if $A_A=0$, that is if $A$ is selfadjoint: in this case the theorem is proved. Suppose that $c\neq 0$ and define $J:=\frac{A_A}{\sqrt{-c}}$. With this definition we find  $J\in \gB(\sH)$, $J^*=-J$ and $J^*J=-I$, i.e., $J$ is a complex structure as wanted, and
$A= aI + bJ$ for $a,b \in \bR$.
\end{proof} 
\begin{remark} {\em The result in (i) can be made stronger with the help of Theorem \ref{threecommutant} we shall prove later, observing that if $A \in \gB(\sH)$ commutes with the unitary representation $U$ it also commute with the von Neumann algebra generated by $U$.}
\end{remark}

\subsection{Square root and polar decomposition in real (and complex) Hilbert spaces}
Another technical tool, which will be very useful in this work, is the {\em polar decomposition theorem} demonstrated in a version which is valid for a real Hilbert 
space, too.

\begin{theorem}\label{Polar}
Let $\sH$ be an, either real or complex, Hilbert space and  $A:D(A)\rightarrow \sH$ a densely-defined closed operator in $\sH$. Then  the following facts hold.\\
{\bf (a)} $A^*A$ is densely defined, positive and selfadjoint.\\
{\bf (b)} There exists a unique pair  of operators $U,P$ in $\sH$ such that,

(i) $A=UP$ where in particular $D(P)=D(A)$ 

(ii) $P$ is selfadjoint and $P\geq 0$ ((8) Def.\ref{defop}),

(iii) $U\in\gB(\sH)$ is isometric on $Ran(P)$  (and thus on $\overline{Ran(P)}$ by continuity),

(iv) $Ker(U)\supset Ker(P)$.\\
The right-hand side of (i)  is called the {\bf polar decomposition} of $A$.  It turns out that, in particular,

(v) $P=|A|:=\sqrt{A^*A}$,

(vi) $Ker(U)= Ker(A) = Ker(P)$,

(vii)  $Ran(U) = \overline{Ran(U)}$,\\
are also valid where $\sqrt{A^*A}$ is interpreted as in Prop.\ref{sqrt} \\
{\bf (c)} If $\sH$ is real, the polar decomposition of $A_\bC$ is $A_\bC=U_\bC P_\bC$ where $A=UP$ is the polar decomposition of $A$. In particular, $|A_\bC|=|A|_\bC$. 
\end{theorem}
\noindent The proof of the theorem is given in Appendix \ref{appProof}.\\

\noindent
$U$  is a  {\em  partial isometry} (Def.\ref{defop}) because $U \in \gB(\sH)$ and  is isometric
on $Ker(U)^\perp$ (it is indeed isometric on $\overline{Ran(P)} = Ker(P^*)^\perp = Ker(P)^\perp =  Ker(U)^\perp$).\ 
 We conclude this section with a pair of  technical proposition, the second concerning the interplay of commutativity of  one-parameter unitary groups and commutativity of elements of the corresponding polar decomposition of the generators. That result will turn out very useful later.
\begin{proposition}\label{LemmaCOMM}
Let $\sH$ be an, either real or complex, Hilbert space. Consider an, either  selfadjoint or anti selfadjoint,  operator $A:D(A)\rightarrow \sH$ with polar decomposition  $A=UP$. The following facts hold.\\ 
{\bf (a)} If $A^*=-A$ and $B\in\gB(\sH)$, $Be^{tA}=e^{tA}B$ is valid if and only if $BA\subset AB$ holds.\\
{\bf (b)} If $B\in\gB(\sH)$ satisfies  $BA\subset AB$, then $BU=UB$ and $BP \subset PB$.\\
{\bf (c)} The  commutation relations are true
 $$UA \subset  AU \quad \mbox{and}\quad U^*A\subset AU^*\:.$$
Moreover, for every measurable function $f: [0,+\infty) \to \bR$:
 $$Uf(P)\subset f(P)U\quad \mbox{and}\quad U^*f(P)\subset f(P)U^*\:.$$
{\bf (d)} $U$ is respectively selfadjoint or anti selfadjoint.\\
{\bf (e)} If $A$ is injective (equivalently if either $P$ or $U$ is injective), then $U$ and $U^*$ are unitary. In this case all the inclusions in (c) are identities.
\end{proposition}

\noindent The proof of this proposition is given in Appendix \ref{appProof}.

\begin{proposition}\label{polarCOMM}
Let $\sH$ be an, either real or complex, Hilbert space and  $A$ and $B$ anti-selfadjoint
operators in $\sH$ with  polar decompositions $A=U|A|$ and $B=V|B|$.\\ 
If the strongly-continuous one-parameter groups generated by $A$ and $B$ commute, i.e.,
$$e^{tA}e^{sB}=e^{sB}e^{tA} \quad \mbox{for every $s,t\in\bR$,}$$ then the following facts hold\\

(i) $UB\subset BU$ and $U^*B\subset BU^*$;\\

(ii) $Uf(|B|)\subset f(|B|)U$ and $U^*f(|B|)\subset f(|B|)U^*$ for every measurable function 

\quad $f: [0,+\infty) \to \bR$;\\

(iii) $UV=VU$ and $U^*V=VU^*$.\\
\noindent  If  any of $A$, $|A|$,  $U$ is injective, then  the inclusions in (i) and (ii) can be replaced by identities. 
\end{proposition}

\noindent The proof of this proposition is given in Appendix \ref{appProof}.

\subsection{Internal  complexificated  structures}\label{InternalComp} 
If $\sH$ is a real Hilbert space, it is possible to define an  associate  complex Hilbert space  by means of an {\em internal complexification procedure} depending on a {\em  complex structure} $J$ as defined in
Def.\ref{defCSJ}.
The   vectors $x\in \sH $ can be viewed as 
elements of a {\em complex} vector space equipped with a {\em Hermitian} 
scalar product both constructed out of $J$.  The complex linear space structure is  \beq (a + ib) x:= a x + b Jx \quad \mbox{ $\forall x \in \sH$ and $\forall a,b  \in \mathbb R$}\:.\label{lsi}\eeq The Hermitian scalar product is, by definition, 
\beq (x| y)_J := ( x|y ) -i (x|Jy) \:,\quad \forall x,y \in \sH\label{spi}\eeq
so that, in particular, \beq \mbox{$x \perp_{\sH_J} y\:\:$  if and only if\:\:  both $x \perp_{\sH} y$ and $x \perp_{\sH} Jy$.}\eeq
and
\beq (x|y) = Re (x|y)_J \:,\quad \forall x,y \in \sH\label{spi2}\eeq
 which immediately implies that the norm generated by that Hermitian scalar product  satisfies 
\beq || x||_J = ||x|| \:,\quad \forall x \in \sH \label{csi}\eeq
The elementary but crucial result comes now.
\begin{proposition}\label{Pintcomp}
Let $\sH$ be a real Hilbert space whose scalar product is denoted by $(\cdot|\cdot)$ and equipped with a complex structure $J$. The complex vectors space over $\sH$
with the complex linear structure (\ref{lsi})
and the Hermitian scalar product (\ref{spi}) has the following properties.\\
{\bf (a)} It is a complex Hilbert space denoted by $\sH_J$.\\
{\bf (b)} $N \subset \sH_J$ is a Hilbert basis of $\sH_J$ if and only if 
$\{u, Ju\:|\: u \in N\}$ is a Hilbert basis of $\sH$. Thus $\sH_J$ is separable if and only if $\sH$ is.
\end{proposition}

\begin{proof}
Due to (\ref{csi}), Cauchy sequences in $\sH$ define pairs of Cauchy sequences  in $\sH_J$ and {\em vice versa}.
 For this reason $\sH_J$ is complete in view of the completeness of   $\sH$.  The second statement is true because $N$ is orthonormal maximal in $\sH_J$ iff $\{u,Ju\:|\: u \in N\}$ is orthonormal maximal in $\sH$ as it can be proved immediately.
\end{proof}

\begin{remark}\label{remJ}$\null$\\
{\em {\bf (a)} Due to (\ref{csi}), the identity map $I: \sH \ni x \mapsto x \in \sH_J$ is evidently an isometry of metric spaces. In particular $\sH$ and $\sH_J$ are homeomorphic.\\
{\bf (b)} From the definition of $\sH_J$ it easily arises that a $\bR$-linear subspace $K\subset \sH$ is  a $\bC$-linear subspace of $\sH_J$ if and only if $J(K)\subset K$ (more precisely, $J(K)=K$ as $JJ=-I$). Moreover $\overline{K}^{\sH}=\overline{K}^{\sH_J}$, hence $K$ is closed or dense in $\sH$ iff it is, respectively, closed or dense in $\sH_J$.
Conversely every $\bC$-linear subspace of $\sH_J$ is trivially a $\bR$-linear subspace of $\sH$.}
\end{remark}
\noindent A $\bC$-linear operator $A: D(A) \to \sH_J$, where $D(A)\subset \sH_J$ is a complex linear subspace, is also a $\bR$-linear operator on $\sH$. The converse is generally false. The following proposition concerns that issue.

\begin{proposition}\label{prop2i}
Let $\sH$ be a real Hilbert space with complex structure $J \in \gB(\sH)$.\\
 A $\bR$-linear operator
$A: D(A) \to \sH$ is a $\bC$-linear operator in $\sH_J$ if and only if $AJ=JA$. In that case $D(A)$ is a complex subspace of $\sH_J$ as well.
\end{proposition}

\begin{proof} If $A$ is $\bC$-linear,  its domain must be $\bC$-linear  and, in view of the complex linear 
structure of $\sH_J$, $AJ=JA$. Conversely, if $A$ is $\bR$-linear and 
$AJ=JA$, it must be $J(D(A))\subset D(A)$ and thus $D(A)$ is complex linear subspace ((b) in Remark \ref{remJ}). Moreover $A i x = AJx = JAx = iAx$ for every $x\in D(A)$ so that $A$
is $\bC$-linear.
\end{proof}

\begin{remark}$\null$\\
{\em {\bf (a)} A complex Hilbert space $\sH$ (with scalar product denoted by $\langle \:\:|\:\:\rangle$) can always be written as $\sK_J$ for a real Hilbert space $\sK$. As a set $\sK = \sH$ equipped with
 the $\bR$-linear structure restriction of the $\bC$-linear one of $\sH$. The real scalar product on $\sK$
is  $(x|y) := Re\langle x| y\rangle$, and  the complex structure over $\sK$ is the  $\bR$-linear  operator $J: \sK \ni x \mapsto i x \in \sK$.\\
{\bf (b)} If a real Hilbert space is finite-dimensional and its dimension is odd, there is no complex structure in $\gB(\sH)$, 
otherwise 
we would obtain a contradiction from (b) in Prop.\ref{Pintcomp}, and  no internal complexification procedure is possible. The reader may easily prove that this is the only obstruction: if the 
dimension of $\sH$ is infinite or finite and even,  a complex structure always exists associated with every  given Hilbert  basis of $\sH$.\\
{\bf (c)} If $\gU \subset \gB(\sH)$, where  $\sH$ is a real Hilbert space, and the elements of $\gU$
commute with a complex structure $J\in \gB(\sH)$, then $\gU_J := \gU$ is also a family of $\bC$-linear operators in $\sH_J$. 
$\gU_J$ is  irreducible if $\gU$ is irreducible since  complex closed subspaces are real closed subspaces.}
\end{remark}

\begin{proposition}\label{prop3i}
Let $\sH$ be a real Hilbert space with complex structure $J \in \gB(\sH)$ and  suppose that the
the  operator $A: D(A) \to\sH$ satisfies $AJ=JA$. The following facts hold.\\
{\bf (a)}  If $D(A)$ is dense, the adjoint $A^*$ of $A$ with respect to $\sH$ coincides with the adjoint operator referred to $\sH_J$.\\
{\bf (b)}   $A$ is closable referring to $\sH$ if and only if it is closable 
referring to $\sH_J$. In this case the two closures coincide.\\
{\bf (c)} Let $A$ be closable and $S \subset D(A)$ s.t. $J(S)\subset S$. Then $S$ is a core for $A$ referring to $\sH_J$, iff it is a core for $A$ referring to $\sH$.\\
{\bf (d)}  $A$ is symmetric, selfadjoint, anti symmetric, anti selfadjoint, essentially selfadjoint, unitary, normal, an orthogonal projector   referring to $\sH$  if and only if $A$   is respectively symmetric, selfadjoint, anti symmetric, anti selfadjoint, essentially 
selfadjoint, unitary, normal, an orthogonal projector
referring to $\sH_J$.\\
{\bf (e)}  If  $A$ is selfadjoint and $P^{(A)}$ is the associated PVM in $\sH$, $P^{(A)}$ is also the PVM of $A$ in $\sH_J$.
Moreover the (point and continuous) spectrum of $A$ referring to $\sH$ coincides to the
 (resp. point and continuous) spectrum of $A$ referring to $\sH_J$.
\end{proposition}

\noindent The proof of this proposition appears in Appendix \ref{appProof}.

\subsection{Elementary fact on  von Neumann algebras in real (and complex) Hilbert spaces}
If $\gM\subset \gB(\sH)$ is a subset in the algebra of bounded operators on the, either real or complex, Hilbert space $\gB(\sH)$, the {\bf commutant} of $\gM$ is:
$$\gM' := \{T \in \gB(\sH)\:\: | \: \: TA-AT =0\quad \mbox{for any $A
\in \gM$}\}\:.$$
If $\gM$\index{$\gM'$} is closed under the adjoint conjugation, then  the 
commutant $\gM'$ is  a  $^*$-algebra with unit. In general: $\gM_1' \subset \gM_2'$
if $\gM_2 \subset \gM_1$ and $\gM\subset (\gM')'$, which imply $\gM' = ((\gM')')'$. 
Hence we cannot reach beyond the second commutant by iteration.
The continuity of the product of operators says that the  commutant $\gM'$ is closed in the 
uniform topology, so if $\gM$ is closed under the adjoint conjugation, its commutant 
$\gM'$ is a $C^*$-algebra ($C^*$-subalgebra) in $\gB(\sH)$.  It is easy to prove $\gM'$ is both strongly and weakly closed. 
The next crucial result due to von Neumann is valid both for the real and complex Hilbert space case.

\begin{theorem} \label{teoDC}
If  $\sH$  is an, either real or complex, Hilbert space and $\gA$ a unital  $^*$-subalgebra of $\gB(\sH)$, the following statements are equivalent.\\
{\bf (a)}  $\gA = \gA''$.\\
{\bf (b)}  $\gA$ is weakly closed.\\
{\bf (c)}  $\gA$ is strongly closed.\\
More precisely, if  $\gB$ is a unital  $^*$-subalgebra of $\gB(\sH)$, then
 $\gB'' = \overline{\gB}^w =  \overline{\gB}^s$\:, the bar denoting the closure with respect to either the weak ($\overline{\cdot}^w$) or strong ($\overline{\cdot}^s$) topology.
\end{theorem}
\begin{proof} 
The  proof of the last statement can be found in every  book on operator algebras, e.g., Theorem 5.3.1 in  \cite{KR} Vol I (see also   \cite {Li, M}) and it does not depend on the field of $\sH$,  either $\bR$ or $\bC$.
(a) implies (b) because $\gA=\gA''$ is the commutant  of $\gA'$ and the commutant of a set is evidently weakly closed.
(b) implies (c) because the strong convergence implies the weak convergence.
(c) implies (a) because  $\overline{\gA}^s=\gA$ for (c),
 $\gA \subset \gA''$ by definition of commutant, and  $\gA''= \overline{\gA}^s$ in view of  the last statement of the theorem. 
\end{proof}

\begin{definition}\label{defAvN} {\em A {\bf von Neumann algebra}  in $\gB(\sH)$ is a unital $^*$-subalgebra of $\gB(\sH)$ that satisfies the three   equivalent properties (a),(b),(c) appearing in Theorem \ref{teoDC}.\\
The {\bf center} of $\gR$ is the Abelian von Neumann algebra  $\gZ_\gR := \gR \cap \gR'$.\\
A von Neumann algebra $\gR$ is a {\bf factor} when  $\gZ_\gR =\{c I\}_{c\in \bK}$ with $\bK=\bR,\bC$.} 
\end{definition}

\noindent A von Neumann algebra $\gA$ in $\gB(\sH)$ is evidently   a $C^*$-(sub)algebra with unit of $\gB(\sH)$. 
Moreover  $\gM'$ is a von Neumann algebra provided $\gM$ is a  $^*$-closed subset of $\gB(\sH)$, 
because $(\gM')'' = \gM'$ as we saw above.  If $\gM \subset \gB(\sH)$ is closed under the Hermitian conjugation,  $\gM''$ turns out to be the smallest (set-theoretically) von Neumann algebra containing $\gM$ as a 
subset. Indeed, if $\gA$ is a von Neumann algebra and $\gM \subset \gA$, then $\gM' \supset \gA'$ and $\gM'' \subset \gA''=\gA$. This fact leads to the following definition.

\begin{definition}\label{defvnagen}
 {\em Let $\sH$ be an, either real or  complex Hilbert, space and  $\gM \subset \gB(\sH)$ a set  closed under 
the Hermitian conjugation. The von Neumann algebra   $\gM''$ is called the {\bf von Neumann algebra generated} by $\gM$. }
\end{definition}

\noindent Differences between the real and complex case arise when one study the interplay of a von Neumann algebra and its lattice 
of orthogonal projectors as is already evident form (d) and (e) of the following elementary result whose proof appears in Appendix \ref{appProof}.

\begin{theorem}\label{teopropvnA} Let  $\gR$ be a von Neumann algebra over the  either real or  complex Hilbert space $\sH$, define $\cJ_\gR := \{J \in \gR \:|\: J^*=-J \:,\: -J^2\in   \cL_\gR(\sH)\}$ and let $\cL_\gR(\sH)$ denote  the set of orthogonal projectors in $\gR$. The following facts hold.\\
{\bf (a)} $A^* = A \in \gR$ if and only if the orthogonal projectors of the PVM of $A$
belong to $\gR$.\\
{\bf (b)} $\cL_\gR(\sH)$ is a complete (in particular $\sigma$-complete) orthomodular lattice which is sublattice of $\cL(\sH)$.\\
{\bf (c)} $\gR$ is irreducible if and only if $\cL_{\gR'}(\sH) = \{0,I\}$.\\
{\bf (d)} If $\sH$ is a complex Hilbert space, then $\cL_\gR(\sH)'' = \gR$.\\
{\bf (e)} If $\sH$ is a real Hilbert space,

(i) $\cL_\gR(\sH)''$  contains all selfadjoint operators in $\gR$,

(ii) $(\cL_\gR(\sH) \cup  \cJ_\gR)'' = \gR$,

(iii) $\cL_\gR(\sH)'' \subsetneq \gR$ if and only if  there is  
$J\in \cJ_\gR\setminus \cL_\gR(\sH)''$.
\end{theorem}
\noindent{\bf Example}.  We  show an elementary  example where  $\cL_\gR(\sH)'' \subsetneq \gR$ is valid.
Let $\sH_0$ be either an infinite-dimensional real Hilbert space or a finite-dimensional 
one with even dimension, so that $\sH_0$ admits a complex structure $J_0$.  Next define 
$\sH := \sH_0 \oplus \bR$ and $J = J_0\oplus 0$ and let $P: \sH \to \sH$ denote 
the orthogonal projector onto $\sH_0$. Obviously $P=-J^2$.
Consider the unital $^*$-algebra  $\gR :=   \{aI+bJ+cJ^2\}_{a,b,c \in \bR}\subset \gB(\sH)$  (notice that $J^3= -J$). It is easy to prove  that $\gR$ is weakly closed and thus it is an algebra of von Neumann.
However $\cL_\gR(\sH)= \{0,I, P, I-P\}$, so that $\cL_\gR(\sH)'' = \{aI+ bP\}_{a, b \in \bR} = \{aI+ cJ^2\}_{a, c \in \bR}$ which is strictly included in $\gR$ since 
it does not contains $J$ itself. As stated in the theorem above, however,   $J^*=-J$ and $-J^2=P\in \cL_\gR(\sH)$ so that $J \in \cJ_\gR$ and $(\cL_\gR(\sH) \cup \cJ_\gR)'' =  \{aI +bP + cJ\}_{a,b,c \in \bR} = \gR$.

\begin{remark}\label{unitgeneralg}{\em Another difference between the two cases regards the group of unitary operators of $\gR$, denoted by $\gU_\gR$. Indeed, as proved in Prop.4.3.5 \cite{Li} and in the subsequent Remark, while in the complex case the linear span  $[\gU_\gR]$  of $\gU_\gR$ equals $\gR$, in the real case we have to take its norm closure: $\overline{[\gU_\gR]}^{\null_{\gB(\sH)}}=\gR$.
}
\end{remark}
\section{Unitary Lie-Group representations}\label{represent}
This last technical section is devoted to introduce the main machinery we will exploit to describe  the continuous quantum symmetries of a quantum system and the observables associated with these symmetries. 
 Decisive tools are  the notions of strongly-continuous unitary representations of Lie groups, anti selfadjoint generators and their universal enveloping algebra. 
We assume that the reader is familiar with the basic theory of Lie groups (e.g., see  \cite{NaimarkStern, Warner,V}).

\subsection{Unitary representations  and Lie algebra of  generators in real (and complex) Hilbert spaces}
\begin{definition}{\em  If $G$ is a topological group, a {\bf strongly-continuous unitary representation} of $G$ over the, either real or complex, Hilbert space $\sH$ is a unitary representation  $G \ni g \mapsto U_g \in \gB(\sH)$  (Def.\ref{urep}) which is strongly continuous (Def.\ref{defcontinuity}).}
\end{definition}
\begin{remark} {\em In the rest of the paper we only consider the case of a {\em finite-dimensional real Lie group} $G$ whose Lie algebra is denoted by $\gg$. The adjectives {\em finite-dimensional} and {\em real} 
will be omitted almost always. In the rest of this work $C_0^\infty(G)$ refers to {\em real}-valued functions.} \end{remark}

\begin{definition}\label{defasagen}{\em
Let $G$ be a Lie group with Lie algebra $\gg$ and consider a
 strongly continuous unitary representation $G \ni g \mapsto U_g$ over the, either real or complex,  Hilbert space $\sH$.
If ${\bf A} \in \gg$ let  us indicate by $\bR \ni t \mapsto \exp(t{\bf A}) \in G$  the generated  one-parameter Lie subgroup.
The {\bf anti-selfadjoint generator associated to ${\bf A}$}, $A : D(A) \to \sH$  is the generator of the unitary group $\bR \ni t \mapsto U_{\exp\{t{\bf A}\}}$ in the sense of Def.\ref{defgen}. }
\end{definition}

\noindent To go on, we need some technical definitions. Let $G\ni g\mapsto x_g\in\sH$ a continuous map, $f\in C_0^\infty(G)$ and let $dg$ denotes the {\em left-invariant Haar measure}. Exploiting Riez' Lemma, $$\int_G f(g) x_g \:dg\: \mbox{denotes the unique vector } x_G\in\sH \mbox{ s.t. }\left(y|x_G\right)=\int_G f(g) (y|x_g)\:dg\ \ \forall y\in\sH\:.$$
It is easy to prove that $\|\int_G f(g) x_g\:dg\|\le \int_G|f(g)|\|x_g\|\:dg$ (see Lemma\ref{Rlemma}). This allows us to exploit the results in Ch.10 \cite{S0} where a different notion of integral is used.
\emptyline
Now we can introduce a very important notion: the so called {\em G\r{a}rding space}.

\begin{definition}\label{defgarding}
{\em Let $G$ be a Lie group and consider a
 strongly continuous unitary representation $U$ of $G$ over the,  either real or complex, Hilbert space $\sH$. If  $f \in C_0^\infty(G)$ and $x\in \sH$, define
\beq
x[f] := \int_G f(g) U_g x \:dg\:.
\eeq
The, respectively real or complex, finite span of all vectors $x[f] \in \sH$ with $f \in C_0^\infty(G)$ and $x\in \sH$ is called {\bf G\r{a}rding space} of
 the representation and is denoted by $D^{(U)}_G$.}
\end{definition}
\noindent 
From the definition it is easy to see that, if $x\in D_G^{(U)}$, then the function $g\mapsto U_gx$ is a smooth map if the differentiation is carried out in the topology of $\sH$ with respect to the Lie group structure of $G$.
Actually an, actually more general, remarkable result  due to Dixmier and Malliavin \cite{DM} shows that also the inverse result holds. We assume the validity of the 
theorem in the complex case and extend the proof to the case of a real Hilbert space.

\begin{theorem}\label{DMtheorem}
Let $\sH$ be either real or complex Hilbert space. Then the G\r{a}rding space coincides with the subspace of the {\bf smooth} vectors of the representation, that is  
the vectors  $x\in\sH$ such that the function $U: G \ni g\mapsto U_gx$ is $C^\infty$.\\
Suppose that $\sH$ is real, then the following facts hold:\\
{\bf (a)} $U_\bC : G \ni g \mapsto (U_g)_\bC$ is a unitary strongly-continuous representation of the Lie Group $G$ on the {\em complex} Hilbert space $\sH_\bC$ and $D_G^{(U_\bC)} = (D_G^{(U)})_\bC$.\\
{\bf (b)} If $J$ is a complex structure commuting with every $U_g$, then $G\ni g\mapsto U_g$ is a unitary strongly-continuous representation of $G$ on the complex Hilbert space $\sH_J$ and the definition of $D_G^{(U)}$ does not depend on the field of scalars. 
\end{theorem}

\begin{proof}
The proof of the first part for the complex case is part of the content of the original Dixmier-Malliavin paper \cite{DM} where the thesis is even proved for representations on Fr\'echet spaces. Suppose that  $\sH$ is real and consider the map $U_\bC : G \ni g \mapsto (U_g)_\bC$ defined by the complexification of $U$: this is clearly a unitary strongly-continuous representation thanks to Prop.\ref{prop2} and the definition of $\sH_\bC$. Let $x,y\in\sH$, then it easy to see that $g\mapsto U_gx$ and $g\mapsto U_gy$ are smooth iff $g\mapsto Ux+iU_y=U_\bC(x+iy)$ is smooth, proving this way also (a) of the second part. In particular, if $x\in D_G^{(U)}$ then $x+i0\in D_G^{(U_\bC)}$ hence, using the complex part of the theorem there must exist a {\em finite} number of  functions $f_k\in C^\infty_0(G)$,  corresponding scalars $a_k,b_k\in\bR$ and corresponding 
vectors $v_k,u_k\in\sH$ such that $x+i0=\sum_k(a_k+ib_k)(v_k+iu_k)[f_k]\in D_G^{(U_\bC)}$.
Using the definition of G\r{a}rding vectors, a direct calculation shows that $(x+iy)[f]=x[f]+iy[f]$ for any $x,y\in\sH$ and $f\in C_0^\infty(G)$, where the left-hand side is defined with respect 
to $\sH_\bC,U_\bC$ and the right-hand side with respect to $\sH,U$. A straightforward calculation gives $x+i0=\sum_k(a_kv_k-b_ku_k)[f_k]+i\sum_k(b_kv_k+a_ku_k)[f_k]$ which implies  $x=\sum_k(a_kv_k-b_ku_k)[f_k]\in D_G^{(U)}$. 
Now let us prove (b) of the second part. Prop.\ref{prop3i} and the definition of $H_J$ proves immediately that $g\mapsto U_g$ is a unitary strongly-continuous representation also on $H_J$. Again, since the notion of differentiability only looks at the norm and the $\bR$-linearity of $H$ a vector $x$ is smooth for $U$ on $\sH$ if and only if it is smooth for $U$ on $\sH_J$.
\end{proof}
\noindent  $D^{(U)}_G$ enjoys very remarkable properties we state in the next theorem.
In the following $L_g :  C_0^\infty(G) \to C_0^\infty(G)$ denotes the standard left-action of $g\in G$ on smooth compactly supported real-valued functions defined on $G$:
\beq
(L_gf)(h):= f(g^{-1}h) \quad \forall h \in G\:,
\eeq
and, if ${\bf A} \in \gg$, then  $X_{{\bf A}} : C_0^\infty(G) \to C_0^\infty(G)$ is the smooth vector field  over $G$ (a smooth differential operator)  defined as:
\beq
\left(X_{{\bf A}}(f)\right)(g) := \lim_{t\to 0} \frac{f\left(\exp\{-t{\bf A}\}g \right)-f(g)}{t}\quad \forall g \in G\:.
\eeq
so that the map
\beq
\gg \ni {\bf A} \mapsto X_{{\bf A}}
\eeq
defines a faithful Lie-algebra representation of $\gg$ in terms of vector fields on $C_0^\infty(G)$.\\
There is a natural way to see the Lie algebra $\gg$, where only the commutator is defined, as immersed into an associative algebra $E_\gg$, where a complete associative product giving rise to the commutator of $\gg$ exists.
$E_\gg$ is called  the {\em universal enveloping algebra}  of $\gg$  (Def.\ref{UEA}). This algebra is real unital  and admits a natural 
{\em real involution} $E_\gg \ni {\bf M} \mapsto {\bf M}^+$
(Def.\ref{defsim}). The physical relevance of $E_\gg$ is that its symmetric elements are related to the observables 
of a quantum physical system admitting $G$ as symmetry group.
An important technical r\^ole is played by  the {\bf Nelson elements} of  $E_\gg$ which are those of the form \beq{\bf N}:=\sum_{i=1}^n{\bf X}_i\circ{\bf X}_i\:,\label{NelsonE}\eeq
where $\{{\bf X}_1,\dots{\bf X}_n\}$ is a basis of $\gg$.\\
The next theorem states the basic properties of G\r{a}rding domain also in relation with a natural representation of $\gg$ 
constructed out of the generators of the representation $U$ of $G$
and   its universal enveloping algebra.

\begin{theorem}\label{teogarding}
Let $G$ be a Lie group with Lie algebra $\gg$ and Lie bracket $[\:,\:]_\gg$, consider a
 strongly continuous unitary representation $G \ni g \mapsto U_g$ over the, either real or complex,  Hilbert space $\sH$ and let indicate by $A$ the anti selfadjoint generator 
associated to ${\bf A} \in \gg$ as in Def.\ref{defasagen}.
The G\r{a}rding space $D^{(U)}_G$ satisfies the following properties.\\
{\bf (a)} $D^{(U)}_G$ is dense in $\sH$.\\
{\bf (b)} If $g\in G$, then  $U_g(D^{(U)}_G) \subset D^{(U)}_G$. More precisely, if $f \in C_0^\infty(G)$, $x \in \sH$, $g\in G$, it holds
\beq
U_gx[f] = x[L_gf]\:.\label{Ugf}
\eeq
{\bf (c)} If ${\bf A} \in \gg$, then $D^{(U)}_G\subset D(A)$  and furthermore $A(D^{(U)}_G) \subset D^{(U)}_G$. More precisely
\beq
A x[f] = x[X_{{\bf A}}(f)]\:.\label{Af}
\eeq
{\bf (d)} The linear  map 
\beq u: \gg \ni {\bf A} \mapsto  A|_{D^{(U)}_G}\eeq
is a Lie-algebra representation in terms of anti symmetric operators on $\sH$ 
((2) in Def.\ref{defop})
defined on the common dense invariant domain $D^{(U)}_G$ so that, in particular,
$$[u({\bf A}), u({\bf B})] = u([{\bf A}, {\bf B}]_\gg)\quad {\bf A}, {\bf B} \in \gg\:,$$
where $[\;, \:]$ is the standard commutator of operators.\\
{\bf (e)} The map $u$  uniquely extends to a real unital algebra representation of the universal enveloping algebra $E_\gg$: If ${\bf M} \in E_\gg$ is taken as in (\ref{dev}),
\begin{equation}
u({\bf M}):=c_0I|_{D_G^{(U)}}+\sum_{k=1}^N \sum_{j=1}^{N_k}c_{jk}u({\bf A}_{j1})\cdots u({\bf A}_{jk})
\end{equation}
It holds $u({\bf M}^+)\subset u({\bf M})^*$ (see Def.\ref{defsim}), in particular $u({\bf M})$ is a symmetric if ${\bf M}= {\bf M}^+$.\\
{\bf (f)}  $D^{(U)}_G$ is a core for every anti self adjoint  generator $A$ with ${\bf A} \in \gg$, that is
\beq
A = \overline{u({\bf A})}\:,\quad \forall {\bf A} \in \gg\:. 
\eeq
{\bf (g)} Suppose that ${\bf M}\in E_\gg$ satisfies both ${\bf M}={\bf M}^+$ and $[{\bf M},{\bf N}]_\gg=0$ for some Nelson element (\ref{NelsonE}) ${\bf N}\in E_\gg$, then $u({\bf M})$ is essentially selfadjoint.
In particular, $u({\bf N})$ is always essentially selfadjoint. 
\end{theorem}

\noindent The proof of this theorem appears in Appendix \ref{appProof}.

\begin{remark}\label{updown}$\null$\\
{\em 
{\bf (a)} Notice that, from $u({\bf M^+})\subset u({\bf M})^*$ it immediately follows that $u({\bf M})$ is closable, its adjoint being densely defined (see Remark \ref{remarkclosure} (d)).\\
{\bf (b)} Referring to the representations $U$ and $U_\bC$, respectively on the real Hilbert space $\sH$ 
and the associated complex one $\sH_\bC$
discussed at the end of Theorem  \ref{DMtheorem}, it is easy to prove that $u({\bf A})_\bC = u_\bC({\bf A})$ for every ${\bf A} \in \gg$, where $\gg \ni {\bf A} \mapsto u_\bC({\bf A})$
is the Lie-algebra representation of $U_\bC$.\\
{\bf (c)} Suppose $g\mapsto U_g$ is defined on the real Hilbert space $\sH$ and commutes with a complex structure $J$. Then it is immediate from their definition that the anti selfadjoint generators of $U_{\exp(t\V{A})}$, for $\V{A}\in\gg$, defined on $\sH$ and $\sH_J$, respectively, coincide. In particular the definition of Lie algebra representation $u$ is independent from the scalar field, thanks to Theorem \ref{DMtheorem}.
}
\end{remark}

\begin{proposition}\label{alggroupcomm}
Let $\sH$ be an either real or complex Hilbert space, $G$ a connected Lie group and $G \ni g \mapsto U_g$ a unitary strongly-continuous representation of $G$ over $\sH$.
If $B\in\gB(\sH)$ the following conditions are equivalent

(i)   $B u({\bf A})\subset u({\bf A})B$ for every ${\bf A}\in\gg$,

(ii)  $B\overline{u({\bf A})}\subset \overline{u({\bf A})}B$ for every ${\bf A}\in\gg$,

(iii)  $BU_g=U_gB$ for every $g\in G$. \\
\noindent If one of these conditions is satisfied, then $B(D_G^{(U)})\subset D_G^{(U)}$.
\end{proposition}

\noindent The proof of this theorem appears in Appendix \ref{appProof}.

\begin{corollary}\label{CORRalggroupcomm}
If $B$ in Proposition \ref{alggroupcomm}  is either  a complex structure or a unitary selfadjoint operator, then each of (i), (ii), (iii) is equivalent to

(iv) $B\overline{u({\bf A})}=\overline{u({\bf A})}B$ for every ${\bf A}\in\gg$.
 \end{corollary}
\begin{proof}
(iv) implies (ii). Conversely, if (ii) holds, applying $B$ to both sides of   $B\overline{u({\bf A})}\subset \overline{u({\bf A})}B$
we obtain $\overline{u({\bf A})}B\subset B\overline{u({\bf A})}$. Together with (ii) this inclusion proves (iv).
\end{proof}

\subsection{Analytic vectors of unitary representations in real (and complex) Hilbert spaces}  There exists another  subspace of $\sH$ relevant to continuous unitary representations of Lie groups, made of "good" vectors.  A function $f:\bR^n\supset U\rightarrow \sH$ is called {\bf real analytic} at $x_0\in U$ if there
 exists a neighborhood $V\subset U$ of $x_0$ where   the function $f$ can be expanded  in power series as (exploiting the standard multi index notation)
\begin{equation}\label{anal}
f(x)=\sum_{|\alpha| \leq n\:, n =0}^{+\infty}(x-x_0)^\alpha v_\alpha,\ \ x\in V
\end{equation}
with suitable $v_\alpha\in\sH$ for every multi index $\alpha\in\bN^n$.

\begin{definition}\label{defDN}{\em
Let $\sH$ be an, either real or complex, Hilbert space and $G \ni g \mapsto U_g$ a strongly-continuous unitary representation on $\sH$ of the Lie group $G$. A vector $x\in\sH$ is said 
to be \textbf{analytic}  for $U$ if the function $g\mapsto U_gx$ is real analytic at every point  $g \in G$,  referring to the analytic atlas of $G$. The linear subspace of $\sH$ made of by these 
vectors is called the \textbf{Nelson space} of the representation and is denoted by $D_N^{(U)}$.
}
\end{definition}
\begin{remark}\label{remanalyticrepr}{\em
Let $H$ be real. A direct application of the definition shows that $x,y$ are analytic for $U$ if and only if $x+iy$ is analytic for $U_{\bC}$, hence $D_N^{(U_\bC)}=(D_N^{(U)})_\bC$}
\end{remark}
\noindent  To go on, according Nelson \cite{N},  we say that a vector
 $x \in \bigcap_{n=0}^{+\infty} D(A^n)$ --  where $A: D(A) \to \sH$ is an operator in a either real or complex Hilbert space $\sH$ --
is {\bf analytic} for $A$, if there exists $t_{x}>0$ such that 
\beq \sum_{n=0}^{+\infty} \frac{t_x^n}{n!}||A^nx|| <+\infty \label{defanalyticvector}\eeq
From the elementary theory of series of powers, we know that  $t$ above can be replaced for every $z\in \bC$ with $|z|< t_x$  obtaining an absolutely convergent series.

\begin{remark}\label{remanalitic}
{\em It should be evident   that the analytic vectors for $A$ form a subspace of $D(A)$. Moreover, 
if $\sH$ is real,
from  (\ref{cs}) and the very definition of $A_\bC$, it immediately  arises that $x,y \in \sH$ are analytic for $A$ if and only if  $x+ iy$ is analytic for $A_\bC$.}
\end{remark}
\noindent One of remarkable Nelson's results states that
\begin{proposition}\label{propNseries} Consider an operator  $A : D(A) \to \sH$ on an either real or complex, Hilbert space $\sH$.\\
{\bf (a)}  If $A$ is anti selfadjoint  and $x \in D(A)$ is analytic with $t_x >0$ as in (\ref{defanalyticvector}), then
$$e^{tA}x = \sum_{n=0}^{+\infty} \frac{t^n}{n!}A^nx\quad \mbox{if $t\in \bR$ satisfies $|t|\leq t_x$.}$$
{\bf (b)} If $A$ is (anti) symmetric and $D(A)$ includes a set of analytic vectors whose finite span is dense in $\sH$, then $\overline{A}$ is (anti)  selfadjoint and $D(A)$,
\end{proposition} 

\begin{proof} If $\sH$ is complex, (a) and (b) are  classic result \cite{N,M} ((b) in the anti selfadjoint case arises from the selfadjoint case by simply 
using $iA$ in place of $A$). If $\sH$ is real and $x\in D(A)$ is analytic for $A$, then 
$x+i0$ is analytic for $A_\bC$. Taking advantage of Prop.\ref{prop2}, we have
$e^{tA}x= e^{tA_\bC}(x+i0) = \sum_{n=0}^{+\infty} \frac{t^n}{n!}A_\bC^n(x+i0) = \sum_{n=0}^{+\infty} \frac{t^n}{n!}A^nx$
 for $t\in \bR$ which satisfies $|t|\leq t_x$. This proves (a) for the real case. Regarding (b) for $\sH$ real, observe that $x+iy$ is analytic for $A_\bC$
when $x,y$ are analytic for $A$. Thus, with the hypotheses in (b),    $\overline{A_\bC}$ is (anti) 
selfadjoint.  Finally (5) in Prop. \ref{prop2} implies the thesis. 
\end{proof}
\begin{theorem}\label{teonelson}
Referring to Def.\ref{defDN}, Def.\ref{defasagen} and Thm \ref{teogarding},  $D_N^{(U)}$ satisfies the following properties.\\
{\bf (a)} $D_N^{(U)}\subset D_G^{(U)}$,\\
{\bf (b)} $U_g(D_N^{(U)})\subset D_N^{(U)}$ for any $g\in G$,\\
{\bf (c)} $D_N^{(U)}$ is dense in $\sH$,\\
{\bf (d)} $D_N^{(U)}$ consists of analytic vectors for every operator $u({\bf A})$ with ${\bf A}\in \gg$,\\
{\bf (e)} $u({\bf A})(D_N^{(U)})\subset D_N^{(U)}$ for any ${\bf A}\in\gg$.\\
{\bf (f)} Let $p : \bR \to \bR$ be a real polynomial such that either
$$ \mbox{$p(-x) = p(x)$ for every $x\in \bR$ or $p(-x) = -p(x)$ for every $x\in \bR$\:.}$$
If  ${\bf A} \in \gg$ then
 $\overline{u(p({\bf A}))}$ is, respectively,   selfadjoint or anti selfadjoint. 
\end{theorem}
\begin{proof}
Let $\sH$ be complex. The proof of (a) and (b) is immediate noticing that an analytic function is in particular smooth and that the multiplication on 
$G$ is analytic with respect to the analytic atlas of $G$.  Properties (c), (d) and (e) straightforwardly arise from  the results in \cite{N}, Sect.2 and Sect.7.  
Now consider $\sH$ real, from Remark \ref{remanalyticrepr} that $D^{(U_\bC)}_N=(D_N^{(U)})_\bC$. This equality, together with $D_G^{(U_\bC)}=(D_G^{(U)})_\bC$ and $u_\bC({\bf A})=(u({\bf A}))_\bC$, gives the thesis.
Regarding (f), from (a), (d) and (e), $D_G^{(U)}$ includes a dense set of analytic vectors for the, respectively, symmetric or anti symmetric operator $u(p({\bf A}))$.  (b) in Prop.\ref{propNseries} immediately proves (f).
\end{proof}
\noindent A final  remarkable consequence of the properties of Nelson's technology  and our version of Schur's lemma  is the following proposition whose proof is in Appendix \ref{appProof}.

\begin{proposition}\label{irrelie}
Let $\sH$ be an, either real or complex, Hilbert space  and $G \ni g \mapsto U_g$ is an irreducible strongly-continuous unitary representation of the  connected Lie-group $G$ on $\sH$.
If ${\bf M}\in E_\gg$ satisfies:

(i)  $[{\bf M},{\bf A}]_\gg=0$ for every ${\bf A}\in\gg$,

(ii) $u({\bf M})$ (defined on $D_G^{(U)}$) is essentially selfadjoint, 

\noindent then it holds  
$$
u({\bf M})=cI|_{D_G^{(U)}} \quad \mbox{for some $c\in\bR$.}$$
\end{proposition}

\section{Wigner elementary relativistic systems in real Hilbert spaces: Emergence of the complex structure}\label{secMAIN1}
Within this section we introduce a  first notion of elementary system with respect to the relativistic symmetry adopting the famous framework introduced by Wigner, but 
now formulated  also in a  {\em real} Hilbert space. As a result, both in the real and complex cases, we will find that the mathematical formulation of the theory naturally produces a complex structure  which is the trivial one in the complex case.
In the real case, it 
permits to reformulate all the model into a complex Hilbert space
 fashion. This  final complex model is in agreement with Sol\`er's picture and, differently from the initial real version, does not carry mathematical information without physical meaning because, differently form the initial  real case,
all selfadjoint operators represent observables.

\subsection{Wigner elementary relativistic systems}
As discussed in Sect.\ref{SecI1}, there are  three possible Hilbert-space formulations of QM on a Hilbert space $\sH$, respectively over $\bR, \bC$ or $\bH$,  as proved by Piron-Sol\`er's analysis of abstract lattices of elementary propositions of a quantum system. However as already noticed, several mathematical requirements assumed in Sol\`er's analysis  could be relaxed already for 
 the complex Hilbert space case and some further hypotheses regarding symmetries should be added. 
Observables, states, symmetries  are however described as discussed in points (1)-(8) of Section \ref{seclist} in the general case of a lattice of elementary propositions which does not coincide with the whole $\cL(\sH)$.
We do not require in fact that all selfadjoint operators 
of $\gB(\sH)$ represent observables, but the observables are the self adjoint operators whose spectral measures belong to a certain a von Neumann algebra $\gR$ which may or may not coincide to $\gB(\sH)$ and the
elementary propositions are the orthogonal projectors of $\cL_\gR(\sH)$ which is an orthomodular $\sigma$-complete sublattice of $\cL(\sH)$. As already discussed, unitary operators induce symmetries (not all symmetries in general).
Relying on the original Wigner's  ideas  about the notion of elementary system with respect to a group of symmetry, we assume 
that (1) an elementary relativistic system supports a faithful
 strongly-continuous irreducible unitary representation of the Poincar\'e group and (2) the  von Neumann algebra $\gR$ is generated by the representation itself.\\
We are therefore led to a definition written below, where the {\em proper orthochronous Poincar\'{e} group} actually  means the real simply-connected Lie group given by the semi-direct product $SL(2,\bC) \ltimes \bR^4$ which more properly is the {\em universal covering} of the proper orthochronous Poincar\'{e} group as understood in relativity. 
 This is because $SL(2,\bC) \ltimes \bR^4$ is the group which actually enters the physical constructions and every representation of the  proper orthochronous Poincar\'{e} group
 is also a representation of $SL(2,\bC) \ltimes \bR^4$.
For this reason we require a {\em local-faithfulness} assumption, i.e., the representation is only required to be injective in a neighborhood of the neutral element of $SL(2,\bC) \ltimes \bR^4$, since only in a nighborhood of the neutral elements 
$SL(2,\bC)$ and the proper orthochronous groups are identical. To corroborate our assumption, we observe that all the {\em complex} strongly-continuous unitary irreducible representations of 
$SL(2,\bC) \ltimes \bR^4$ with physical meaning are  locally faithful: (1)
for positive squared mass with integer spin they are faithful,  (2) for positive squared mass with semi-integer spin  they are faithful up to the sign of the $SL(2,\bC)$ element, so they are locally faithfull, and (3) they are againg faithful up to the sign of the $SL(2,\bC)$ element for zero squared mass and non-trivial momentum representation \cite{V2}. 

\begin{definition}\label{defERS}
A  real (complex) {\bf Wigner relativistic elementary system} ({\bf WRES}) is a unitary strongly continuous real (resp. complex) representation of the proper orthochronous Poincar\'{e} group $\cP$,
$$
U:\cP\ni g\mapsto U_g\in\gB(\sH)
$$
over the real (resp. complex) separable Hilbert space $\sH$ which is 
 {\em irreducible} and {\em locally faithful}, i.e., 
$U$ is injective in a neighborhood of the neutral element of $\cP$.\\ 
If $\gR_U$ is the von Neumann algebra generated by $U$ (definition \ref{defvnagen}), the {\bf observables} of the system
 are the selfadjoint operators $A$ whose  PVMs belong to $\cL_{\gR_U}(\sH)$.
\end{definition}

\begin{remark} $\null$\\
{\em {\bf (a)} Evidently, the {\bf bounded observables} are thus the selfadjoint operators of 
 $\gR_U$ and the {\bf elementary observables} are the elements of $\cL_{\gR_U}(\sH)$  itself.}\\
{\em {\bf (b)} In both the real and the complex case,  the PVM of $A=A^*$ belongs to  a von Neumann algebra $\gR$ if and only if $A$  is {\bf affiliated to $\gR$}, that is  $VA\subset AV$ for every unitary operator $V \in \gR'$. This easily arises from Th.\ref{st} (c) (ii) and Remark\ref{unitgeneralg}. Thus, the observables of a  WRES are all of  the selfadjoint operators affiliated to
$\gR_U$.\\ {\bf (c)} It is easy to prove that  $Vu({\bf M})\subset u({\bf M})V$, for every ${\bf M}\in E_\gp$ and every $V \in \gR_U' = \{U_g\}'_{g \in \cP}$.
Therefore $V\overline{u({\bf M})}\subset \overline{u({\bf M})}V$ also holds.
 As a consequence, if $u({\bf M})$ is essentially selfadjoint, then the selfadjoint operator $\overline{u({\bf M})}$ is an observable. We will come back later to this point in Corollary \ref{coraggEgoss} following another way.}
\end{remark}
\subsection{Emergence of an  complex structure (unique up to sign) from Poincar\'e symmetry} A strongly-continuous unitary representation  $U$ of $\cP$ (not necessarily
 irreducible or locally faithful) gives rise to an associated representation $u$ on the G\r arding domain $D_G^{(U)}$  of the corresponding Lie algebra
 $\gp$ of the proper orthochronous Poincar\'e group $\cP$, in accordance with Theorem \ref{teogarding} (d),
$$u:\gp\ni{\bf A}\mapsto u({\bf A}) : D_G^{(U)} \to \sH \:.$$
Let us fix a Minkowskian reference frame in Minkowski spacetime. From now on, for  $i=1,2,3$, $\V{k}_i\in \gp$  are the three generators of the {\em boost} one-parameter subgroups along the three spatial axes, 
$\V{l}_i\in \gp$ are the three generators of the {\em spatial rotation} one-parameter subgroups around the three axes, and 
$\V{p}_\mu \in \gp$, where  $\mu =0, 1,2,3$,   are the four  generators of the {\em spacetime displacements} one-parameter 
subgroups along the four Minkowskian axes. We have the well-known commutation relations for $i,j=1,2,3$,
\begin{equation}\label{COMM}
\begin{aligned}
&[\V{p}_0,\V{p}_i]_\gp=[\V{p}_0,\V{l}_i]_\gp=[\V{p}_i,\V{p}_j]_\gp=0,\ \ [\V{p}_0,\V{k}_i]_\gp= \V{p}_i\:,\\
&[\V{l}_i,\V{l}_j]_\gp= \sum_{k=1}^3\varepsilon_{ijk} \V{l}_k,\ \ [\V{l}_i,\V{p}_j]_\gp=\sum_{k=1}^3\varepsilon_{ijk} \V{p}_k,\ \ [\V{l_i},\V{k}_j]_\gp= \sum_{k=1}^3\varepsilon_{ijk} \V{k}_k\:,\\
&[\V{k}_i,\V{k}_j]_\gp= -\sum_{k=1}^3\varepsilon_{ijk} \V{l}_k,\ \ [\V{k}_i,\V{p}_j]_\gp= -\delta_{ij}\V{p}_0\:.\\
\end{aligned}
\end{equation}
We finally define the associated  basic {\em anti selfadjoint} generators
\begin{equation}\label{basispoinc}
\tilde{K}_i:=\overline{u(\V{k}_i)}\:, \ \ 
 \tilde{L}_i:=\overline{u(\V{l}_i)} \:, \ \
\tilde{P}_0:=\overline{u(\V{p}_0)}\:, \ \ \tilde{P}_i:=\overline{u(\V{p}_i)}\quad i=1,2,3\:,
\end{equation}
which satisfy the same commutation relations as the Lie algebra generators of  $\cP$, in accordance with with Theorem \ref{teogarding} (d) and for $i,j,k = 1,2,3$,
\begin{equation}
\begin{aligned}
&[\tilde{H},\tilde{P}_i]=[\tilde{H},\tilde{L}_i]=[\tilde{P}_i,\tilde{P}_j]=0,\ \ [\tilde{H},\tilde{K}_i]=c  \tilde{P}_i\:,\\
&[\tilde{L}_i,\tilde{L}_j]= \sum_{k=1}^3\varepsilon_{ijk} \tilde{L}_k,\ \ [\tilde{L}_i,\tilde{P}_j]=\sum_{k=1}^3\varepsilon_{ijk} \tilde{P}_k,\ \ [\tilde{L}_i,\tilde{K}_j]= \sum_{k=1}^3\varepsilon_{ijk} \tilde{K}_k\:,\\
&[\tilde{K}_i,\tilde{K}_j]= -\sum_{k=1}^3\varepsilon_{ijk} \tilde{L}_k,\ \  [\tilde{K}_i,\tilde{P}_j]= -c^{-1}\delta_{ij}\tilde{H}\:,\\
\end{aligned}
\end{equation}
where {\em every operator is evaluated on $D_G^{(U)}$} and we have introduced the  anti selfadjoint generator of the temporal  displacements  $\tilde{H}:=c\tilde{P}_0$ and $c$ is the speed of light as usual.\\
If $\sH$ is complex and $U$ represents a WRES,  these operators produce well-known basic observables simply by the addition of a factor $i$.  This is not the case if $\sH$ is real unless 
a common complex structure, commuting with each of them exists.
A symmetric  operator  defined on the invariant domain $D_G^{(U)}$, which will play a crucial role in our discussion, is 
\beq M^2_U := \left.\left(-\tilde{P}_0^2+ \sum_{k=1}^3\tilde{P}_k^2\right)\right|_{D_G^{(U)}}\:. \label{Moperator}\eeq
This operator is  physically associated  with the {\em squared mass} of the system (the apparent overall wrong sign on the right-hand side is due to the fact that the $\tilde{P}_\mu$ are anti selfadjoint  instead of selfadjoint).
Suppose that $\sH$ is real. We intend to prove that, if Definition  \ref{defERS} holds and $M^2_U \geq 0$, then  there exists  a -- unique up to the sign -- complex 
structure which makes the theory complex. A candidate for this complex structure is the  operator $J$ appearing  in the polar decomposition $\tilde{H}=J|\tilde{H}|$ (see Theorem \ref{Polar}). 
Within this picture, the selfadjoint and positive operator $H:=|\tilde{H}|$ may  be interpreted as the {\em energy operator}, the {\em Hamiltonian},  of the system. Especially exploiting   
Theorem \ref{polarCOMM}, we will prove that $J$ is actually a complex structure.
A sketch of an alternative proof of the existence of the complex structure $J$ appears in \cite{NeOl17} relying on the mathematical technology of {\em modular theory}.
 Since 
  it turns out that $J$ commutes with both every operator $U_g$ of the Poincar\'e group  representation and  every anti selfadjoint operator $\overline{u({\bf A})}$ associated with any ${\bf A}\in\gp$, $J$ is the wanted complex 
structure. Our $J$ exists also if $\sH$ is complex, but in this case  it reduces to the much more trivial operator $J = \pm i I$. Let us see everything  in details.

\begin{theorem}\label{poinccomplexstructure}
Consider an either {\em real  or complex Wigner elementary relativistic system}  and adopt definitions (\ref{basispoinc}) and (\ref{Moperator}).
Let $\tilde{H}:=c\tilde{P}_0$ and $\tilde{H}=JH$ its polar decomposition. The following facts hold provided $M^2_U \geq 0$.\\

\noindent {\bf (a)} $J\in \gR_U$ and $J$ is a complex structure on $\sH$. \\

\noindent {\bf (b)} $J \in \gR_U \cap \gR'_U$ because   $JU_g =U_gJ$ for all $g \in \cP$. In particular the complex structure $J$ is Poincar\'e invariant.\\

\noindent {\bf (c)}  $Ju({\bf A})\subset u({\bf A})J$ for every ${\bf A}\in\gp$ so that, in particular,    $J$ leaves $D^{(U)}_G$ invariant.\\

\noindent {\bf (d)}  If  $J_1$ is  a  complex structure on $\sH$ such that either $J_1 \in \gR_U'$
or $J_1 u({\bf A}) \subset u({\bf A}) J_1$ for every ${\bf A}\in\gp$ are valid, then $J_1=\pm J$. \\

\noindent {\bf (e)}   If ${\bf A} \in \gp$, then  $J\overline{u({\bf A})}=\overline{u({\bf A})}J$ and this operator is an observable of the system, i.e.,  it is selfadjoint and its PVM belongs to $\gR_U$. \\

\noindent {\bf (f)}  If $\sH$ is real,  passing to the complexified Hilbert space $\sH_J$, $$\cP \ni g \mapsto U_g : \sH_J \to \sH_J$$ 
defines a {\em complex WRES} whose associated von Neumann algebra is made of  the same operators  as for the initial real WRES, but now
$$\gR_U = \gB(\sH_J)\quad \mbox{and}\quad \cL_{\gR_U}(\sH_J) = \cL(\sH_J)\:,$$
in accordance with the thesis of  Sol\`er theorem {\em Sth} (this is false  if referring to $\sH$).\\

\noindent {\bf (g)}   If $\sH$ is complex, then $J=\pm i I$ and again 
$$\gR_U = \gB(\sH)\quad \mbox{and}\quad \cL_{\gR_U}(\sH) = \cL(\sH)\:,$$
in accordance with {\em Sth}.
\end{theorem}
\noindent To prove the theorem we need some intermediate results.

\begin{lemma}\label{lietransform}  Let $U: \cP \ni  g \mapsto U_g\in \gB(\sH)$ be a strongly continuous unitary representation over the, either real or complex, Hilbert space $\sH$.
Using the above defined notations, 
  for  $i=1,2,3$
 we have
\begin{align}
&e^{z\tilde{K}_i}e^{a\tilde{P}_0}e^{-z\tilde{K}_i}=e^{a\cosh z \:\tilde{P}_0}e^{-a \sinh z \:\tilde{P}_i}\quad \mbox{if} \quad a, z \in \bR\:. \label{eqKJ2}\\
&e^{z\tilde{K}_i}\tilde{P}_0e^{-z\tilde{K}_i}x=\cosh z\:\tilde{P}_0x-
\sinh z\: \tilde{P}_ix\quad \mbox{if} \quad   x \in D_G^{(U)}\:, z \in \bR\:. \label{eqKJ}
\end{align} 
\end{lemma}
\begin{proof}
Take $z,a\in\bR$, then a straightforward calculation with the one-parameter subgroups $\bR \ni s \mapsto \exp(s{\bf A})$ of $\cP$ 
gives $$\exp(z\V{k}_i)\exp(a\V{p}_0)\exp(-z\V{k}_i)=\exp(a\cosh z\V{p}_0)\exp(-a\sinh z\V{p}_i)\:.$$
Applying the representation $U$ to both  sides of this identity
we have (\ref{eqKJ2}).  Now, let $u,v\in D_G^{(U)}$.
Since the G\r arding domain is invariant under  $U$, it is easy to see that
\begin{equation*}
\begin{aligned}
&\left(\left.e^{z\tilde{K}_i}\tilde{P}_0e^{-z\tilde{K}_i}v\right|u\right)=\frac{d}{d a}\Big|_{a=0}\left(\left.e^{z\tilde{K}_i}e^{a\tilde{P}_0}e^{-z\tilde{K}_i}v\right|u\right)=\frac{d}{d a}\Big|_{a=0}\left(\left. 
e^{a(\cosh z)\tilde{P}_0}e^{-a(\sinh z)\tilde{P}_i}v\right|u\right)=\\
&=\frac{d}{d a}\Big|_{a=0}\left(\left. e^{-a(\sinh z)\tilde{P}_i}v\right|e^{-a(\cosh z)\tilde{P}_0}u\right)=\left(-(\sinh z)\left. \tilde{P}_iv\right|u\right)+\left(v\left|-(\cosh z\tilde{P}_0\right.)u\right)=\\
&=\left(\left.((\cosh z)\tilde{P}_0-(\sinh z)\tilde{P}_i)v\right|u\right)\:.
\end{aligned}
\end{equation*}
Since $D_G^{(U)}$ is dense, (\ref{eqKJ}) is true.
\end{proof}

\begin{lemma}\label{Mass} With the hypotheses of Theorem \ref{poinccomplexstructure}, the following facts are valid.\\
{\bf (a)} $M^2_U =\mu I|_{D_G^{(U)}}$ for some  $\mu\ge 0$. \\
{\bf (b)}  $Ker(\tilde{P_0}) = \{0\}$.
\end{lemma}
\begin{proof} (a) Let ${\bf M} = {\bf M}^+ = -\V{p}_0^2+\sum_k\V{p}^2_k  \in E_\gp $ be   the  element in the  universal enveloping algebra with  $u({\bf M})=M^2_U$, where $u$ is the associative algebra homomorphism defined in (e) of
 Theorem \ref{teogarding}.   
From (\ref{COMM}) we find $[{\bf M},{\bf A}]_\gp=0$ for every ${\bf A}\in\gp$ and thus
 $[{\bf M}, {\bf N}]_\gp=0$ where ${\bf N}$ is a Nelson element of $E_\gp$. Thus Theorem \ref{teogarding} (g) proves that $M^2_U := u({\bf M})$ is essentially selfadjoint on $D_G^{(U)}$ and  Proposition \ref{irrelie}
implies that $M^2_U= \mu I|_{D_G^{(U)}}$.  In particular $\mu \geq 0$ if and only if $M^2_U\geq 0$. \\
(b) Let us first suppose that $\mu >0$. Let $x\in Ker(\tilde{P}_0)$. Since $D_G^{(U)}$ is a core for  $\tilde{P}_0$ (Theorem \ref{teogarding} (f)), there is a sequence $D_G^{(U)} \ni x_n \to x$ with $\tilde{P}_0x_n \to \tilde{P}_0x =0$. 
As a consequence, taking advantage of  the definition of $M^2_U$,
$$ \sum_{k=1}^3 ( \tilde{P}_k x_n| \tilde{P}_k x_n) = -\mu (x_n|x_n) + ( \tilde{P}_0 x_n| \tilde{P}_0 x_n)\:.$$
For $n \to +\infty$, the right-hand side converges to $-\mu ||x||^2$, so that
\beq \lim_{n \to + \infty}  \sum_{k=1}^3 || \tilde{P}_k x_n||^2 = - \mu ||x||^2\:.\label{LL}\eeq
Since the right-hand side is non-positive whereas the left-hand side is non-negative, we conclude that
$\lim_{n \to + \infty}  \sum_{k=1}^3 || \tilde{P}_k x_n||^2 = - \mu ||x||^2 =0$.
With our hypothesis $\mu >0$, we find $x=0$ and thus $Ker(\tilde{P}_0)= \{0\}$. Let us pass to the remaining case $\mu=0$. Now
(\ref{LL}) implies $\lim_{n \to + \infty}  \sum_{k=1}^3 || \tilde{P}_k x_n||^2=0$ and therefore $\lim_{n \to + \infty} || \tilde{P}_k x_n||^2=0$ for every $k=1,2,3$. 
Since $\tilde{P}_k$ is closed and $x_n \to x \in Ker(\tilde{P}_0)$, we conclude that
$Ker(\tilde{P}_0) \subset D(\tilde{P}_k)$ and, more precisely,  \beq Ker(\tilde{P}_0) \subset Ker(\tilde{P}_k)\:, \quad k=1,2,3\:.\label{KerKer} \eeq
To go on observe that, from Stone theorem,
\beq Ker(\tilde{P}_0) = \{x \in \sH \:|\: e^{t\tilde{P}_0}x =x \quad \forall t \in \bR\} = \{x \in \sH \:|\: U_{\exp (t \V{p_0})}x =x \quad \forall t \in \bR\}\label{KKKP0}\eeq
Since $\exp (t \V{p_0})$ commutes with the one-parameter subgroups generated by $\V{p}_j$
 and $\V{l}_j$, we have from (\ref{KKKP0}) that $Ker(\tilde{P}_0)$ is invariant under the corresponding subgroups unitarily represented through $U$.
However, from (\ref{KerKer}) which immediately implies
\beq  e^{b  \:\tilde{P}_k}x=x  \quad \mbox{if $x \in Ker(\tilde{P}_0)$ and $b \in \bR$}\label{KerKer2}\:,\eeq
we also conclude that $Ker(\tilde{P}_0)$ is invariant under the unitary representation of one-parameter group generated by every $\V{k}_i$.
Indeed, from (\ref{eqKJ2}) and (\ref{KerKer2}), for $x \in Ker(\tilde{P}_0)$,
$$e^{a\tilde{P}_0}e^{-z \tilde{K}_i}x =e^{-z\tilde{K}_i} e^{a\cosh z \:\tilde{P}_0}e^{-a \sinh z  \:\tilde{P}_i}x = e^{-z\tilde{K}_i} e^{a\cosh z \:\tilde{P}_0}x =  e^{-z\tilde{K}_i}x \quad \forall a, z \in \bR\:,$$
hence $e^{z\tilde{K}_i}x \in  Ker(\tilde{P}_0)$ in accordance to (\ref{KKKP0}).
 Since $\cP$ is a connected Lie group, every $g \in \cP$ is the product of a finite number of elements of one-parameter groups generated by the vectors of every
 fixed basis of $\gp$. Lifting this result to the Hilbert space $\sH$ by means of the representation $U$, we conclude that the closed subspace $Ker(\tilde{P}_0)$ is
 invariant under $U$. Since $U$ is irreducible, either $Ker(\tilde{P}_0)= \{0\}$ or $Ker(\tilde{P}_0)= \sH$. In the second case $\tilde{P}_0 =0$ (and more strongly $\tilde{P}_k=0$ for $k=1,2,3$ from (\ref{KerKer})).
In this case $U_{\exp(t\V{p}_0)} =I$ for every $t \in \bR$ and thus $\cP \ni g \mapsto U_g$ is not 
locally faithful
 contrarily to the hypothesis on $U$. We conclude that $Ker(\tilde{P}_0)= \{0\}$ also if $\mu=0$.
\end{proof}

\begin{proof}[Proof of Theorem \ref{poinccomplexstructure}] $\null$

{\bf (a)} First, notice that the polar 
decomposition of $\tilde{P}_0$ is simply given by $J(c^{-1}H)$. Lemma \ref{Mass} (b) says that  $Ker(\tilde{H})=\{0\}$, hence 
Prop.\ref{LemmaCOMM} (d),(e) guarantees that $J$ satisfies $J^*=-J$ and $JJ=-I$.
Let us prove that $J \in \gR_U$. If $B \in \gR'_U$, then $[B, e^{t\tilde{H}}]=0$ for every $t \in \bR$. Prop.\ref{LemmaCOMM} (a) implies that $B\tilde{H}\subset \tilde{H} B$ and thus
$[B,J]=0$ in view of (b) of the same Proposition. In other words $J \in \gR_U''= \gR_U$.
 The proof of  (a) is concluded.

{\bf (b) and (c)}   It should be clear that 
$JU_g =U_gJ$ for all $g \in \cP$ and  $Ju({\bf A})\subset u({\bf A})J$ for every ${\bf A}\in\gp$ are 
 equivalent statements due to Proposition \ref{alggroupcomm} and Corollary \ref{CORRalggroupcomm}. Furthermore $JU_g =U_gJ$ for all $g \in \cP$ is the same as $J \in \gR_U'$ because
   $\gR'_U= \{U_g\}_{g\in \cP}'''=  \{U_g\}_{g\in \cP}'$.  Therefore we will prove  only
   that $JU_g =U_gJ$ for all $g \in \cP$.
   We divide this technical proof into six parts where  we will denote $J$ by $J_H$ for notational convenience.

{\em First part}.  Let ${\bf A}\in \gp$ such that $[\V{p}_0,{\bf A}]=0$, then Baker-Campbell-Hausdorff formula together with  the fact that $U$ is a representation,  give $U_{exp(t\V{p}_0)}U_{exp(s{\bf A})}=U_{exp(s{\bf A})}U_{exp(t\V{p}_0)}$, i.e. 
$e^{t\tilde{P}_0}e^{sA}=e^{sA}e^{t\tilde{P}_0}$ for every $s,t\in\bR$, where $A=\overline{u({\bf A})}$.
Prop. \ref{LemmaCOMM} and \ref{polarCOMM} imply that, for the mentioned  ${\bf A}\in \gp$ commuting with $\V{p}_0$ and referring to their polar decomposition   $A=J_A|A|$, we have (0) $J_He^{tA}=e^{tA}J_H$ (1) $J_HA=AJ_H$, (2) $J_H|A|=|A|J_H$,
 (3) $J_H\sqrt{|A|}=\sqrt{|A|}J_H$ and (4) $J_HJ_A=J_AJ_H$.   Notice that in particular, thanks to point (0), $J_H$ commutes with the one parameter subgroups generated by $\V{p}_0,\V{p}_i,\V{l}_i$. All these identities will be exploited  shortly. 

{\em Second part}. Let us focus attention to the boost generators $\V{k}_i$, the associated (unitary) one parameter subgroups and their  anti selfadjoint generators  $\tilde{K}_i$.
 We want to prove that, exactly as it happens for the already discussed one-parameter subgroups,  $J_He^{z \tilde{K}_i}=e^{z \tilde{K}_i}J_H$ if $z \in \bR$.  To this end, observe that the polar decomposition of the closed
 operator $X:=e^{z\tilde{K}_i}\tilde{H} e^{-z\tilde{K}_i}$ is trivially constructed out of the polar decomposition of $\tilde{H}$ and reads  $X= [e^{z\tilde{K}_i}J_H e^{-z\tilde{K}_i}][e^{z\tilde{K}_i}H e^{-z\tilde{K}_i}]$ since the two factors 
satisfy the requirements listed in  Theorem \ref{Polar} fixing the polar decomposition of $X$. However it  also holds
$X=(J_H)(-J_H e^{z\tilde{K}_i}\tilde{H} e^{-z\tilde{K}_i})$, hence if we succeed in proving that also the couple $U:=J_H, B:=-J_H e^{z\tilde{K}_i}\tilde{H} e^{-z\tilde{K}_i}$ satisfies the 
conditions of Theorem \ref{Polar} and therefore defines another  polar decomposition of $A$, then by uniqueness of the polar decomposition we get in particular that  $J_H=e^{z\tilde{K}_i}J_H e^{-z\tilde{K}_i}$. 
This is our thesis $J_H e^{z\tilde{K}_i}=e^{z\tilde{K}_i}J_H$. 

{\em Third part}. According to the final comment in the second part of the proof, let us prove that the above defined operators  $U,B$ satisfy the  requirements (i)-(iv)  listed in  Theorem \ref{Polar}. Item (i) is true by construction.  Item (iii) is trivial, since $J_H$ is unitary. Item
  (iv) is equivalent to $Ker(B)=\{0\}$ which is immediate since $J_H, e^{\pm z\tilde{K}},\tilde{H}$ are all injective. It remains to prove (ii),  that $B= -J_H e^{z\tilde{K}}\tilde{H} e^{-z\tilde{K}}$ is positive and selfadjoint. 
This third part is  devoted to rephrase the positivity property of  $B$ into a more operative way. 
 It is useful to start  by  considering the generator of 
the space translations $\tilde{P}_i=J_{P_i}P_i$, where $P_i:=|\tilde{P}_i|$, noticing  that $J_{P_i}\sqrt{P_i}\subset \sqrt{P_i} J_{P_i}$ thanks again to Proposition \ref{LemmaCOMM}. 
Furthermore, since $[\V{p}_0,\V{p}_i]=0$, the  identities  established in the first part of this proof  hold  for $A=\tilde{P}_i$.
Thanks to Lemma \ref{lietransform}, we get for $v\in D_G^{(U)}$ that
\begin{equation}
\begin{aligned}
(v|Bv) &=\left(v\left|-J_H e^{z\tilde{K}_i}\tilde{H} e^{-z\tilde{K}_i}\right.v\right)=\left(v
\left|(\cosh z)(-J_H\tilde{H}) v-(c\sinh z )(-J_H\tilde{P}_i) \right. v\right)=\\
&=\cosh z\left((v|H v)-(c\tanh z)(v|(-J_H\tilde{P}_i) v)\right) \:.\label{Bpos}
\end{aligned}
\end{equation}
Since $\cosh z>0$ and $|\tanh  z|<1$ for every $z\in\bR$, in order to prove that $(v|Bv)\ge 0$ for $v \in D_G^{(U)}$  it suffices to prove that $(v|H v)\ge c|(v|(-J_H\tilde{P}_i) v)|$. Define $S:=-J_HJ_{P_i}$, 
which is clearly selfadjoint thanks to the properties of the $J$ operators, then it holds $|(x|Sx)|\le(x|x)$ for every $x\in\sH$ because both operators have norms bounded by $1$.
We have $J_HJ_{P_i}\sqrt{P_i}\subset J_H\sqrt{P}_i J_{P_i}=\sqrt{P_i}J_HJ_{P_i}$, from which it follows that $S\sqrt{P_i}\subset\sqrt{P_i} S$.
Now, let $v\in D_G^{(U)}\subset D(P_i)$. It holds $P_i=\sqrt{P_i}\sqrt{P_i}$, hence $v\in D(\sqrt{P_i})$ and $\sqrt{P_i} v\in D(\sqrt{P_i})$. All this justifies what follows
\begin{equation}
\begin{aligned}
&|(v|-J_H\tilde{P_i} v)|=|(v|-J_HJ_{P_i}P_i v)|=|(v|SP_i v)|=|(v|S\sqrt{P_i} \sqrt{P_i} v)|=\\
&=|(v|\sqrt{P_i} S\sqrt{P_i} v)|=|(\sqrt{P_i} v| S\sqrt{P_i} v)|\le (\sqrt{P_i} v|\sqrt{P_i} v)=(v|P_i v)\:.
\end{aligned}
\end{equation}
Thanks to this inequality, it suffices to prove that $c(v|P_i v)\le (v|H v)$ to conclude from (\ref{Bpos}) that $B\geq 0$
on $D_G^{(U)}$.
 The proof of the positivity property of $B$ ends if establishing  the inequality \beq c(v|P_i v)\le (v|H v)\quad \mbox{if $v \in D_G^{(U)}$}\label{ineq}\eeq and next extending the result to the 
full domain of $B$. This is done within the next part of the proof.

{\em Fourth part}. Thanks to Lemma \ref{Mass}, defining $m^2 = c^{-2}\mu \geq 0$, we have $-\tilde{H}^2 v=(mc^2)^2v-\sum_{k=1}^3c^2\tilde{P}_k^2v$ if $v \in D_G^{(U)}$, from which
\begin{equation}\nonumber
\begin{aligned}
(v|-\tilde{H}^2v)&=m^2c^4-c^2\sum_{i=1}^3(v|\tilde{P}_i^2 v)=m^2c^4+c^2\sum_{i=1}^3(\tilde{P}_i v|\tilde{P}_i v)\ge c^2(\tilde{P}_k v|\tilde{P}_k v)\ \ k=1,2,3
\end{aligned}
\end{equation}
where we supposed, without loss of generality, that $\|v\|=1$.
In other words,
\beq
\|\tilde{H} v\|\ge c\|\tilde{P}_k v\| \quad \mbox{if $v \in D_G^{(U)}$}\label{IneqGard}\:.
\eeq
Our next step consists in proving that (\ref{IneqGard}) extends to the whole $D(\tilde{H})\cap D(\tilde{P}_k)$, which is actually equal to $D(\tilde{H})$. Let us prove this. From 
Theorem \ref{teogarding} (f),  we know  that both $\tilde{H}$ and $\tilde{P}_k$ are the closures of their restriction to $D_G^{(U)}$. So,  if  $v\in D(\tilde{H})$, there exists $\{v_n\}_{n \in \bN}\subset D_G^{(U)}$ such that $v_n\to v$ and $\tilde{H} v_n\to \tilde{H} v$. Thanks to (\ref{IneqGard}), we see that $\{\tilde{P}_k v_n\}_{n \in \bN}$ is a Cauchy 
sequence in $\sH$, thus converging  to some $y\in\sH$. Since  $\tilde{P}_k$ is closed, $v\in D(\tilde{P}_k)$ and $y= \tilde{P}_k v$. This gives $D(\tilde{H})\subset D(\tilde{P}_k)$. Now, we 
have $\|\tilde{P}_k v\|=\lim_{n\to\infty} \|\tilde{P}_k v_n\|\le \lim_{n\to\infty}c^{-1}\|\tilde{H} v_n\|= c^{-1}\|\tilde{H} v\|$, hence (\ref{IneqGard}) is valid on 
$D(\tilde{H})\cap D(\tilde{P}_k) = D(\tilde{H})$.
This result implies (\ref{ineq}) as we go to prove.\\
Everything we have  so far established is valid for both real or complex  $\sH$. Here we make a distinction. First assume that $\sH$ is complex.
 Notice that  the spectral measures of $i\tilde{H}$ and $i\tilde{P}_k$ commute because 
$e^{t\tilde{H}}e^{s\tilde{P}_k}=e^{s\tilde{P}_k}e^{s\tilde{H}}$ for every $s,t\in\bR$
(Thm 9.35 in \cite{M}). As $\sH$ is separable,  this guarantees 
the existence of a {\em joint spectral measure} $E$ on $\bR^2$ (e.g., \cite{M}) such that $f(i\tilde{H})=\int_{\bR^2}f(\lambda_1)dE(\lambda)$ and $f(i\tilde{P}_k)=\int_{\bR^2}f(\lambda_2)dE(\lambda)$ for 
every measurable function $f$ on $\bR^2$ where $\lambda=(\lambda_1,\lambda_2) \in \bR^2$. Moreover $E(\Delta)\tilde{H}\subset \tilde{H}E(\Delta)$ for every Borelian $\Delta\subset\bR^2$, hence in particular it holds $E(\Delta)(D(\tilde{H}))\subset D(\tilde{H})$. Now,
if  $v\in D(i\tilde{H})=D(\tilde{H})$, (\ref{IneqGard}) which is valid on the whole $D(\tilde{H}) = D(\tilde{H})\cap D(\tilde{P}_k)$ as proved above, and (\ref{inin}) yield
\begin{equation}
\begin{aligned}
\int_\Delta|\lambda_1|^2d\mu_v(\vec{\lambda})= 
\int_{\bR^2}|\lambda_1|^2d\mu_{E_\Delta v}(\vec{\lambda}) =
||i\tilde{H} E_\Delta v||^2 \ge c^2||i\tilde{P}_k E_\Delta v||^2 =c^2\int_\Delta|\lambda_2|^2d\mu_v(\vec{\lambda})\:,
\end{aligned}
\end{equation}
for every Borelian 
$\Delta\subset\bR^2$.
So that $\int_\Delta (|\lambda_1|^2-c^2|\lambda_2|^2)d\mu_v\ge 0$ for every Borelian $\Delta$. As a consequence $|\lambda_1|^2-c^2|\lambda_2|^2\ge 0$ almost everywhere $\bR^2$ 
with respect to the measure $\mu_v$. 
This implies that also  $|\lambda_1|\geq c|\lambda_2|$  almost everywhere on $\bR^2$ with respect to $\mu_v$ because 
$$0 \leq (|\lambda_1|^2-c^2|\lambda_2|^2)= (|\lambda_1|-c|\lambda_2|) (|\lambda_1|+c|\lambda_2|)\:.$$
As an immediate consequence, if $v\in D_G^{(U)} \subset D(H)$,
\begin{equation}
(v|H v)=(v||i\tilde{H}|v)=\int_{\bR^2}|\lambda_1|d\mu_v\ge c\int_{\bR^2}|\lambda_2|d\mu_v= c(v||i\tilde{P}_k|v)=c(v|P_k v)\:.
\end{equation}
Since (\ref{ineq}) therefore holds, in accordance with the third part of this proof, this concludes the proof of $(v|Bv)\ge 0$ for every $v\in D_G^{(U)}$ when $\sH$ is complex. \\
If $\sH$ is real, passing to the complexified representation $U_\bC$ (which is not necessarily irreducible differently from $U$ but this fact does not play any role in this part of the proof), 
(\ref{IneqGard})  is valid for $\tilde{H}_\bC$ and $\tilde{P}_{k\bC}$ on $D(\tilde{H}_\bC)$ and the generated one-parameter groups still commute. Thus the proof used for the complex case  is valid for $\tilde{H}_\bC$ and $\tilde{P}_{k\bC}$ giving rise to (\ref{ineq}) on $D_G^{(U_\bC)}$ for the absolute values of these operators, hence for 
 $H_\bC$ and $P_{k\bC}$, because
 $|A_\bC|=|A|_\bC$ (Theorem \ref{Polar} (c)).
By direct inspection, one sees that 
$$c(x+iy|P_{k\bC} (x+iy))_\bC \leq  (x+iy|H_\bC (x+iy))_\bC\quad x+iy \in D_G^{(U_\bC)}$$ 
implies (\ref{ineq}), namely
$$c(v|P_k v)) \leq  (v|Hv)\quad v \in D_G^{(U)}$$
exploiting  $D_G^{(U_\bC)}= (D_G^{(U)})_\bC$   (last statement in Theorem  \ref{DMtheorem})
and the fact that $H= |\tilde{H}|$ and $P_k= |\tilde{P}_k|$ are symmetric on $D_G^{(U)}$.
 Since we have established that (\ref{ineq}) is valid also in the real case, we conclude again from the third part of this proof that $B \geq 0$ on $D_G^{(U)}$ also if $\sH$ is real.\\
Let us finally extend the property $B\geq 0$ to the whole domain of $B$.
$B$ is the composition of $-J_H$, which is unitary, and $e^{z\tilde{K}}\tilde{H} e^{-z\tilde{K}}$ which 
is easily seen to be the closure of its restriction to the G\r arding domain $D_G^{(U)}$
($e^{z\tilde{K}}$ its a bijection from this domain to itself  and this set is a core for $\tilde{H}$, see Theorem \ref{teogarding} (b),(f) and Remark \ref{essentclos} (f)). The same Remark shows that $B$ is closed with $D_G^{(U)}$ as a core. This immediately implies that $B$ is positive on its domain it being positive on a core.\\
Before going on, we observe that the proof to show that (\ref{ineq}) implies $B\geq 0$
on $D_G^{(U)}$  can be rephrased as it stands for the relevant complexified operators of the complexified representation $U_\bC$ if $\sH$ is real,
so that we have also established that $B_\bC = -J_{H_\bC} e^{z\tilde{K}_\bC}\tilde{H}_\bC e^{-z\tilde{K}_\bC}$ is positive on $D_G^{(U_\bC)}$ and where $J_{H_\bC}= (J_{H})_\bC$.
Finally also the proof of the fact that $B$ is positive on its whole domain extends as it stands to $B_\bC$ on its whole domain.

{\em Fifth part}. Now let us prove that $B=B^*$. First assume that $\sH$ is complex. If  $x,y\in D(B)$, comparing the expansions of  $(x+y|B(x+y))$ and
 $(x+iy|B(x+iy))$ taking into account the fact that $(u|Bu)\in \bR$ since $B$ is positive, we easily have $(x|By) = \overline{(y|Bx)}$. In other words, $B\subset B^*$. Now, since $J_H$ is bounded, from $B= -J_H e^{z\tilde{K}}\tilde{H} e^{-z\tilde{K}}$ we have (see Remark \ref{sumselfadj})
$$B^*=e^{z\tilde{K}}\tilde{H}^* e^{-z\tilde{K}}(-J_H)^*=-e^{z\tilde{K}}\tilde{H} e^{-z\tilde{K}}J_H$$
so that $B \subset B^*$ can be rephrased as
$-J_H e^{z\tilde{K}}\tilde{H} e^{-z\tilde{K}} \subset -e^{z\tilde{K}}\tilde{H} e^{-z\tilde{K}}J_H$.
Applying $J_H$ on the left side and $-J_H$ on the right one, we find
$-e^{z\tilde{K}}\tilde{H} e^{-z\tilde{K}}J_H \subset -J_He^{z\tilde{K}}\tilde{H} e^{-z\tilde{K}}$,
that is
$B^*\subset B$ and thus
 $B=B^*$.\\
If $\sH$ is real, as observed at the end of the fourth part of this proof,  we know that $B_\bC \geq 0$ and, exploiting again the proof above for the complexified operators, we have $B_\bC = B^*_\bC$. Proposition \ref{prop2} (2) finally implies $B=B^*$.\\
We have so far established that $U, B$ satisfy the requirement listed in requirements (i)-(iv)  listed in  Theorem \ref{Polar}, so that 
$J_H$ commutes with the unitary representation of the one-parameter groups generated by $\V{k}_i$.

{\em Sixth part}.  Let us conclude our proof by establishing that  $J_H U_g = U_g J_H$ is true for every $g \in \cP$.
From the previous five parts of the proof we know that $J_H$ commutes with the unitary representations of the
one-parameter subgroups generated by each element of the natural basis of the Lie algebra of $\gp$ made of the vectors  $\V{p}_\mu$, $\V{l}_k$, $\V{k}_k$. As is well-known, there is a sufficiently small neighborhood $O$ of the identity element of a Lie group  group ($\cP$ in our case) whose elements are products of a finite number of elements on one-parameter subgroups generated by an arbitrarily fixed basis of the Lie algebra. Also,  if the group is connected (as $\cP$ is), every element of the group can be written as product of finite elements chosen in a, arbitrarily fixed, neighborhood of the identity element of the group. Since $J_H$ commutes with the unitary representations of the
one-parameter subgroups generated by each element of the natural basis of the Lie algebra of $\gp$, it therefore commutes with every element $U_g$ of the representation: $J_HU_g=U_gJ_H$ for every $g\in \cP$ concluding the proof of (b)   and (c).

{\bf (d)} First of all, notice that, in view of Prop.\ref{alggroupcomm}, $J_1 u({\bf A}) \subset u({\bf A}) J_1$ implies $J_1 U_g = U_g J_1$ for every $g\in \cP$. In other words $J_1 \in \gR_U'$. We therefore may assume  $J_1 U_g = U_g J_1$ in every case.
The proof of (d) in the complex case immediately arises from (g) whose proof is independent from this one.  Let $\sH$ be real.  Exploiting Prop.\ref{LemmaCOMM} (a) and (b) for the one-parameter group generated by $\tilde{P}_0$, we get  $J_1J=JJ_1$. This means that $J_1$ is complex linear with respect to the complex structure induced on $\sH$ by $J$. 
Since $U$ is irreducible with respect to $\sH$, it is irreducible also referring to $\sH_J$. Since $J_1$ commutes with every $U_g$, the complex  version of Schur lemma
for complex Hilbert spaces implies  $J_1=(\alpha+\beta i)I = \alpha I+\beta J$ for some $\alpha,\beta \in \bR$. Since both $J,J_1$ are simultaneously anti selfadjoint and unitary it follows that $\alpha=1$ and $\beta=\pm 1$.

{\bf (e)} The fact that $J\overline{u({\bf A})}= \overline{u({\bf A})}J$ is equivalent to the already proved statement in (c) due to Prop.\ref{alggroupcomm} and Corollary \ref{CORRalggroupcomm}. Since $J$ is bounded, it holds $(J\overline{u({\bf A})})^*=\overline{u({\bf A})}^*J^*=\overline{u({\bf A})}J=J\overline{u({\bf A})}$ (see Remark \ref{sumselfadj}): this proves that $J\overline{u({\bf A})}$ is selfadjoint. It remains to prove that its PVM is contained in $\gR_U$. Let $B\in \gR_U'$, then $Be^{t\overline{u({\bf A})}}=BU_{\exp(t{\bf A})}=U_{\exp(t{\bf A})}B=e^{t\overline{u({\bf A})}}B$, from which it follows $B\overline{u({\bf A})}\subset \overline{u({\bf A})}B$ thanks to Prop.\ref{LemmaCOMM} (a). This inclusion, together with $J\in\gR_U$, gives $B[J\overline{u({\bf A})}]\subset [J\overline{u({\bf A})}]B$ and so, Th.\ref{st} (c) (ii) guarantees that the PVM of $J\overline{u({\bf A})}$ commutes with $B$. It being $B\in\gR_U'$ generic and $\gR_U=\gR_U''$, we have the thesis.

{\bf (f)} Everything  is easily established by collecting the previously obtained results.
In particular, since $J \in \gR_U \cap \gR_U'$ both the commutant $\{U_g,\ g\in \cP\}'=\gR_U'$ and the double commutant 
$\{U_g\}''_{g \in \cP}= \gR_U$ computed in
 $\sH$ coincides with the respective commutant and double commutant in $\sH_J$. Since $\gR_U$ is irreducible in $\sH$ it is also irreducible in $\sH_J$. Due to proposition \ref{SL2}, $\gR_U'= \{cI\}_{c \in \bC}$ so that  $\gR_U = \gR_U''= \{cI\}_{c \in \bC}' = \gB(\sH_J)$ and thus 
$\cL_{\gR_U} = \cL_{\gB(\sH_J)}= \cL(\sH_J)$.

{\bf (g)} Item (b)  together with the complex case of   Schur lemma guarantees that $J=cI$, for some $c\in\bC$.  Since $J$ is anti selfadjoint and unitary it must be $c=\pm 1$. The proof of the second part is immediate.
\end{proof}
\begin{remark}$\null$\\
{\em {\bf (1)} A (real or complex) WRES  admits in particular the following observables associated with the generators of $U$ respectively called: {\bf energy}, {\bf three momentum components},
{\bf angular momentum components}, {\bf boost components}:
$$H = -J \tilde{H}\:,\quad  P_k = -J \tilde{P}_k\:, \quad  L_k = -J \tilde{L}_k\:, \quad  K_k = -J \tilde{K}_k\:,  \quad k=1,2,3\:.$$
All these observables but $K_k$ are constants of motion since they commutes (like $J$ does) with the time evolution generated by $H$. (Each $K_k$ define a $t$-parametrized constant of motion e.g., see \cite{M}.)
This system has ``non-negative energy'' $H = -J \tilde{H} = |\tilde{H}|\geq 0$ and ``non-negative squared  mass`` $m^2 \geq 0$, where  $\overline{M^2_U} = m^2 I$, if the positivity hypothesis on $M_U^2$ is assumed.
The case of vanishing mass is encompassed in the proved theorem  in spite of the absence of the position operator. At least in  this case Heisenberg principle cannot be stated and St\"uckelberg's argument cannot be used to rule out real quantum mechanics.\\
{\bf (2)} The found representation is
 an irreducible unitary complex strongly-continuous
 representation of $SL(2,\bC) \ltimes \bR^4$
which is also
 locally faithful and thus, in particular, the representation of its Abelian  subgroup $\bR^4$
admits non-vanishing self-adjoint generators. It therfore  admits the well-known features known from Wigner-Mackey theory (see e.g. \cite{V2}).  For positive mass $m^2$,
it must include an irreducible  finite-dimensional representation of $SU(2)$ and thus a definite value of the {\em spin}, constructed out of the {\em Pauli-Lubanski} vector.
If $m^2=0$, the representation either admits definite {\em helicity} or a non-physical infinite-dimensional spin representation.
We do not enter into the details of these structure because this subject is very well known.\\
{\bf (3)} Physically speaking, (b) has an important implication.
If $J_1$ is associated to  the generator of temporal displacements of a different Minkowskian reference frame, connected with the initial one by means of a transformation of $\cP$, then $J_1=J$. Indeed, observe that, changing Minkowskian reference frame 
by means of $g \in \cP$, the new generator $\tilde{H}'$ of the time displacements is expected to be  related to the initial one by means of the relation $\tilde{H}' = U_g \tilde{H} U_g^*$.
The unitary operator $J'$ appearing in the polar decomposition of $\tilde{H}'$ is therefore $J' = U_g J U_g^* = U_gU_g^*J = J$. In other words, the complex formulation is {\em independent} from the choice of Minkowskian reference frame.}
\end{remark}
\noindent  We have  a final  technical  corollary regarding the interplay of $J$ and the extension of  $u$ over the whole $E_\gp $ also establishing how  unbounded observables are generated by $J$ and elements of $E_\gp$.

\begin{corollary}\label{coraggEgoss}
In the same hypotheses of Theorem \ref{poinccomplexstructure}, it holds 
$$J(D_G^{(U)}) = D_G^{(U)} \quad \mbox{and}\quad J (u({\bf M})+Ju({\bf N})) = (u({\bf M})+Ju({\bf N}))J\quad \mbox{for every ${\bf M}, {\bf N}\in E_\gp$.}$$
More strongly
$$J \overline{u({\bf M}) + Ju({\bf N}) } = \overline{u({\bf M}) + Ju({\bf N}) }J\quad \mbox{for every ${\bf M}, {\bf N}\in E_\gp$.}$$
If $\overline{u({\bf M}) + Ju({\bf N}) }$ is selfadjoint, then it is  an observable of the WRES.
\end{corollary}
\begin{proof}
We know from (c) that $J(D_G^{(U)}) \subset D_G^{(U)}$ and thus $J(J(D_G^{(U)})) \subset J(D_G^{(U)})$ which implies $D_G^{(U)}\subset J(D_G^{(U)})$ because $JJ=-I$ and $D_G^{(U)}$ is a subspace and thus 
$J(D_G^{(U)}) = D_G^{(U)}$. Next,
from the very definition of $u({\bf M})$
(the properties of the associative unital $^*$-algebra homomorphism $u$ of (e) in Theorem \ref{teogarding})
 and (c) we also have $J u({\bf M})\subset u({\bf M})J$ which can be made more precise into $J u({\bf M})=u({\bf M})J$ because $J^{-1}(D_G^{(U)}) = D_G^{(U)}$ is the domain of $u({\bf M})J$ by definition of composite operators, but  $D_G^{(U)}$ is also the domain of $J u({\bf M})$ since $D(J)= \sH$. 
With a trivial extension of the used argument,  $J (u({\bf M}) + Ju({\bf N})) = (u({\bf M}) + Ju({\bf N}))J$
 Taking the closures of both sides, since $J$ is bounded, we have
$J \overline{u({\bf M}) + Ju({\bf N}) } = \overline{u({\bf M}) + Ju({\bf N}) }J$. 
Regarding the last sentence, if $\overline{u({\bf M}) + Ju({\bf N}) }$ is selfadjoint, due to the second statement above and Th.\ref{st} (c) (ii), its  PVM  commutes with $J$  and thus belongs to $\gB(\sH_J)= \gR_U$ namely to $\cL_{\gR_U}(\sH)$ so that $\overline{u({\bf M}) + Ju({\bf N}) }$ is an observable of the WRES.
\end{proof}

\section{A physically more accurate  approach: Emergence of the complex structure}\label{secMAIN2}
The discussion in the previous section though leading to an interesting final result was based on a not very accurate notion of relativistic elementary physical system as presented in Definition \ref{defERS} in terms of {\em Wigner elementary physical system}. 
The first problem with that definition
 concerns the nature of the Poincar\'e representation. In principle, differently from  the complex case where 
everything is consequence of well-known Wigner's and Bargmann's theorems, in real Hilbert spaces there  is  no  reason to assume that the Poincar\'e symmetry is implemented (a) in terms of a {\em strongly-continuous} representation and (b) without taking possible {\em multiplicators} in front of the unitary operators into account. 
From the most general viewpoint, relying on the content of Sec.\ref{SecI1}, we should instead  assume that the action of Poincar\'e group is given  in terms of automorphisms of the lattice of elementary propositions of the system. Furthermore, the natural  notion of continuity of this sort of representation might  concern the probabilities associated to every state 
on the system, viewed as a $\sigma$-additive probability measure on the afore-mentioned lattice.
Secondly, two notions of irreducibility actually appeared into a mixed form in Definition \ref{defERS}, one concerned the group representation and the other regarding the algebra of observables. In principle one 
may try to  keep these two notions distinct from each other or prove that one is a consequence of the other.  \\
The next discussion will be performed for a  real  Hilbert space $\sH$ only. 
 The quaternionic case will be analyzed elsewhere in its whole generality. 
 The extension of our approach to the complex case would lead again to the result stated in Theorem \ref{poinccomplexstructure} for the complex Hilbert space case as the reader may straightforwardly prove using a proof strictly analogous to that  of Theorem \ref{bargmann}. Therefore the apparently rigid definition of complex WRES (Definition \ref{defERS}) is actually completely appropriate to describe 
{\em complex} elementary relativistic systems.\\
In the rest of this section we will make use of some notions and results of quaternionic Hilbert spaces already exploited in \cite{V2} and summarized in Appendix \ref{QHS}. For recent papers on this subject where the spectral theory is developed in details see  \cite{GMP1} and \cite{GMP2}.  
Here we just point out how a quaternionic Hilbert space can be obtained from a real one.
If $\sH$ is a real Hilbert space admitting two complex structures $J,K$ such that $JK=-KJ$, a quaternionic Hilbert space $\sH_{J,K}$ can be constructed out of $\sH$ in a fashion similar to the procedure to build the complex Hilbert space $\sH_J$
from the real Hilbert space $\sH$.
The elements of $\sH_{J,K}$ are the vectors of $\sH$ viewed as additive group and equipped with the right  product of vectors and quaternions
$$\psi (a1+ bi+cj+dk) := a\psi + bJK \psi + cJ \psi + d K\psi \quad \mbox{ for $a,b,c,d \in \bR$ and $\psi \in \sH$}$$ and the Hermitian scalar product is 
$$( \psi | \phi)_{J,K} := (\psi|\phi) - i (\psi|JK \phi) - j (\psi| J \phi) - k (\psi| K \phi)\:.$$
It turns out that   
$||x||_{J,K} := \sqrt{( x | x)_{J,K}}= ||x||$  for every $x\in \sH_{J,K}$  so that, in particular,  $\sH_{J,K}$ is complete because $\sH$ is. 
$\gB(\sH_{J,K})$ coincides with the subset of $\gB(\sH)$
whose elements commute with $J$ and $K$. 
It is  easy to prove that 
the adjoint of $A\in \gB(\sH)$ commuting with $J$ and $K$ coincides with the adjoint of $A$ viewed as an element of $\gB(\sH_{J,K})$. In particular $\cL(\sH_{J,K})$ coincides with the subset of $\cL(\sH)$
whose elements commute with $J$ and $K$. Moreover if $U\in\gB(\sH)$ commutes with $J,K$ then it is unitary on $\sH$ if and only if it is unitary on $\sH_{J,K}$.

\subsection{General  elementary quantum  systems, states and symmetries}
To describe an elementary physical system in the absence of any group of symmetry, it seems to be reasonable to assume that there is a von Neumann algebra of observables $\gR$, represented over the  Hilbert space $\sH$. The proper  observables of the system are the selfadjoint operators whose spectral measures belong to $\gR$  and the elementary observables are the elements of  the  lattice of projectors $\cL_{\gR}(\sH)$.  An elementary system should not have superselection rules (or we can always restrict ourselves to deal with a single superselection sector) so that it is natural to assume that the center of $\cL_{\gR}(\sH)$ is trivial. This is equivalent to say that there are no orthogonal projectors in the center $\gZ_\gR$ of $\gR$ different form  $I$ and $0$.  A further reasonable  requirement for an elementary system is that $\gR$ is irreducible, that is $\gR'$ does not contain non-trivial orthogonal projectors: these projectors could be interpreted as elementary observables of another external system, 
whereas we would like that our elementary system be the complete system we are dealing with.
 Supposing that the center of $\gR$ is trivial, the request that $\gR$ is irreducible  may also be justified by assuming that $\gR$ includes a so-called {\bf  maximal set of compatible  observables} $\cA$ \cite{BeCa} as it happens in several concrete examples for systems described in complex Hilbert spaces. A maximal set of compatible  observables, by definition, is a subset of self-adjoint operators $\cA \subset \gR$ such that (a) the  elements  of $\cA$ are pairwise compatible and (b) if $T \in \gB(\sH)$ is 
selfadjoint and commutes with each element of $\cA$,  then $T$ is a {\em function} of the observables of $\cA$ in the sense that $T \in \cA''$. Under these hypotheses  $\gR$ is irreducible. Indeed, consider a closed subspace invariant under $\gR$. Its orthogonal projector $P$ therefore satisfies $P\in \gR'$ and thus $P \in \cA'$ in particular.  The hypothesis on $\cA$ yields  $P \in \cA'' \subset \gR'' =  \gR$.  We have proved that $P \in \gR \cap \gR'$ which we have assumed to  be  trivial,  so that $P=0$ or $P=I$and thus $\gR$ is irreducible.

\noindent   All that gives rise to the following definition.
\begin{definition}\label{generalelementarysystem}
{\em A real  {\bf  elementary system} is an irreducible von Neumann algebra $\gR$  over the real separable Hilbert space $\sH$.}
\end{definition}

\begin{remark}$\null$\\ {\em
If $\sH$ were complex in the definition above, we would find the  trivial result $\gR' = \{cI\:|\: c \in \bC\}$ as we know from Proposition \ref{SL}, so that $\gR= \gB(\sH)$ necessarily.}\end{remark}
\noindent Things dramatically change when  $\sH$ is real, for two main reasons. The first difference regards the role of the  lattice $\cL_\gR(\sH)$ represented  by the elementary observables of the system.
Focussing on the  von Neumann algebra $\gO:=\cL_\gR(\sH)''$, we know from Th.\ref{teopropvnA} (e) that, unlike the complex case, $\gO$ is a proper subalgebra of $\gR$ in general.  Nevertheless,  a direct check shows that $\cL_\gO(\sH)=\cL_\gR(\sH)$ so that  the same  lattice of propositions  is shared by two different von Neumann algebras, one properly contained in the other. However, differently from $\gR$, its subalgebra  $\gO$ is not necessarily irreducible and it 
may not represent an elementary system according to our definition.  We will not address this issue any further here sticking to Def. \ref{generalelementarysystem} but  leaving open the possibility  that $\gR\setminus\gO$ contains some relevant operators for the description of the system. 
The second important difference from the complex case  concerns the commutant of the irreducible algebra $\gR$ which, in the real case, may have  three different  forms, as the following result clarifies.
\begin{theorem}\label{threecommutant} Let $\gR$ be a von Neumann algebra on the real Hilbert space $\sH$. If $\gR$ is irreducible, then the following facts hold.\\

\noindent {\bf (a)} $\gR'$ is of three possible mutually exclusive types listed below.

(i)  $\gR'= \{a I\:|\: a\in \bR\}$ ({\bf real-real type}).\\

 (ii)   $\gR'= \{aI + bJ\:|\:a,b \in \bR\}$  where $J$ is a complex structure determined up to its sign. Furthermore $J \in \gR$  ({\bf real-complex type}).\\

 (iii)   $\gR' = \{aI + bJK + cJ + dK\:|\:a,b,c, d \in \bR\}$ 
where $J,K$ and  $JK=-KJ$ are complex structures. Furthermore $J,K, JK \not \in \gR$ ({\bf real-quaternionic type}).\\

\noindent {\bf (b)} Correspondingly, $\gR$, $\gZ_\gR$, and $\cL_\gR(\sH)$ are of three possible mutually exclusive types:

(i)  $\gR = \gB(\sH)$,  $\gZ_\gR = \{aI\:|\: a \in \bR\}$ and $\cL_\gR(\sH)= \cL(\sH)$ ({\bf real-real type}).\\

(ii) $\gR = \gB(\sH_J)$, $\gZ_\gR = \gR'=\{a I+BJ\:|\: a,b\in \bR\}$ and $\cL_\gR(\sH)= \cL(\sH_J)$ ({\bf real-complex type}). \\

(iii) $\gR = \gB(\sH_{J,K})$,  $\gZ_\gR = \{aI\:|\: a \in \bR\}$  and $\cL_\gR(\sH)= \cL(\sH_{J,K})$ ({\bf real-quaternionic type}).
\end{theorem}

\begin{proof} (a) Let $A \in \gR'$.
Dealing with as in the proof of (i) in Proposition \ref{SL2}, we have that 
 $A= aI+bL$ for some $a,b \in \bR$ and some complex structure $L$ depending on the element $A$. $\gR'$ is an  real associative unital normed algebra with the further property that $||AB||= ||A||\:||B||$. Indeed, by direct computation $||(aI+ bL)x||^2= (a^2+b^2) ||x||^2$ so that $||aI+bL||^2= a^2+b^2$. Furthermore, iterating the procedure, where $L'$ is another complex structure,  $||(aI+bL)(a'I+b'L')x||^2 = (a^2+b^2)(a'^2+b'^2) ||x||^2 = ||aI+bL||^2\:||a'I+b'L'||^2||x||^2$ and thus
$||(aI+bL)(a'I+b'L')||= ||aI+bL||\:||a'I+b'L'||$. Next, a known result \cite{UW} establishes that, as $\gR'$ is a real associative unital normed algebra 
where $||AB||= ||A||\:||B||$, there must exist a  real associative unital  normed algebra isomorphism $h$ from  $\gR'$ to $\bR$, $\bC$ or the algebra of quaternions $\bH$. 
In the first case it simply holds $\gR'=h^{-1}(\bR)=\{aI\:|\:a\in\bR\}$.
In the second case $\gR' =h^{-1}(\bC) = \{aI + bJ\:|\:a,b \in \bR\}$ where 
$J := h^{-1}(i)$. As $h^{-1}$ is an isomorphism $JJ= h^{-1}(jj) = h^{-1}(-1) =-I$. In the third case $\gR' = h^{-1}(\bH) = \{aI + bJ + cK + dJK\:|\:a,b,c, d \in \bR\}$ with
$J := h^{-1}(j)$,  $K := h^{-1}(k)$,  $JK := h^{-1}(i)$ where  $i,j,k \in \bH$ (with $i=jk=-kj$) are the three imaginary quaternionic units. Again, as in the real-complex case, we get $JJ= h^{-1}(jj) = h^{-1}(-1) =-I$
and $KK= h^{-1}(kk)= h^{-1}(-1) =-I$. 
Let us prove that $J$ in the real-complex case and $J,K$ in the real-quaternionic one are anti selfadjoint concluding that they are complex structures. The proof is the same for all of them, so take $J$. Since $\gR'$ is a $\null^*$-algebra, it holds $J^*\in\gR'$, in particular $J^*J\in\gR$ which is clearly self-adjoint and positive. Being $\gR$ irreducible, Lemma \ref{SL} guarantees that $J^*J=aI$ for some $a\ge 0$. Multiplying both sides on the right by $-J$, keeping in mind that $JJ=-I$, we get $J^*=-aJ$. Now, if we take the adjoint on both sides we obtain $J=-aJ^*$ which in particular assures that $a\neq 0$, $J$ being a unitary operator. So, $J^*=-\frac{1}{a}J$. Putting all together we have $0=J^*-J^*=\left(a-\frac{1}{a}\right)J$. Again, since $J$ is unitary it must be $a=\frac{1}{a}$, hence $a=1$, concluding the proof of the anti-selfadjointness of $J$.
$JK$ turns out to be a complex structure since $J$ and $K$ are complex structures and $JK=-KJ$.
In the complex case, since $J$ commutes with $\gR'$, it must  belong to $\gR''=\gR$. 
In the quaternionic case, if $J \in \gR$ we would have $J K=KJ$ which is impossible since we know that $JK = -KJ$ and $JK\neq 0$.
The same arguments applies for $K$ and $JK$.
 If, in the real-complex case, $J'$ is another complex structure in $\gR'$, 
it commutes with $J$ (as it belongs to $\gR$). Therefore $JJ' \in \gR'$  is self adjoint and thus $JJ' = aI$, namely $J'= -aJ$, because $\gR$ is irreducible. Since $JJ=J'J' =-1$ we must have $a=\pm 1$.\\
(b) In the first case 
$\gR' = h^{-1}(\bR) = \{a I\:|\: a\in \bR\}$ and thus $\gR=\gR''=  \{a I\:|\: a\in \bR\}' = \gB(\sH)$. 
$\cL_\gR(\sH)= \cL(\sH)$ follows trivially.
In the second case, $A \in \gR$ if and only if $[A, J]=0$, which is the same as saying that $A: \sH \to \sH$ is $\bC$-linear 
that is an operator on  $\sH_J$. Since $||x||_J = ||x||$ for $x\in \sH$ ( $=\sH_J$ as a set) we have $||A||_{\sH}=||A||_{\sH_J}$. Therefore $\gR = \gB(\sH_J)$. 
If $P \in \gB(\sH)$ commutes with $J$, it holds $P=P^*$ with respect to $(\cdot|\cdot)$ if and only if it happens with respect to $( \cdot|\cdot)_J$. Therefore $P$ is an orthogonal projector of $\gB(\sH)$ commuting with $J$
if and only if it is an orthogonal projector of $\gB(\sH_J)$. Thus $\cL_\gR(\sH)= \cL(\sH_J)$. 
The proof of  $\gR = \gB(\sH_{J,K})$ for the third case is very similar to that for the real-complex case. The remaining statements  are then trivial.
\end{proof}

\begin{remark}\label{remcQ}$\null$ \\
{\em  {\bf (a)}  Consider  a {\em complex} Hilbert space $\sH$ with scalar product $\langle\cdot|\cdot\rangle$.
$\gB(\sH)$  can always be viewed as a {\em real elementary system} in the {\em real-complex case} over a suitable real Hilbert space $\sH_\bR$.  This is obtained  by defining 
 $\sH_\bR=\sH$ with $\langle x|y \rangle_\bR := Re \langle x|y\rangle$ for all $x,y \in \sH$ equipped with the complex structure $J :  \sH_\bR \ni x \mapsto  ix \in \sH_\bR$ viewed as $\bR$-linear operator. With these choices  $(\sH_\bR)_J=\sH$,
and $\gB(H) = \{aI + bJ\:|\: a,b \in \bR\}'$ the commutant being that defined in $\gB(\sH_\bR)$.   $\gB(H)$ is irreducible in $\sH_\bR$ because $J\in \gB(\sH)$ and thus every  orthogonal projector $P\in \gB(\sH_\bR)$ commuting with $\gB(H)$ commutes with $J$
and thus $P$ is a {\em complex} orthogonal projector  in $(\sH_\bR)_J =\sH$ commuting with $\gB(H)$ so that $P=0,I$.
Finally it  $\gB(H)' = \{aI + bJ\:|\: a,b \in \bR\}''= \{aI + bJ\:|\: a,b \in \bR\}$.
We conclude that  $\gB(\sH)$ is a real elementary system over $\sH_\bR$ in the real-complex case. } \\
{\em  {\bf (b)}  Consider  a {\em quaternionic} Hilbert space $\sK$ with scalar product $\langle\cdot|\cdot\rangle$. $\gB(\sK)$ can always be viewed as a {\em real elementary system} in the {\em real-quaternionic case}  over a suitable real Hilbert space $\sK_\bR$.  This is obtained  by defining 
$\sK_\bR=\sK$ with $\langle x|y \rangle_\bR := Re \langle x|y\rangle$ for all $x,y \in \sK$ equipped with the three  complex structures   $J :  \sK_\bR \ni x \mapsto  xj \in \sK_\bR$ and 
 $K :  \sK_\bR \ni x \mapsto  xk \in \sH_\bR$,
viewed as $\bR$-linear operator.   With these choices  $(\sK_\bR)_{JK}=\sK$ 
and one finds by direct inspection that  $\gB(\sK) = \{aI + bJK + cJ + dK\:|\: a,b,c,d \in \bR\}'$ the commutant being that defined in $\gB(\sK_\bR)$. 
Since  $\gB(\sK)$ is the commutant of a set of operators in $\gB(\sK_\bR)$ which is $^*$-closed, it is automatically  a real von Neumann algebra on $\sK_\bR$.
It is also irreducible since every orthogonal projector $P \in \gB(\sK_\bR)$ commuting with $\gB(\sK)$
is of the form $P= aI + bJK + cJ + dK$ so that $P=P^*$ implies   $aI + bJK + cJ + dK =aI - bJK - cJ - dK$. Thus $P=aI$
from $P=\frac{1}{2}(P+P^*)$,
 with $a=0,1$ since $PP=P$.
We conclude that  $\gB(\sK)$ is a real elementary system over $\sK_\bR$ in the real-quaternionic case because   $\gB(\sK)' = \{aI + bJK + cJ + dK\:|\: a,b,c,d \in \bR\}'' \supset  \{aI + bJK + cJ + dK\:|\: a,b,c,d \in \bR\}$ and  thus 
$\gB(\sK)'= \{aI + bJK + cJ + dK\:|\: a,b,c,d \in \bR\}$ by Theorem \ref{threecommutant}. } 
\end{remark}
\noindent Theorem \ref{threecommutant} with the help of important achievements in \cite{V2},  permits us to characterize  {\em states} and the  {\em symmetries} of an elementary system adopting  the general version of these notions as presented in Sect.\ref{SecI1}. The definition and some properties of trace-class operators are listed in Appendix \ref{tracelass}.

\begin{proposition}\label{statandsym} Consider an elementary system described by the irreducible von Neumann algebra $\gR$ over the separable  real Hilbert space $\sH$. The following assertions are true.\\

\noindent {\bf (a)} Assuming that  $\dim(\sH) \neq  2$, if  $\mu : \cL_\gR(\sH) \to [0,1]$ is a $\sigma$-additive probability measure -- i.e., a state -- there is a unique selfadjoint positive unit-trace trace-class operator $T\in \gR$
 such that  $$tr(T P)_\sH= \mu(P)\quad  \mbox{for every}\quad  P \in \cL_\gR(\sH)\:.$$
Every selfadjoint positive unit-trace trace-class operator $T\in \gR$  defines a state by means of the same relation.\\

\noindent {\bf (b)} If $h : \cL_\gR(\sH) \to \cL_\gR(\sH)$ is a lattice automorphism
-- i.e., a symmetry -- there is a unitary operator $U: \sH \to \sH$ 
 such that \beq h(P) = UPU^{-1}\quad  \mbox{for every}\quad P \in \cL_\gR(\sH)
\:,\label{hU}\eeq
and the following facts are true.

(i) In both the real-real and real-quaternionic case, $U \in \gR$.

(ii) In the real-complex case, $U$ may either commute with $J$ (thus $U \in \gR$) or anticommute with $J$ (thus $U \not \in \gR$ and  $U^2 \in \gR$)

(iii) Every  unitary operator $U$ that satisfies (i) or (ii), depending on the case, defines a  symmetry by means of (\ref{hU}) and  another  unitary operator $U'$ of the same kind satisfies (\ref{hU}) in place of $U$ for the same $h$ if and only if $U'U^{-1}\in \gZ_\gR$. 
\end{proposition}
\begin{proof} 
We leave to the reader the proof of the following  elementary result relying on the content of Appendix \ref{tracelass}. $A \in \gB(\sH)$ is  a selfadjoint positive  trace-class operator commuting with $J$ in the real-complex case or with $J$ and $K$ in the real-quaternionic case,
  if and only if $A$ is a selfadjoint positive  trace-class operator, respectively, in  $\gB(\sH_J)$ or $\gB(\sH_{J,K})$.  Under these hypotheses and
with obvious notation,
\begin{align}
tr(A)_{\sH_J} &= \frac{1}{2} tr(A)_{\sH} \quad \mbox{in the real-complex case,}\label{tr1}\\
tr(A)_{\sH_{J,K}} &= \frac{1}{4} tr(A)_{\sH} \quad \mbox{in the real-quaternionic case.}\label{tr2}
\end{align}
(a)
Due to Theorem \ref{threecommutant} (b), a $\sigma$-additive probability measure over $\cL_\gR(\sH)$ is a   $\sigma$-additive measure over, respectively,  $\cL(\sH)$,  $\cL(\sH_J)$ or $\cL(\sH_{J,K})$. Here we can apply the Gleason-Varadarajan 
theorem (Theorem 4.23 in \cite{V2}) proving that there is a unique unit-trace trace-class selfadjoint positive operator $T_0$ in, respectively, $\gB(\sH)$,  $\gB(\sH_J)$ or $\gB(\sH_{J,K})$ such that 
$\mu(P) = tr(T_0P)$ with $P$ in, respectively,  $\cL(\sH)$,  $\cL(\sH_J)$ or $\cL(\sH_{J,K})$. As seen above such $T_0$ can be viewed as  a   trace-class selfadjoint positive operator  in $\gB(\sH)$. This operator has trace $1$ in the first case,
has trace $2$ and commutes with $J$ in the second case, and has trace $4$ and commutes with  $J$, $K$ and $JK$ in the third case. 
In the first case $T:=T_0$ fulfills all requirements. Let us consider the second case.
 Since $tr(TP)_{\sH_J}= tr(TPP)_{\sH_J}= tr(PTP)_{\sH_J}$, where $PTP$ is trace class, positive and selfadjoint if $T$ is, we can use (\ref{tr1}) and $T:= \frac{1}{2}T_0$ fulfills all requirements. In the third case  we can similarly 
 use (\ref{tr2}) obtaining that $T:= \frac{1}{4}T_0$ fulfills all requirements.  
By direct inspection one sees that a positive selfadjoint trace class unit-trace operator $T$ which, in the first case commutes with $J$ and in the second case commutes with $J$, $K$ (and $JK$), defines 
a $\sigma$-additive probability measure over $\cL_\gR(\sH)$. That $T$ is uniquely fixed by the afore-mentioned properties as a consequence of Gleason-Varadarajan 
theorem  because it fulfills the  requirements when $T$ (multiplied by $2$ or $4$, depending on the case) is viewed as an operator in, respectively, $\gB(\sH)$,  $\gB(\sH_J)$ or $\gB(\sH_{J,K})$.\\
(b) With the same strategy adopted to prove (a),  i.e., sticking to $\sH$ in the real-real case, or passing to describe everything from $\sH$ to $\sH_J$ or $\sH_{J,K}$ for, respectively, the real-complex or the real-quaternionic case,  points (i),(22) easily arise from Theorems 4.27, 4.28 in \cite{V2}.  
Notice that the quaternionic generalization of Wigner Theorem presented in Theorem 4.27 in \cite{V2} (see Remark \ref{remquatV2} for the conventions used in \cite{V2}) gives an apparently more general result: in the real-quaternionic case the symmetry $h$ is represented by $V\cdot V$ where $V:\sH_{J,K}\rightarrow \sH_{J,K}$ is an additive and bijective function such that $V(\psi p)=V(\psi)q^{-1}pq$ and $(V\phi|V\psi)_{J,K}=q^{-1}(\phi|\psi)q$ for some $q\in\bH$ with $|q|=1$ depending only on the function. Moreover another map $W$ like $V$ generates $h$ if and only if $W\psi=(V\psi)a$ for every vector $\psi$ and a constant $a\in \bH$ with $|a|=1$. Corollary 4.28 in \cite{V2} shows how, by taking $a :=q^{-1}$ we can always find a \textit{linear and unitary} representative of $h$. In point (iii) of \textbf{(b)} we deal only with quaternionic linear and unitary operators representatives $U,U'$, hence they must be equal up to a \textit{real} number $a$ with $|a|=1$ since $a=U'U^{-1}$ must be quaternionic linear as well.
\end{proof}

\begin{remark}
{\em Let us focus on the case of $\sH$  complex instead of real and equipped with a complex irreducible von Neumann algebra $\gR$. The statement (a) above holds true as it stands for $\dim(\sH) \neq 2$.
In fact, one has $\gR=\gB(\sH)$ and $\cL_\gR(\sH)= \cL(\sH)$ and the statement is nothing but the statement of Gleason's theorem (Theorem 4.23 in \cite{V2} specialized to complex Hilbert spaces). 
The statement concerning (\ref{hU}) in (b) is now nothing but a version of the standard Wigner-Kadison theorem (e.g., see \cite{M}) for complex Hilbert spaces and it is true with $U$ either unitary or anti unitary depending on $h$. Every  unitary (anti unitary) operator $U$ defines a  symmetry by means of (\ref{hU}) and  another  unitary (respectively anti unitary) operator $U'$ satisfies (\ref{hU}) in place of $U$ for the same $h$ if and only if $U'U^{-1}= e^{ia}I$ for some $a \in \bR$.}
\end{remark}

\subsection{A more physical notion of  elementary relativistic system} We are in a position to state a physically more precise version of the notion of a real  elementary {\em relativistic} system, making use of the general notions introduced in Sec.\ref{sec1} and relying upon two ideas. First, an elementary {\em relativistic} system must be elementary according   to Def.\ref{generalelementarysystem} so it includes an irreducible von Neumann algebra $\gR$. Secondly, it has to support a representation of  Poincar\'e group $\cP$ viewed as   \textit{maximal symmetry group} of the system. In line with (6) of Sec.\ref{seclist},  this representation  is realized in terms of automorphisms of the lattice of projectors. We therefore assume the existence of  a 
locally faithful
 group representation $h:\cP\ni g\mapsto h_g\in \mbox{Aut}(\cL_\gR(\sH))$.
The demand of {\em elementariness} of the system is completed by  further specific requirements  on $h$.  On the one side, since  the system  has to be regarded as a \textit{realisation} of the physical symmetries, $h$ must contain all information about observables of the system. Since observables are described by operators, this idea can be implemented by picturing $h$ in terms of unitary operators in accordance with Proposition \ref{statandsym} and exploiting their products, linear combinations and weak limits to get the PVMs of the self-adjoint operators of $\gR$. 
On the other side, the demand of elementariness must also involve some  irreducibility property of $h$. We therefore assume that no non-trivial sublattices of $\cL_\gR(\sH)$ are left fixed under  $h$. If it  were the case, the observables constructed out of the elements of the sublattice could be viewed  as describing a Poincar\'e invariant  subpart of the overall  system, against the idea of elementariness.
Finally, it is  reasonable to lift the continuous nature of $\cP$ to its representation in a weak operational way, using the natural 
(seminormed) topology induced by the   probability measures representing quantum states.

\begin{definition}\label{RRES}
A real  {\bf  relativistic elementary system} (real {\bf RES}) is a real elementary system  $\gR$  over the real  separable Hilbert space $\sH$ equipped with a 
 representation of  Poincar\'e group $h: \cP \ni g \mapsto h_g \in \mbox{\em Aut}(\cL_\gR(\sH))$ which is
 locally faithfull ($\cP$ is injectively represented by $h$ in a neighborood of the neutral element) and satisfies  the following requirements.\\
\noindent {\bf (a)}  $h$ is irreducible, in the sense that $h_g(P)=P$ for all $g\in \cP$ implies either $P=0$ or $P=I$.\\

\noindent {\bf (b)}  $h$ is continuous, in the sense  that the map $\cP \ni g \mapsto \mu(h_g(P))$ is continuous for every fixed $P \in \cL_\gR(\sH)$ and every fixed quantum state $\mu$.\\

\noindent {\bf (c)} $h$ defines the observables  of the system. That is,  according to Proposition \ref{statandsym} (b) and representing $h$ in terms of unitary operators $U_g\in \gR$ (defined up to unitary factors in $\gZ_\gR$),
$h_g(P) = U_gPU_g^{-1}$ for $g \in \cP$ and $P \in \cL_\gR(\sH)$,
it must be
 $$\left(\{U_g\:| g \in \cP\:\} \cup \gZ_\gR \right)'' \supset \cL_\gR(\sH)\:.$$
\end{definition}
\begin{remark}\label{remarkRRES}$\null$\\
{\em {\bf (a)} Suppose that such a $P\neq 0,I$ exists and take the set $\cL_P:=\{Q\in\cL_\gR(\sH),\:|\: Q\le P\}$. It is easy to see that this is a lattice ($\le_P:=\le$) which is still complete, orthomodular ($Q^{\perp_P}:=Q^\perp\wedge P$) and $\sigma$-complete. Since $h_g$ is a lattice automorphism it preserves the order and so, if $Q\in\cL_P$, we have $h_g(Q)\le h_g(P)=P$, i.e. $h_g(Q)\in\cL_P$. $\cL_P$ is then a proper sublattice of $\cL_\gR(\sH)$ which is left invariant under the action of $h$, but we have assumed 
that  such sublattices do not exist.}\\
{\em {\bf (b)} We have explicitly 
assumed that $U_g \in \gR$ is always valid,
excluding the case of $U_g$ anticommuting with $J$  in the real-complex case (corresponding to an {\em anti unitary} operator in the complex Hilbert space $\sH_J$). The reason is the following. From the polar decomposition of $\cP = SL(2,\bC) \ltimes \bR^4$ one sees that every $g \in \cP$ can be always decomposed into a product of this kind $g= rrbb$ where $r$ is a spatial rotation and $b$ a boost. Using Proposition \ref{statandsym}(b)(iii), we have that $U_g = e^{cJ} U_r^2 U_b^2$ for some $c\in \bR$. It is now clear that, even if it were either $U_rJ=-JU_r$ or $U_bJ=-JU_b$ or both, then $U_g$ would commute with $J$ in any case.}\\
{\em {\bf (c)} The center $\gZ_\gR$ plays a  role in defining the observables as is evident in the requirement
$\left(\{U_g\:| g \in \cP\:\} \cup \gZ_\gR \right)'' \supset \cL_\gR(\sH)$
and  this is relevant only in the real-complex case. In the other two cases $\gZ_\gR$ is trivial and can be omitted. The utmost reason for the appearance of $\gZ_\gR $ in the formula above is that any particular representative $U_g$ 
of $g \in \cP$ has no meaning in its own right as the physical content is owned by $h_g$, that is by "$U_g$ up to phases", i.e., elements of $\gZ_\gR$ in the meaning of Prop. \ref{statandsym} (b) (iii). 
Of course it holds $\left(\{U_g\:| g \in \cP\:\} \cup \gZ_\gR \right)''=\left(\{V_g\:| g \in \cP\:\} \cup \gZ_\gR \right)''$ for  two different choices of representatives $U_g,V_g$ of $h_g$, when $g \in \cP$. }\\
{\em {\bf (d)} If $\sH$ were complex, the definition above could  be restated as it stands and it would reduce to the already given definition of complex WRES.  Indeed, since the only complex irreducible von Neumann algebra is  $\gB(\sH)$ itself,
we would find $\gR= \gB(\sH)$. Moreover, the representation  $h$ would be implemented by a
 locally-faithful
  irreducible unitary  representation $\cP \ni g \mapsto U_g \in \gB(\sH)$  as a consequence of the famous theorem by Wigner on continuous symmetries and a celebrated result by Bargmann we shall exploit to prove Proposition \ref{bargmann} below. More strongly, on a complex Hilbert space, WRES and RES are \textit{equivalent} definitions since every WRES gives rise to a RES in a trivial way.}\\
{\em {\bf (e)} If $\sH$ is real only one direction of the equivalence of (d) is trivially true: every real WRES gives rise to a real RES. The converse is far from  obvious. However it is true if an additional physical requirement is assumed. The proof of this fact (See proposition \ref{lastProposizione}) is the last result of this paper.}
\end{remark}

\noindent The  map $\cP\ni g\mapsto U_g$ introduced in  Definition \ref{RRES} (c) is not, in general, a group representation since we may have  $U_g U_h=\Omega(g,h)U_{gh}$ for  operators $\Omega(g,h)\in\gU(\gZ_\gR)$ where $\gU(\gZ_\gR)$ henceforth denotes  the set of unitary operators in the center of $\gR$. In particular $U_e= \Omega(e,e)$ putting $g=h=e$
in the identity above.
Such a map  $\cP\ni g\mapsto U_g$ is known as a {\bf projective unitary representation} of $\cP$ while the function $\Omega:\cP\times\cP\rightarrow \gU(\gZ_\gR)$ is said to be the {\bf multiplier function} of the representation. 
\begin{remark}
{\em The structure of $\gZ_\gR$ implies  the following algebraic identifications for a real relativistic system:   $\gU(\gZ_\gR)=\bZ_2 I$ -- the multiplicators are signs -- in the real-real and real-quaternionic cases and  $\gU(\gZ_\gR)=U(1) I$ -- the multiplicators are complex phases --  in the real-complex case.}
\end{remark}
\noindent The associativity property of the operator multiplication easily gives the {\bf cocycle}-property,
\begin{equation}\label{multprop}
\Omega(r,s)\Omega(rs,t)=\Omega(r,st)\Omega(s,t)\quad \mbox{for all $r,s,t\in\cP$}\:.
\end{equation} 
For  any function $\chi:\cP\rightarrow \gU(\gZ_\gR)$ the map $\cP \ni g\mapsto \chi(g)U_g$ is still a projective representation associated with the same representation $h$ of $\cP \ni g\mapsto U_g$, whose multiplier is now given by $$\Omega_\chi(g,h)=\chi(g)\chi(h)\chi(gh)^{-1}\Omega(g,h)\quad \mbox{for all $g,h \in \cP$.}$$
A natural question then concerns the possibility of getting rid of the multipliers by finding a function $\chi$ such that $\Omega_\chi=I$ in order to end up with a proper unitary representation from a given projective unitary representation.
A positive answer can be given for all of the three cases. 

\begin{proposition}\label{bargmann}
Let $\gR$ and $h$ respectively  be the von Neumann algebra and  the Poincar\'{e} representation  of a real RES
as in Definition \ref{RRES}. The following facts hold.\\

\noindent {\bf (a)} There exists a 
 locally-faithful 
 strongly-continuous unitary representation $\mathcal{P}\ni g\mapsto U_g\in\gR$  on $\sH$ such that $h_g(P)=U_g P U_g^{-1}$ for every $g\in\mathcal{P}$ and every $P \in \cL_\gR(\sH)$.\\

\noindent{\bf (b)} $\mathcal{P}\ni g\mapsto U_g\in\gR$   is irreducible respectively  on $\sH$, $\sH_J$ or $\sH_{J,K}$ according to the three cases of Proposition\ref{statandsym},
\end{proposition}
\begin{proof}
\noindent We simultaneously  prove (a) and (b). We already know that $h_g(\cdot) =V_g\cdot V_g^*$ for some unitary operator $V_g\in\gR$. By the continuity hypothesis on $h_g$ and (a) of Proposition\ref{statandsym}, we see that the maps $\cP \ni g\mapsto tr(P_\psi h_g(P_\phi))_\sH=|(\psi|V_g \phi)|^2$, $\cP \ni g\mapsto tr(P_\psi h_g(P_\phi))_{\sH_J}=|(\psi|V_g \phi)_J|^2$ and $\cP \ni g\mapsto tr(P_\psi h_g(P_\phi))_{\sH_{J,K}}=|(\psi|V_g \phi)_{J,K}|^2$ respectively, are  continuous for every $\psi,\phi\in \sH,\sH_J,\sH_{JK}$.
Let us focus on the real-complex case first. Thanks to the above remark, following  the analysis contained in the well-known paper \cite{Ba}, we  get a strongly-continuous unitary representation $\cP \ni g\mapsto U_g$ on $\sH_J$ such that $U_g=\chi_g V_g$ for some $\chi_g\in U(1)$, hence generating $h$. 
This unitary representation is 
locally faithful because $h$ is locally faithful. (If $U_g = U_f$ with $g,f$ in the neighborhood of the neutral element where $h$ is injecive, we have $h_g = U_g\cdot U_g^*=
 U_f\cdot U_f^* =h_f$, so that  $f=g$.) 
Notice that since  $\sH= \sH_J$ as a set and $||\cdot||_J = ||\cdot||$, $\cP \ni g\mapsto U_g$ is also strongly continuous on $\sH$.
  Irreducibility  on $\sH_J$ follows from the following argument. 
Since the family $\gU:=\{U_g,\:|\:g\in\cP\}$ is closed under the adjoint operation, thanks to Remark \ref{remirred} (a) we need only to prove that $\gU'\cap\cL(\sH_J)=\{0,I\}$, but this is a direct consequence of the irreducibility of $h$. Indeed, if $P$ is a complex projector commuting with every $U_g$ then $h_g(P)=U_gPU_g^*=P$ for every $g\in\cP$ and thus  $P=0$ or $P=I$.
Let us next focus on the real-quaternionic  case. Thanks to the analysis of \cite{Em} we can always find a strongly-continuous unitary representation $\cP \ni g\mapsto U_g$ on $\sH_{J,K}$  such that $U_g=\chi_g V_g$ for some $\chi_g\in \bZ_2$, hence generating $h$. The same kind of arguments used in the real-complex case prove irreducibility and 
local faithfulness 
of the found unitary representation.\\
Let us conclude the proof discussing the real-real case.
We affirm that we may always choose an equivalent representative $\cP \ni g\mapsto U_g$ such that $U_e=I$, it is strongly continuous over an open neighborhood of the identity $A_e$ and its multiplier $(g,h)\mapsto \Omega(g,h)$ is continuous over $A'_e\times A'_e$ with $A'_e\subset A_e$,  a  smaller open neighborhood of $e$ which can always be assumed to be connected ($\cP$ is a Lie group and as such it is locally connected).
The proof of this fact can be found within the proof of Proposition 12.38 in \cite{M} which is valid both for complex and real Hilbert spaces since there is no distinctive role played by the imaginary unit. 
 Since $\Omega(g,h) \in \{\pm I\}$ which is not connected if equipped with the topology induced by $\bR$ and $\Omega(e,e)=I$, the continuity of $\Omega$ guarantees that $\Omega(g,h)=I$ for every $g,h\in A'_e$. In other words  $U_gU_h=U_{gh}$ for every $g,h\in A'_e$. 
As the  group $\gU(\sH)$ of all unitary operators over $\sH$ is a topological group with respect to the strong operator topology, the continuous function $\cP \ni g\mapsto U_g$ is then a local topological-group homomorphism as in Definition B, Chapter 8, Par.47 of \cite{P}. Since, as established in \cite{connect}, $\gU(\sH)$  is connected if $\dim \sH$ is not finite and $\cP$
is a simply connected Lie group, we can apply  Theorem 63 \cite{P} proving that there exists a strongly-continuous unitary representation $\cP \ni g\mapsto W_g \in \gU(\sH)$ such that $W_g=U_g$ on some open neighborhood of the identity $A_e''\subset A_e'$. 
If $\dim(\sH) = n < +\infty$, then $\gU(\sH)$ can be identified to the topological group $O(n)$. Its open subgroup $SO(n)$ is the connected component including the identity element $I$. In this situation, we can restrict ourselves to deal with a smaller initial open set $A'_e \cap B$ where $B$ is the pre-image through the map $U$ (which is continuous on $A'_e$) of an open set including $I$ and completely included in $SO(n)$. As $SO(n)$ is connected, we can finally exploit the same procedure as in the infinite dimensional case,  proving that there exists a strongly-continuous unitary representation $\cP \ni g\mapsto W_g \in \gU(\sH)$ such that $W_g=U_g$ on some open neighborhood of the identity element $A_e''\subset A_e'\cap B$.
To conclude, we observe that since  the Lie group $\cP$ is connected, a standard result guarantees that every $g\in\cP$ can be written as $g=g_1\cdots g_n$ for some $g_1,\dots,g_n\in A_e''$. So, $W_g=W_{g_1}\cdots W_{g_n}=U_{g_1}\cdots U_{g_n}$ and 
$h_g = h_{g_1}\circ \cdots \circ h_{g_n}$,
from which it easily follows $h_g=W_g\cdot W_g^*$ for every $g\in\cP$. Local faithfulness and irreducibility of the representation  $\cP \ni g\mapsto W_g$ arises form the same properties of $h$ as in the other two cases.
\end{proof}

\subsection{Emergence of an (up to sign unique) complex structure from Poincar\'e symmetry}
We are in a position to state and prove our second main result of this work, establishing that, even relying on the more accurate definition of relativistic elementary system as in Definition \ref{RRES}, when assuming a standard hypothesis that physically means that the squared mass of the particle is non-negative, 
one finally achieves a {\em complex} Wigner elementary relativistic  system 
which is equivalent to our relativistic elementary system.  Again the initial real theory can be naturally rephrased into a better {\em complex} theory. In particular, now the lattice of elementary observables 
coincides with the whole lattice of orthogonal projectors of the complex Hilbert space in agreement with the picture of the thesis of Sol\`er's theorem, even if we started from different hypotheses. Our second result is however more refined than our first achievement   because it studies the interplay of the final complex structure due to relativistic invariance arising  in  Theorem \ref{poinccomplexstructure}
 with the complex structures of the classification in theorem \ref{threecommutant}.

\begin{theorem}\label{main2} Consider a real relativistic elementary system defined by a real von Neumann algebra $\gR$ over  the separable real Hilbert space $\sH$ and a representation 
$\cP \ni g \mapsto h_g \in Aut(\cL_\gR(\sH))$. Let  $U: \cP \ni g \mapsto U_g \in \gR$ be a corresponding 
 locally-faithful
 strongly-continuous unitary representation of $\cP$ on $\sH$ as in Proposition \ref{bargmann}(a). If the associated operator $M^2_U$ (\ref{Moperator}) satisfies $M^2_U \geq 0$, the following facts hold.\\

\noindent {\bf (a)} $\gR$ is of  real-complex type with preferred complex structure $J \in \gR'$ defined up to sign.\\

\noindent  {\bf (b)} $U: \cP \ni g \mapsto U_g \in \gB(\sH_J)$ is irreducible over $\sH_J$ and defines a complex  WRES  
which is equivalent to the  real RES:

(i) $h_g(P) = U_gPU^{-1}_g$ for every $P\in \cL_\gR(\sH)$ and  $g \in \cP$

(ii) $\gR =\gR_U(\sH_J)$,\\
In particular $\gR = \gB(\sH_J)$ and $\cL_\gR(\sH) = \cL(\sH_J)$ in agreement with the thesis {\em Sth} of Sol\`er's theorem.\\

\noindent  {\bf (c)}   $J$  in (a)  is  Poincar\'e invariant and  coincides up to the sign with the unitary factor  of the polar decomposition of the anti self adjoint generator of the  temporal translations $\bR \ni t \mapsto U_{\exp(t {\bf p}_0)}$.
\end{theorem}
\begin{proof} Since $\gR$ is irreducible we have  three mutually exclusive cases for $\gR'$, as discussed in Theorem \ref{threecommutant}. Let us start by supposing that $\gR'$ is of real-complex type so that, up to sign, there is a preferred complex structure $J \in \gR'$. \\
(b) Using $J$ to construct $\sH_J$, As given by Proposition \ref{bargmann} $\cP\ni g\mapsto U_g\in\gR$ is a complex 
locally-faithful
 irreducible strongly-continuous unitary representation, i.e. a complex WRES. By construction if satisfies item (i). Let us prove (ii). The complex von Neumann algebra generated by $U$, $\gR_U(\sH_J)$ must coincide with the whole $\gB(\sH_J)$ because $U$ is (complex-) irreducible so that $\gR_U(\sH_J) = \gR_U(\sH_J)'' = \{U_g \:|\: g \in \cP\}'' =  \{cI\:|\: c \in \bC\}' = \gB(\sH_J)$. On the other hand, we already know that $\gR=\gB(\sH_J)$, in particular $\cL_\gR(\sH)=\cL(\sH_J)$. The two descriptions are clearly equivalent.\\
(c) First, notice that the anti selfadjoint generators of $U_{\exp(t\V{p}_\mu)}$, with $\mu=0,1,2,3$ defined on $\sH$ and $\sH_J$ as well as the definition of $D_G^{(U)}$ do not depend on the scalar field (see Theorem \ref{DMtheorem} and Remark \ref{updown}), hence the same holds true for the symmetric operator $M_U^2$. In particular Prop.\ref{prop2} (9) guarantees that $M_U^2$ is positive also on $H_J$.
Now, let $J_1,|P_0|$ be the polar decomposition of the anti self adjoint generator of $U_{\exp(t\V{p}_0)}$, defined on $\sH$. It is easy to see that this couple satisfies (i)-(iv) of (b) Th.\ref{Polar} also with respect to $\sH_J$. 
Since all the hypotheses of Th.\ref{poinccomplexstructure} are satisfied for $\cP\ni g\mapsto U_g$ on $\sH_J$, point (g) gives $J_1=\pm iI=\pm J$.
The fact that $J$ is Poincar\'e invariant is evident since $U_gJU_g^{-1}=JU_gU_g^{-1}=J$.\\
The proof concludes by proving that in a real RES, $\gR'$ can be neither of  real-real type nor real-quaternionic type.
\begin{proposition} $\gR$ defining  a real RES cannot be of real-real type if $M^2_U \geq 0$.
\end{proposition}
\begin{proof}
Let us start by assuming that $\gR'$ is of real-real type  so that, by Theorem \ref{threecommutant}, $\gR= \gB(\sH)$ and $\cL_\gR(\sH)= \cL(\sH)$.
Thanks to Proposition \ref{bargmann} the RES $\gR$, $h$ defines a real a 
locally-faithful
 irreducible strongly-continuous unitary representation $\cP \ni g \mapsto U_g \in \gR$ over the real space $\sH$. Theorem \ref{poinccomplexstructure}(b) implies that there is a complex structure $J$ which commutes with the representation  $U$. On the other hand we have from the  definition of RES, using in particular the fact that $\gZ_\gR$ is trivial in the real-real case,
$
\cL(\sH) = \cL_\gR(\sH)\subset(\{U_g| g\in\cP\}\cup\gZ_\gR)''=\{U_g| g\in\cP\}''
$
from which $J\in \{U_g| g\in\cP\}'= \{U_g| g\in\cP\}''' \subset \cL(\sH)''' = \cL(\sH)'$. This is impossible because  $J\neq 0$ is  anti selfadjoint
whereas $\cL(\sH)'$ is made of selfadjoint operators due to the following lemma.
\begin{lemma}\label{lemmaIR}
Let $\sH$ be a real Hilbert space, then $\cL(\sH)'=\{aI \:|\: a \in \bR\}$\:.
\end{lemma}
\noindent The proof of the lemma above appears in Appendix \ref{appProof}.
\end{proof}
\begin{proposition}\label{quatimp2} $\gR$ defining a real RES cannot be of real-quaternionic type  if $M^2_U \geq 0$.
\end{proposition}
\begin{proof} Assume that $\gR'$ is of real-quaternionic type. 
 $\sH$ cannot have dimension $1$ as quaternionic Hilbert space. If it were the case,  the representation $U$ could be seen as a 
locally-faithful 
 unitary  representation on a $4$-dimensional {\em real} Hilbert space. Thus  $U$ would include a
 locally-faithful 
 unitary  representation $V$ of $SL(2,\bC)$ on $\bR^4$ and $V_\bC$ whould be a 
locally-faithful 
 unitary  representation of $SL(2,\bC)$ on the $2$-dimensional complex Hilbert space $(\bR^4)_\bC$. This is not possible since the continuous finite-dimensional  unitary  representations of $SL(2,\bC)$ are completely reducible and the irreducible ones are the trivial representation only \cite{knapp}. In other words, the initial representation $U$ would be the trivial representaion against the 
local faithfulness hypothesis.
To deal with the case of quaternionic dimension $> 1$
we need some technical results whose proofs appear in Appendix \ref{appProof}.
\begin{lemma}\label{quatimp1}
Suppose that $\gR$ is an irreducible  real von Neumann algebra over the real Hilbert space with $\gR'$ of real-quaternionic type and let $J\in\gR'$ be a complex structure as in Theorem \ref{threecommutant} (a)(iii).
If $A,B\in\gR$, then $A+JB=0$ if and only if $A=B=0$. Moreover $$\gR+J\gR:=\{A+JB\:|\: A,B\in\gR\}=\gB(\sH_J)\:.$$
\end{lemma}
\noindent Now, consider the 
locally-faithful
 strongly-continuous irreducible unitary representation $\cP\ni g\mapsto U_g\in\gB(\sH_{J,K})$ given by Proposition \ref{bargmann}. This is still 
 locally-faithful, 
strongly-continuous and  unitary  if viewed over  $\gB(\sH_J)$ instead of $\gB(\sH_{J,K})$. We affirm that it is also irreducible on $\sH_J$, let us prove it.
Let $P\in\cL(\sH_J)$ such that $[P,U_g]=0$ for every $g\in\cP$. Thanks to Lemma \ref{quatimp1} it must be $P=A+JB$ for some $A,B\in\gR$. Since $P$ is selfadjoint (adjoints of $P,A,B,J$ defined on $\sH$ or $\sH_J$ coincide) $0=P-P^*=(A-A^*)+J(B+B^*)$. Thanks again to Lemma \ref{quatimp1} this implies $A^*=A$ and $B^*=-B$. Next $PP=P$ gives $A^2-B^2=A$ and $AB+BA=B$. Now, notice that $0=[P,U_g]=[A,U_g]+J[B,U_g]$, hence $[A,U_g]=[B,U_g]=0$. 
Now, notice that it simultaneously hold that $A,B\in\{U_g,\:|\:g\in\cP\}'\subset\cL_\gR(\sH)'$ (thanks to (c) in the definition of RES) and that $A,B$ are quaternionic-linear. This means that $A,B$ belong to the set $\cL(\sH_{J,K})'\subset \gB(\sH_{J,K})$ which is trivial thanks to the following Lemma whose proof appears in Appendix \ref{appProof}. 
\begin{lemma}\label{commutanquaternionic}
Let $\sH$ be a quaternionic Hilbert space of dimension strictly greater than $1$, then $\cL(\sH)'=\{aI \:|\: a\in\bR\}$.
\end{lemma}
 
\noindent Hence $A=aI$ and $B=bI$ for some $a,b\in\bR$. Since $B$ turns out to be both anti selfadjoint and selfadjoint it must vanish and so, from $A^2-B^2=A$ it follows $a^2=a$, i.e. $a=0,1$, concluding the proof of  irreducibility
of $\cP \ni g \mapsto U_g \in \gB(\sH_J)$.\\

\noindent The found results implies that $\gR$ satifying the hypotheses of Theorem \ref{main2} cannot be of real-quaternionic type as we go to prove.
 Consider the one-parameter subgroup $\bR \ni t\mapsto \exp(t\textbf{p}_0) \in \cP$ of time-translations. By the complex version of Stone Theorem
we have  $U_{\exp(t\textbf{p}_0)}=e^{tP_0}$ for a unique anti-selfadjoint operator $P_0$ on $\sH_J$. Thanks to Theorem \ref{Polar}, there exists a unique pair of operators $V,P$
defining the polar decomposition of $P_0$
 on $\sH_J$. They are completely defined by the requirements that $P_0=VP$, $P$ is positive and selfadjoint and $V\in\gB(\sH_J)$ is isometric on $Ran(P)$ and vanishes on $Ker(P)$. 
The anti-selfadjoint generators of $U$ do not change if we consider $U$ over $\sH$ or $\sH_J$. The hypothesis $M^2_U \geq 0$ is therefore valid also when thinking of $U$ as a representation over  $\sH_J$. Since this strongly-continuous representation is 
locally-faithful, 
unitary and irreducible over the complex Hilbert space $\sH_J$, invoking Theorem \ref{poinccomplexstructure}(g) we conclude that $V=\pm iI=\pm J$. $P_0$ is of course also real linear and anti selfadjoint, as $V,P$ are. Moreover,  selfadjointness and positivity of $P$ still hold in $\sH$, and $V$ is still isometric on $Ran(P)$ and vanishing on $Ker(P)$ if understood as operators on $\sH$. Thanks again to Theorem\ref{Polar}, this implies that $P_0=VP$ is also the polar decomposition of $P_0$ in the  real Hilbert space $\sH$, and as already noticed in general, $U_{\exp(t\textbf{p}_0)}=e^{tP_0}$ is valid in $\sH$. Here the contradiction comes.
As $\gR$ is of real-quaternionic type, the complex structure $K\in \gR'$ commutes with every element $U_g$, in particular with $U_{\exp(t\textbf{p}_0)}=e^{tP_0}$ and thus Lemma \ref{LemmaCOMM}(a) yields  $KP_0\subset P_0K$ and so (b) dives $KV=VK$.
Since $V=\pm J$ we therefore have $KJ=JK$ in contradiction with $KJ=-KJ$.
\end{proof}
\noindent The proof of Theorem \ref{main2} is concluded.
\end{proof}

\noindent Once established that every real  relativistic elementary system can always be pictured in a complex Hilbert space in terms of a complex Wigner relativistic elementary system and this description is better than the real one for the reasons discussed above, it remains open  the theoretical question whether or not there exist {\em intrisically} complex Wigner relativistic elementary system. In other words, given a complex Wigner relativistic elementary system (remind that in the complex case WRES and RES are equivalent concepts) is it  always possible to interpret it as arising from a real relativistic elementary system?  The answer is positive and immediately proved
Indeed, suppose we have a complex Wigner relativistic elementary system $\cP\ni g\mapsto U_g\in\gB(\sH)$ on a complex Hilbert space $(\sH,\langle\cdot|\cdot\rangle)$. As discussed in Remark \ref{remcQ}(a), referring to the Hilbert space structure defined by $\sH_\bR:=\sH$ and $(\cdot|\cdot):=Re\langle\cdot|\cdot\rangle$, the set of complex linear operators $\gB(\sH)$ gives rise to an irreducible von Neumann algebra $\gR$ on $\sH_\bR$ with $\cL_\gR(\sH_\bR)=\cL(\sH)$. Defining $h:=U_g\cdot U_g^{-1}$, we finally get a real relativistic elementary system. Remembering that $J:=iI$ is a complex structure on $\sH_\bR$ such that $(\sH_\bR)_J=\sH$, trivially reversing the reasoning we have that $(\gR,h)$ is equivalent to the initial complex Wigner relativistic elementary system. \\
As already announced in Remark \ref{remarkRRES}(e) we establish another relevant result, showing that a real RES can actually be derived from an equivalent real WRES, if the usual positive-squared-mass condition holds true.
\begin{proposition}\label{lastProposizione}
With the hypotheses of Theorem \ref{main2}, $M_U^2 \geq 0$ in particular, the representation $U$ is also irreducible and $\gR=\gR_U$, so that  $U$ determines a real WRES equivalent to the real RES defined by $\gR$ and $U$.
\end{proposition}
\begin{proof}
As demostrated in Theorem \ref{main2}, $\gR$ is of real-complex type with preferred complex structure $J$ and the representation $U$ is complex irreducible on $\sH_J$. We intend to prove that $U$ is irreducible also on $\sH$. To this end, suppose that $P \in \gB(\sH)$  
is an orthogonal projector in $\sH$ and $U_gP=PU_g$ for every $g\in \cP$. We have to prove that $P=0$ or $P=I$. Consider the operator $P':= JP+PJ$, it is anti selfadjoint and commutes with $J$ so that $P' \in \gB(\sH_J)$. So, since $U$ is (complex) irreducible, $P' = \lambda I$ with $\lambda$ imaginary because $P'^*=-P'$. In other words, going back to the real Hilbert space $\sH$, it holds (1) $JP+PJ=aJ$ for some $a \in \bR$. 
We derive (2) $JP = -PJP + a JP$ and, taking the adjoint,
(3) $-PJ = PJP -aPJ$. (2) and (3) yield  $[J,P]= a[J,P]$. If $a\neq 1$ we must have 
$[J,P]=0$ and thus $P$ is complex linear and coincides with either $0$
or $I$ and the proof ends. Instead, if $a=1$, (1) reduces to  $JP=(I-P)J$ where necessarily $P \neq  0, I$ and therefore we have the  orthogonal decomposition of $\sH$ into proper real closed subspaces $\sH = \sH_P \oplus \sH_P^\perp$, where $\sH_P := P(\sH)$. Finally, $A: \sH_P \to \sH_P^\perp$, where $A:= J|_{\sH_P},$ turns out to be a bijective isometry.
Referring to the decomposition $\sH = \sH_P \oplus \sH_P^\perp$ we have, $J= A \oplus_I (-A^{-1})$ \footnote{Since $A$ and $A^{-1}$ swap the subspaces $\sH_P,\sH_P^\perp$, we define their $\oplus_I$ sum as $A\oplus_I (-A^{-1})(u,v):= (-A^{-1}v,Au)$. The symbol $\oplus$ denotes the standard direct sum of operators.}.
As $P$ commutes with $U$, both real subspaces $\sH_P$ and $\sH_P^\perp$ are invariant under $U$. It is easy to see that the $U_g$s are also surjective if restricted as operators on $\sH_P$ or $\sH_P^\perp$. This, together with the fact that they are isometric, gives rise to a couple of unitary representations $U_P,U_{P^\perp}$ of $\cP$ on, respectively, $\sH_P$ and $\sH_P^\perp$ such that $U = U_P\oplus U_{P^\perp}$. Moreover these representations are also irreducible, indeed if there is an orthogonal projector $Q\le P$ with  $Q\neq 0,P$ and commuting with $U$, we can repeat the construction obtainig $\sH = \sH_Q \oplus \sH_Q^\perp$ and 
 $J(\sH_Q)= \sH_Q^\perp$ bijectively and this is impossible because it would give $\sH_Q^\perp=J(\sH_Q)\subset J(\sH_P)=\sH_P^\perp$, hence $Q^\perp\le P^\perp$, i.e. $P\le Q$ which is impossible.
By construction, $AU_{P} = U_{P^\perp}A$ which implies that both representations are
 locally faithful since $U = U_P\oplus U_{P^\perp}$ is locally faithful and $A$ is a vector space isomorphism.
It is possible to prove, taking into account all the corresponding definitions, that $M^2_{U_P} \oplus M^2_{U_{P^\perp}} = M^2_U$. This gives $M^2_{U_P} \oplus M^2_{U_{P^\perp}}  \geq 0$ and therefore  both $M^2_{U_P}\geq 0$ and $M^2_{U_{P^\perp}}\geq 0$. We are in the hypotheses of  Theorem \ref{poinccomplexstructure} (a) for both representations $U_P$  and  $U_{P^\perp}$  which imples  that, up to sign, there are two complex structures   
$J_P$ and $J_{P^\perp}$ on the  real Hilbert spaces $\sH_P$ and $\sH_P^\perp$ commuting with $U_P$ and $U_{P^\perp}$ respectively, and  $AJ_PA^{-1}=J_{P^\perp}$. The last identity implies
$J  (J_P \oplus J_{P^\perp}) = (A \oplus_I (-A^{-1})) (J_P \oplus J_{P^\perp}) =  (J_P \oplus J_{P^\perp})  (A \oplus_I (-A^{-1}))  = (J_P \oplus J_{P^\perp}) J$.
As a consequence $J_P \oplus J_{P^\perp}$ is complex linear. Furthermore, 
$(J_P \oplus J_{P^\perp})^2=-I$ and since $J_P \oplus J_{P^\perp}$ is isometric, we conclude that it is also unitary and thus it is a complex structure over (the complex Hilbert space) $\sH_J$.
By construction $J_P \oplus J_{P^\perp}$ commutes with $U$ and thus Theorem \ref{poinccomplexstructure}(d),(g) imply that $J_P \oplus J_{P^\perp} = \pm J$. This is impossible because the left-hand side leaves $\sH_P$ invariant while the right-hand side transforms it into $\sH_P^\perp$. Now, it remains to prove that $\gR_U=\{U_g\:|\:g\in\cP\}''=\gR$. Applying Theorem \ref{poinccomplexstructure} to the real WRES $g\mapsto U_g$ we get a complex structure $J_1$ on (the real Hilbert space) $\sH$ commuting with $U$ which is unique up to the sign. Moreover it holds $\gR_U=\gB(\sH_{J_1})$. Since $J$ is also a complex structure commuting with $U$ it must be $J_1=\pm J$ and so $\gR=\gB(\sH_J)=\gB(\sH_{J_1})=\gR_U$.
\end{proof}

\section{Conclusions}
This work has produced some, in our view interesting,  results (Theorems \ref{poinccomplexstructure} and \ref{main2}) regarding the formulation of quantum theories for elementary relativistic systems. We have in particular established that it is  not physically justified  to formulate the theory on a real Hilbert  space because some physical natural requirements give rise to an essentially unique and Poincar\'e invariant complex structure which commutes with all observables of the theory.  This structure  permits us to reformulate the whole theory in a complex Hilbert space. This formulation is less redundant than  the initial real one,  since differently from the real case, all selfadjoint operators represent observables. The final result is in agreement with the final picture of Sol\`er's theorem which however relies on different physical hypotheses.  This complex structure permits also to associate conserved quantities to the anti selfadjoint generators of the Poincar\'e group allowing the formulation of the quantum version of Noether theorem. Our results are valid also for massless particles where the position observable cannot be defined and the physical analysis by St\"uckelberg, leading to similar conclusions, cannot by applied.  
 The description of a relativistic elementary system has been discussed within two different frameworks. The former is closely related to Wigner's idea of elementary particle (Definition \ref{defERS}), the second (Definition \ref{RRES}) is based on a finer analysis and takes several technical subtleties into account like the fact that representations of continuous symmetries are generally projective unitary and not unitary. Both frameworks lead to the identical  final result.\\
It is however necessary to stress that our notion of elementary system does not encompass relevant physical situations where the commutant of the algebra of observables is not Abelian as it happens in the description of quarks, since the commutant includes a representation of $SU(3)$. However this situation is neither considered by the Wigner notion of elementary particle in complex Hilbert spaces. \\
A final remark about intrinsically quaternionic formulations will conclude our paper. Referring to the three possibilities arising from the thesis {\em Sth} of Sol\`er's  theorem a possibility remains open. This is the formulation of a quantum theory regarding an elementary relativistic theorem on a {\em quaternionic} Hilbert space.  Presumably the algebra of observables cannot coincide with the whole class of selfadjoint operators of $\gB(\sK)$
and the irreducibility of $U$ should be valid only referring to a sublattice of  projectors $\cL\subsetneq \cL(\sK)$, the true elementary observables of the quantum system, 
 similarly to what happens in Definition \ref{RRES}.  Indeed if this were not the case, we would presumably fall into a situation similar to the one discussed in the proof of Theorem \ref{main2} when we demonstrated that the real-quaternionic case leads to a contradiction. 

\section*{Acknowledgements} The authors are grateful to C. Dappiaggi and  S. Mazzucchi for useful comments and discussions.

\appendix 

\section*{Appendix}
\section{Elementary lattice types}\label{Alattices} 
If $\cL$ is a bounded lattice,  $a\in \cL\setminus\{{\bf 0}\}$ is said to be an {\bf atom}, if ${\bf 0} \leq p \leq  a$ implies either $p={\bf 0}$ or $p=a$. 
 Furthermore  $a\in \cL$ is said to {\bf cover} $b\in \cL$ if $a\geq b$, $a\neq b$, and $a\geq c \geq b$ implies either $c=a$ or $c=b$.
A bounded orthocomplemented lattice $\cL$  is said to be {\bf distributive} or {\bf Boolean} if, for all $p,q,r \in \cL$,  we have $p\wedge (q \vee r)= (p \wedge q) \vee (p \wedge r)$ and $p\vee (q \wedge r)= (p \vee q) \wedge (p \vee r)$.\\
The following definitions are valid for  a bounded orthocomplemented lattice $\cL$.

(i) $\cL$  is said to be {\bf orthomodular}, if  $q \geq p$ implies $q= p\vee((p^\perp )\wedge q)$,  $\forall p,q \in \cL$ (if $\cL$ is distributive it is always  orthomodular, the converse is false).

(ii) $\cL$  is said to be  {\bf complete}, resp., {\bf  $\sigma$-complete},   if every  set, resp. countable set, 
 $A \subset \cL$ admits least upper bound ($\vee_{a\in A} a:= \sup A$)
 and  greatest lower  bound ($\wedge_{a\in A} a:= \inf A$) in $\cL$. (In this case {\bf De Morgan's rules} turn out to be valid  also for the case of $A$ infinite, resp., countably infinite:  
$\left(\vee_{a\in A} a \right)^\perp = \wedge_{a\in A} a^\perp$ and 
$\left(\wedge_{a\in A} a \right)^\perp = \vee_{a\in A} a^\perp$.)

(iii)  $\cL$  is said to be {\bf atomic},  if  for any $r\in \cL \setminus\{{\bf 0}\}$ there exists an atom $a$ with  $a\leq r$.

(iii)'  $\cL$  is said to be {\bf atomistic},  if it is atomic and  for every $r \in \cL \setminus\{{\bf 0}\}$,  $r$ is the $\sup$  of the set of atoms $a\leq r$.

(iv)  $\cL$  is said  satisfy the {\bf  covering property},  if $a,p\in \cL$ with  $a$ atom, satisfy $a\wedge p = {\bf 0}$, then  
$a\vee p$ covers $p$.

(v)  $\cL$  is said to be {\bf  separable}, if $\{r_j\}_{j\in A} \subset \cL$ satisfies 
$r_i \perp r_j$, $i\neq j$, then $A$ is finite or countable.

(vi)  $\cL$  is said to be {\bf irreducible}, if the only elements of $\cL$ commuting with every elements of $\cL$ are ${\bf 0}$ and ${\bf 1}$.\\
The bounded orthocomplemented lattice of orthogonal projectors $\cL(\sH)$ in a real, complex or quaternionic  Hilbert space $\sH$ satisfies all properties (i)-(vi) above.
In particular $\cL(\sH)$ is $\sigma$-complete if $\sH$ is separable.
 $P \in \cL(\sH)$
is an atom if and only if $\dim(P(\sH))=1$.  $\cL(\sH)$  is not boolean if $\dim(\sH)>2$.

\section{Definitions and technical results for real and complex Hilbert spaces}\label{secstatic}
\begin{definition}\label{defSP} {\em If $\sH$ is a, respectively real or complex vector space, a respectively {\bf real} or {\bf Hermitian scalar product} is a map 
$(\cdot|\cdot) : \sH \times \sH \to \bR$ resp. $\bC$, which is 

(i)   $\bR$-linear, resp., $\bC$-linear  in the right entry; 

(ii) symmetric ($(x|y)= (y|x)$), resp., Hermitian ($(x|y)= \overline{(y|x))}$); 

(iii) positively defined ($(x|x) \geq 0$ and $(x|x)=0$ implies $x=0$).}
\end{definition}
\noindent Under these conditions the {\bf Cauchy-Schwartz inequality} is valid
$$|(x|y)| \leq \sqrt{(x|x)}\sqrt{(y|y)}\quad x,y, \in \sH\:.$$
and the map $\sH \ni x \mapsto ||x||:= \sqrt{(x|x)}$ turns out to be a norm over $\sH$.\\
The {\bf polarization identity} holds for the respectively real and complex case if $x,y \in \sH$:
$$(x|y) = \frac{1}{4}\left(||x+y||^2 -  ||x-y||^2\right)\:,$$
$$(x|y) =  \frac{1}{4}\left(||x+y||^2 -  ||x-y||^2- i||x+iy||^2 +i ||x-iy||^2\right)$$
\begin{remark} {\em These identities immediately imply that a real or complex linear map between two, respectively, both real or both complex vector spaces, equipped with respectively, real or Hermitian,  
 scalar products, preserves the scalar products if and only if it preserves the associated norms. }
\end{remark}

\begin{definition}\label{defHRC} {\em A {\bf real} or {\bf complex  Hilbert space}  is a respectively  real  or complex vector space $\sH$ equipped 
with a, respectively  real  or  Hermitian, scalar product $(\cdot|\cdot)$ and such that $\sH$ is complete with respect to the norm  $||x|| := \sqrt{( x|x )}$, $x \in \sH$. \\
If $\sH_1$, $\sH_2$ are  both real or both complex  Hilbert spaces,  $f : \sH_1 \to \sH_2$ is a {\bf Hilbert space isomorphism} if it is, respectively, real or complex linear, surjective and preserves the norm 
(thus it is also injective). In this case $\sH_1$ and $\sH_2$ are said to be {\bf isomorphic}.\\ If $\sH_1=\sH_2$, said $f$ is called {\bf Hilbert space automorphism}.}
\end{definition}
\begin{definition}\label{deforth}  {\em If $M\subset \sH$, the closed subspace $M^\perp := \{x\in  \sH\:|\: (x|y)=0\:, \forall y \in M\}$ is the  (respectively real or complex) {\bf orthogonal} of  $M$.}
\end{definition}
 \noindent  Properties of $^\perp$ are identical
 for the real and complex case (e.g. see \cite{R}), in particular, $$\overline{\mbox{span}(M)} = (M^\perp)^\perp \mbox{  and  }\sH = \overline{\mbox{span}(M)} \oplus M^\perp$$
where the bar denotes the topological closure and  $\oplus$ denotes the orthogonal sum of subspaces. The {\bf Riesz lemma} holds for both real and complex Hilbert spaces (the proof being the same (e.g. see \cite{R})):
\begin{theorem}\label{Rlemma}  Let  $\sH$ be a   real or complex  Hilbert space.  $\phi: \sH \to \bR$, respectively $\bC$, is linear and continuous if and only if has the form $\phi = ( x_\phi|\:\:  )$, where $x_\phi \in \sH$ is uniquely determined by $\phi$. Moreover $\|\phi\|:=\sup_{\|x\|=1}|\phi(x)|=\|x_\phi\|$.
\end{theorem}
\begin{definition}\label{defHB} {\em A {\bf Hilbert basis} of the, either real or complex, Hilbert space $\sH$ is a maximal set $N$ of unit-norm pairwise orthogonal vectors.}
\end{definition}
\noindent Zorn's lemma implies the existence of a Hilbert basis for every real or complex Hilbert space. A real or complex  Hilbert space $\sH$  is {\bf separable} (i.e., it admits a dense numerable subset) if either it is finite dimensional or admits a countable Hilbert basis.  If $N\subset \sH$ is a Hilbert basis, the 
standard orthogonal decompositions hold 
$$x= \sum_{z\in N} (z|x)z\;, \quad ||x||^2= \sum_{z\in N} |(z|x)|^2\:, \quad (x|y) = \sum_{z\in N} (z|x)(z|y)\quad \mbox{for every $x, y\in \sH$,}$$ 
the first series  converges with respect to the norm topology of $\sH$,   at most a countable set of summed elements do not vanish  in each series,  and each  series can be rearranged arbitrarily. The proofs of
 these fact can be found, e.g., in \cite{R,S,M}  and they are essentially identical for the real and complex case. 
 
\begin{definition}\label{defopant} {\em Let  $\sH$ be a  real or complex  Hilbert space and $D(A)\subset \sH$ a, respectively real or complex,  subspace. 
An {\bf operator} in $\sH$ is a, respectively $\bR$-linear or $\bC$-linear,  map $A : D(A) \to \sH$.\\
If $\sH$ is complex,  an {\bf anti linear operator} in $\sH$  is a map $A : D(A) \to \sH$, $D(A)$ complex subspace of $\sH$, such that
 $$A(a x + b y)= \overline{a}A(x)
+ \overline{b}A(y)\quad\mbox{ if $a,b \in \bC$ and $x,y\in D(A)$.}$$
In both cases $D(A)$ is called the {\bf domain} of $A$.\\
If $A: \sH \to \sH$ is an operator  $$||A|| := \sup_{||x||=1} ||Ax||\:.$$ If $||A||<+\infty$, $A$ is said to be {\bf bounded} and   $\gB(\sH)$ denotes the set of the bounded operators in $\sH$ with domain coinciding with $\sH$. These are all of continuous operators $A : \sH \to \sH$.\\
The symbol $I$ always  denotes the {\bf identity map}  $I : \sH \ni x \mapsto x \in \sH$.}
\end{definition}
\noindent It turns out that: 

(i)   $\gB(\sH)$ is an {\bf unital (associative) algebra} over, respectively, $\bR$ or $\bC$.
(The  vector space  structure  is $(aA+bB)x := aAx+bBx$ for $x\in \sH$ and $a,b \in \bR$, resp. $\bC$, the associative product 
is  the  composition of functions $(AB)x= A(Bx)$ if $x\in \sH$, and $I$ is the multiplicative unit.)

(ii)  $\gB(\sH) \ni A \mapsto ||A||$  is a norm  satisfying $||AB||\leq ||A||\:||B||$ if $A,B \in \gB(\sH)$ and  $||I||=1$.

(iii) $\gB(\sH)$ is complete with respect to the said norm. \\
Properties (i),(ii),(iii)  make $\gB(\sH)$ a, respectively, real or complex  {\bf  (associative)  unital  Banach algebra} (see, e.g.,\cite{M}).

\begin{definition}\label{defdomain} {\em If $A: D(A) \to \sH$ and $B: D(B) \to \sH$ are, respectively $\bR$ or  $\bC$, linear operators with domains  $D(A),D(B) \subset \sH$,  $$A\subset B \mbox{  means that $D(A) \subset D(B)$  and $B|_{D(A)}=A$.}$$
We adopt standard conventions regarding domains  of combinations of operators $A,B$,

 (i) $D(AB) := \{x \in D(B) \:|\: Bx \in D(A)\}$,

(ii) $D(A+B) := D(A) \cap D(B)$,

(iii)  $D(\alpha A)= D(A)$ for $\alpha \neq 0$.}
\end{definition}
\begin{remark}\label{remassociativity} {\em With these standard definition of domains, we adopt everywhere in the work, the sum and the product turn out to be {\em associative} referring to three operators $A,B,C$ with arbitrary  domains in the same Hilbert space $\sH$: $(A+B)+C= A+(B+C)$ and $(AB)C=A(BC)$. Furthermore $A \subset B$ implies both
$AC \subset BC$ and $CA \subset CB$. Finally $A\subset B$ and $B\subset A$ imply $A=B$.}
\end{remark}  

\begin{definition}\label{defagg} {\em Let $\sH$ be an either real or complex Hilbert space and consider the, respectively  $\bR$-linear
or $\bC$-linear operator
$A: D(A) \to \sH$ where $D(A)\subset \sH$ is dense.\\
$A^* : D(A^*)\to  \sH$  is the {\bf adjoint} operator of $A$ if 
$$D(A^*) := \left\{ y \in \sH \:|\: \exists z_y \in \sH \mbox{ s.t. }
( y|Ax)  = ( z_y|x)\: \forall x \in D(A)  \right\}\:,\quad  A^* y := z_y\:,\:\forall y \in D(A^*)\:.$$}
\end{definition}
\noindent The fact that $D(A)$ is dense immediately implies that $A^*$ is well-defined.

\begin{remark}\label{sumselfadj}$\null$\\
{\em By direct application of the given definitions we get well-known technical results.\\
{\bf (a)} If $A$ is densely defined and $A\subset B$ then $B^* \subset A^*$. \\
{\bf (b)}  $A^* \in \gB(\sH)$  whenever $A \in \gB(\sH)$.
In this case $(A^*)^*=A$. \\
{\bf (c)} For densely defined operators $A,B$ in $\sH$ and $a\in \bR$
or $a \in \bC$ depending on the nature of $\sH$, if $D(A^*B)$ and $D(AB)$ are dense, then
$$(\alpha A)^* = \overline{\alpha}A^*\quad  \mbox{and}   \quad A^*+B^* \subset (A+B)^*\quad \mbox{and} \quad A^*B^* \subset (BA)^*\:,$$
where obviously $\overline{\alpha}=\alpha$ if $\alpha\in \bR$.
The above relations hold in a stronger version on $\gB(\sH)$, making $\gB(\sH) \ni A \mapsto A^* \in \gB(\sH)$ a, respectively real or complex,  {\bf involution}  over the algebra $\gB(\sH)$:
$$(\alpha A)^* = \overline{\alpha}A^*\quad  \mbox{and} \quad   A^*+B^*= (A+B)^* \quad \mbox{and} \quad A^*B^* = (BA)^* \quad A,B \in \gB(\sH)$$
{\bf (d)} The above identities are also valid if  $B\in \gB(\sH)$ and $A$ is densely defined.\\
{\bf (e)} Let $D(A)$ be dense and $U$ be unitary, then it is easy to show that $(UAU^*)^*=UA^*U^*$
}
\end{remark}

\begin{remark} {\em The unital algebra $\gB(\sH)$ is closed with respect to the involution 
$\gB(\sH) \ni A \mapsto A^* \in \gB(\sH)$ so that it is a, respectively  real  or complex, {\bf unital  $^*$-algebra}.
Since  it also satisfies the {\bf $C^*$-property}:  $||A^*A||= ||A||^2$ for $A \in \gB(\sH)$, $\gB(\sH)$ is a, respectively  real  or complex, {\bf unital $C^*$-algebra} (see, e.g.,\cite{M}).}
\end{remark}

\begin{definition}\label{defclosop} {\em Let $\sH$ be an either real or complex Hilbert space and consider the, respectively real or complex, linear operator
$A: D(A) \to \sH$.\\
{\bf (1)} $A$  is said to be {\bf closed} if  the {\bf graph}   of $A$, that is the set pairs $(x, Ax) \subset \sH \times \sH$ with $x\in D(A)$,  is closed in the product topology of $\sH\times \sH$. \\
{\bf (2)} $A$ is {\bf closable}
if the closure of its graph is the graph of an operator, denoted by $\overline{A}$, and called the {\bf closure} of $A$.\\
{\bf (3)} If $A$ is closable, a respectively real or complex, subspace $S \subset D(A)$ is called {\bf core} for $A$ if $\overline{A|_S} = \overline{A}$.}
\end{definition}

\begin{remark}\label{remarkclosure}$\null$\\
{\em {\bf (a)} Directly from the definition,  $A$ is closable if and only if  there are no sequences of elements $x_n \in D(A)$ such that $x_n \to 0$ and $Ax_n \to y\neq 0$ as $n \to +\infty$. In this case $D(\overline{A})$ is made of the elements $x\in \sH$ such that $x_n \to x$ and $Tx_n \to y_x$ for some sequences $\{x_n\}_{n \in \bN} \subset D(A)$
and some $y_x \in \sH$. In this case $\overline{A}x = y_x$. \\
{\bf (b)} As a consequence of (a) one has that,
if $A$ is closable, then $aA+bI$ is closable and   $\overline{aA+bI} = a \overline{A} + b I$ for every $a,b$ real or complex numbers in accordance with  $\sH$.\\
{\bf (c)}  Directly from the definition,  $A$ is closed if and only if 
$D(A) \ni x_n \to x \in \sH$ and $Tx_n \to y \in \sH$ imply both $x \in D(A)$ and $y=Ax$.\\
{\bf (d)} If $A$ is densely defined,
$A^*$ is always closed from the definition of adjoint operator and (c) above. Moreover, a densely defined operator $A$ is closable if and only if $D(A^*)$ is dense. In this case $\overline{A}= (A^*)^*$. The proof is the same in real and complex case see, e.g., \cite{M}.\\
{\bf (e)} If $A$ is densely defined we have $$Ker(A^*)=Ran(A)^\perp\:,  \quad Ker(A^*)^\perp = \overline{Ran(A)}\:, \quad Ker(A)\subset Ran(A^*)^\perp\:.$$ The last inclusion becomes an identity when $A \in \gB(\sH)$. The proofs of these relations are 
elementary and  identical in real and complex Hilbert spaces \cite{M}.}
\end{remark}
\noindent The {\bf closed graph theorem} holds for both the real and the complex Hilbert space case  since the well known proof is valid in real or complex Banach spaces.
\begin{theorem}\label{cgT}
Let $\sH$ be a real or complex Hilbert space. A, respectively $\bR$-linear  or $\bC$-linear, operator $A: \sH \to \sH$ is closed if and only if $A \in \gB(\sH)$.
\end{theorem}

\begin{definition}\label{defop} {\em Let $\sH$ be an either real or complex Hilbert space and consider the,  respectively $\bR$-linear  or $\bC$-linear,  operator $A: D(A) \to \sH$. $A$ is said to be

(1) {\bf  symmetric} if it is densely defined and $A \subset A^*$,

(2) {\bf anti symmetric} if it is densely defined and $-A \subset A^*$

(3) {\bf selfadjoint}  if it is densely defined and  $A=A^*$, 

(4) {\bf anti selfadjoint} if it is densely defined  $-A=A^* $, 

(5) {\bf essentially selfadjoint} if it is symmetric and  $(A^*)^* = A^*$. 

(6) {\bf unitary} if  $A^{*}A= AA^* = I$. 

(7) {\bf normal} if it is densely defined and $AA^*=A^*A$.

(8) {\bf positive}, written $A \geq 0$, if $(x|Ax)\geq 0$ for every $x \in D(A)$.

(9) an {\bf isometry} if $D(A)= \sH$ and $A$ is norm preserving.

(10) a {\bf partial isometry} if $A \in \gB(\sH)$ and $A$ is norm preserving on $Ker(A)^\perp$.}
\end{definition}

\begin{remark}\label{essentclos}$\null$\\
{\em {\bf (a)} If $A$ is symmetric and $D(A)= \sH$ (so that $A=A^*$), then it is bounded as an immediate consequence of the
 closed graph theorem.\\
{\bf (b)} If $A$ is unitary then $A, A^* \in \gB(\sH)$, the proof is elementary. Notice also that the following facts are equivalent 
for the operator  $A: \sH \to \sH$: (i) $A$ is unitary, (ii) $A$ is a surjective isometry, (iii) $A$ is surjective and preserves the scalar 
product, (iv) $A$ is a Hilbert space automorphism.  Finally $A$ is an isometry if and only if $A^*A=I$.
If $A$ is a partial isometry, it is easy to prove that
$A^*A$ is the orthogonal projector (see Def.\ref{defproJ}) onto $Ker(A)^\perp$ and $AA^*$ is the orthogonal projector onto $Ran(A)= \overline{Ran(A)}$.\
\\
{\bf (c)} It is easy to show that a symmetric operator is always closable, moreover for such an operator the following conditions 
are equivalent :
(i) $(A^*)^* = A^*$ ($A$ is essentially selfadjoint), (ii) $\overline{A}= A^*$, (iii) $\overline{A}= (\overline{A})^*$.
 If these conditions are valid, $\overline{A}= (A^*)^* =A^*$ is the unique selfadjoint extension of $A$. 
The proof is the same in the real and the complex case (e.g., see \cite{M}).\\
{\bf (d)} If $A\subset B$ are symmetric operators and $A$ is essentially selfadjoint, then also $B$ is essentially selfadjoint and $\overline{A}= \overline{B}$. The proof is elementary.\\
{\bf (e)}  In the complex Hilbert space  case $A$, is anti selfadjoint if and only if $iA$ is selfadjoint.\\
{\bf (f)} Let $U$ be unitary, then $A$ is closable iff $UA$ is closable iff $AU$ is closable. In this case $\overline{UA}=U\overline{A}$ and $\overline{AU}=\overline{A}U$. As a consequence $\overline{UAU^*}=U\overline{A}U^*$.}
\end{remark}

\begin{definition}\label{defproJ} {\em Let $\sH$ be an either real or complex Hilbert space.
 $P \in \gB(\sH)$ is called {\bf orthogonal projector} when $PP=P$ and $P^*=P$.  ${\cal L}(\sH)$ denotes the set  of orthogonal projectors of $\sH$.}
\end{definition}

\begin{remark} {\em  Let $\sH$ an either  real or complex Hilbert space.
If $P \in {\cal L}(\sH)$, then $P(\sH)$ is a closed, respectively real or complex subspace. If $\sH_0 \subset \sH$ is a closed, respectively real or complex subspace, there exists exactly one  $P \in  {\cal L}(\sH)$
such that $P(\sH) = \sH_0$. Finally, $I-P \in  {\cal L}(\sH)$ and it projects onto $\sH_0^\perp$. The proofs are  identical in real and complex Hilbert spaces (e.g., see \cite{M}). }
\end{remark}
\noindent The definition of spectrum of the operator $A: D(A) \to \sH$ is the same for both real and  complex Hilbert spaces.
\begin{definition}\label{defspec} {\em Let $\sH$ be an either real or complex Hilbert space 
and let  $\bK$ denote the field of $\sH$. Consider a $\bK$-linear operator $A: D(A) \to \sH$, with 
$D(A)\subset \sH$. The {\bf resolvent set} of $A$ is the subset of $\bK$,
 $$\rho(A) := \{\lambda \in \bK \:|\:  (A-\lambda I) \mbox{ is injective on $D(A)$, } \overline{Ran(A-\lambda I)}=\sH\:, (A- \lambda I)^{-1} \mbox{is bounded} \}$$
The {\bf spectrum} of $A$ is the set $\sigma(A):=\mathbb{K}\setminus\rho(A)$ and it is given by the union of the following pairwise disjoint three parts:

(i) the  {\bf point-spectrum}, $\sigma_p(A)$,  where $A-\lambda I$ not injective ($\sigma_p(A)$ is the set of {\bf eigenvalues} of $A$),

(ii) the  {\bf continuous spectrum},  $\sigma_c(A)$, where $A-\lambda I$  injective, $\overline{Ran(A-\lambda I)} = \sH$ and $(A-\lambda I)^{-1}$ not bounded, 

(iii) the {\bf residual spectrum}, $\sigma_r(A)$, where $A-\lambda I$ injective and $\overline{Ran(A-\lambda I)} \neq \sH$.}
\end{definition}

\begin{remark}$\null$\\
{\em {\bf (a)} If $A=\pm A^*$ or if $A$  is unitary,  the residual spectrum is absent  (e.g., see \cite{M}). \\
{\bf (b)} If $\bK=\bC$, then  $A=A^*$ implies $\sigma(A)\subset \bR$, $A=-A^*$ implies $\sigma(A)\subset i \mathbb R$, and $AA^*=A^*A=I$ implies $\sigma(A)\subset \{ e^{ia}\:|\: a \in \bR\}$. (e.g., see \cite{M}).\\
{\bf (c)} If $\bK=\bR$, it turns out that $A=A^*$ implies $\sigma(A)\subset \bR$, $A=-A^*$ implies $\sigma(A) = \emptyset$, $AA^*=A^*A=I$  implies $\sigma(A)\subset \{\pm 1\}$. The proof is similar to the one for the complex case.}
\end{remark}

\begin{definition}\label{defPVM} {\em Let $\sH$ be an either real or complex Hilbert space and $\Sigma(X)$ a $\sigma$-algebra over $X$. A {\bf projector-valued measure (PVM)} over $X$ is a map $\Sigma(X) \ni E \mapsto P_E \in {\cal L}(\sH)$ such that

 (i)  $P_X=I$, 

(ii) $P_EP_F = P_{E\cap F}$,

 (iii) 
$\sum_{j \in N} P_{E_j}x= P_{\cup_{j\in N}E_j}x$ for $x\in \sH$, $N$  finite or countable, $E_j \cap E_k = \emptyset$ if $k\neq j$. }
\end{definition}

\begin{remark}
{\em If $x,y \in \sH$, 
$\Sigma(X) \ni E \mapsto ( x|P_Ey) =: \mu^{(P)}_{xy}(E)$ is a {\em signed measure} if $\sH$ is real  or, respectively, a {\em complex measure} if $\sH$ is complex. In both cases  the {\em finite  total variation} is denoted by $|\mu^{(P)}_{xy}|$. It holds  
 $\mu^{(P)}_{xy}(X)= ( x| y )$ and  $\mu^{(P)}_{xx}$ is always positive and finite. The proof are elementary and  identical  in the real and the complex case (e.g., see \cite{R,M}).}
\end{remark}
\noindent We have a fundamental notion defined in the following proposition which can be demonstrated with  an essentially identical proof for real and complex Hilbert spaces \cite{R,S,M}.
\begin{proposition}\label{propint}
Let $\sH$ be an either real or complex Hilbert space  and $P: \Sigma(X) \to {\cal L}(\sH)$ a PVM. 
If $f: X \to \bK$ is measurable  where $\bK$ is the field of $\sH$, define 
$$\Delta_f := \left\{x \in \sH \:\left|\:  \int_{X} |f(\lambda)|^2 \mu^{(P)}_{xx}(\lambda)\right.\right\}\:.$$
{\bf (a)} $\Delta_f$ is a, respectively, real or complex  subspace of $\sH$ and  there is a unique operator  
\beq \int_X f(\lambda) dP(\lambda) : \Delta_f \to \sH \label{intop}\eeq
such that 
\beq \left( x  \left| \int_X f(\lambda) dP(\lambda) y \right.\right) = 
\int_{X} f(\lambda) \mu^{(P)}_{xy}(\lambda)\quad \forall x \in \sH \:, \forall y \in \Delta_f \label{intop2}\eeq
{\bf (b)} $\Delta_f$ is dense in $\sH$ and  the operator in (\ref{intop})   is closed and normal. \\
{\bf (c)} The operator in (\ref{intop})   is bounded if and only if $\Delta_f = \sH$ and this is equivalent to say that $f$ is essentially bounded with respect to $P$.\\
{\bf (d)} It holds
\beq \left( \int_X  f(\lambda)\:  d P(\lambda)\right)^* = \int_X \overline{f(\lambda)}  d P(\lambda)\:, \eeq 
where $ \overline{f(\lambda)}$ is replaced by  $f(\lambda)$ in the real-Hilbert space case,
and 
\beq\label{inin} \left|\left| \int_X  f(\lambda)\:  d P(\lambda) x\right|\right|^2 = \int_X |f(\lambda)|^2 d \mu^{(P)}_{xx}(\lambda) \quad \forall x \in \Delta_f \:.\eeq
\end{proposition}

\begin{remark}
{\em The integral in the right-hand side of (\ref{intop2}) is well defined for $y \in \Delta_f$ since it turns out that 
$f$ is $L^2(X,\Sigma(X), \mu^{(P)}_{yy}) \subset L^1(X, \Sigma(X),|\mu^{(P)}_{xy}|)$. In particular, the estimate holds 
\beq  \int_X |f(\lambda)|\:  d |\mu^{(P)}_{xy}|(\lambda) \leq ||x|| \sqrt{\int_X |f(\lambda)|^2 d \mu^{(P)}_{yy}(\lambda)}\qquad  \forall y \in \Delta_f \:,  \forall x \in \sH\:. \eeq 
The proof is essentially the same in the real and the complex case (e.g., see \cite{R,M}).} 
\end{remark}
\noindent We are in a position to state  the fundamental tool of the spectral theory.
\begin{theorem}[Spectral Theorem]\label{st}
Let $\sH$ be a Hilbert space over the field $\bK= \bR$ or $\bC$  and consider a $\bK$-linear operator  $A: D(A) \to \sH$ with $D(A)\subset \sH$ a dense $\bK$-linear subspace. Denote by ${\cal B}(\bK)$ the Borel $\sigma$-algebra on $\bK$. \\
{\bf (a)}  If $\bK= \bC$ and $A$ is normal (in particular selfadjoint, anti-selfadjoint, unitary), then there is a unique PVM, 
$P^{(A)} : {\cal B}(\bC) \to {\cal L}(\sH)$, such that
$$A = \int_{\bC} \lambda dP^{(A)}(\lambda)\:.$$
{\bf (b)} If $\bK= \bR$ and $A$ is selfadjoint, then  there is a unique PVM, 
$P^{(A)} : {\cal B}(\bR) \to {\cal L}(\sH)$, such that
$$A = \int_{\bR} \lambda dP^{(A)}(\lambda)\:.$$\\
{\bf (c)}
In both cases the following facts hold.

(i)  $supp(P^{(A)}) = \sigma(A)$, where the {\bf support} $supp(P^{(A)})$ of $P^{(A)}$ is the complement in $\bK$ of  the union of all open sets  $O \subset \bK$ with $P_O^{(A)}=0$. As $P^{(A)}$ is supported in $\sigma(A)$, the integrals in (a) and (b) can be restricted to this set.

(ii) $B \in \gB(\sH)$ satisfies $P^{(A)}(E)B=BP^{(A)}(E)$ for every $E \in \cB(\bK)$ 
 iff $BA\subset AB$.\\
{\bf (d)} Finally, if $A$ is selfadjoint in both cases:

(i) $\lambda \in \sigma_p(A)$  $\;\Leftrightarrow\;$  $P^{(A)}(\{\lambda\}) \neq 0$.  In this case $P^{(A)}(\{\lambda \})$ is the orthogonal projector onto the eingenspace of $A$ with eigenvalue $\lambda$;

(ii) $\lambda \in \sigma_c(A)$ $\;\Leftrightarrow\;$  $P^{(A)}(\{\lambda\}) =0$ and  $E_\lambda \subset \bR$ open with $E_\lambda \ni\lambda$ gives $P^{(A)}(E_\lambda) \neq 0$;

(iii) if $\lambda \in \sigma(T)$ is isolated, then $\lambda \in \sigma_p(A)$;

(iv) if $\lambda \in \sigma_c(A)$, then for any $\epsilon >0$ there exists 
$\phi_\epsilon\in D(A)$, $||\phi_\epsilon ||=1$ with 
$$0<||A\phi_\epsilon - \lambda \phi_\epsilon|| \leq \epsilon\:.$$
\end{theorem}
\noindent The proof of (a) can be found, e.g., in \cite{R,MV,S,M}.  The proof presented in \cite{S} for the complex case can be re-adapted to the real case (b)  since it does not use the Cayley transform but real functions only. The proof of (c) and (i),(ii) of (b) when $\sH$ is complex  can be found in \cite{M} and it is essentially identical in  the real case.\\
A useful technical result arising from the spectral theorem is the following whose proof is identical in the real and complex case \cite{M}.
\begin{proposition}\label{propopos} If $A$ is a selfadjoint operator in a, either real or complex,  Hilbert space, $A \geq 0$ if and only if $\sigma(A) \subset [0,+\infty)$.
\end{proposition}
\noindent In view of theorem \ref{st}, if $f:\bR  \to \bK$ is measurable and $A$ selfadjoint, we use the notation
\beq
f(A)  : =  \int_\bR f(\lambda) dP^{(A)}(\lambda)   \:.\label{fA}
\eeq
As $P^{(A)}$ is supported in $\sigma(A)$ the definition above can equivalently be stated 
restricting the integral (and the domain of $f$) to $\sigma(A)$. 
An important example of "operator function" is the following:
\begin{proposition}\label{sqrt}
Let $A$ as in Prop.\ref{propopos}, then $\sqrt{A}$ defined through (\ref{fA}) is the unique selfadjoint positive operator such that $\sqrt{A}\sqrt{A}=A$ 
\end{proposition}
\noindent An elementary but important result is the following whose proof is identical in the real and complex Hilbert space case (see, e.g., \cite{M})
\begin{proposition}\label{defpolynomial}
Let $\sH$ be a, respectively, real or complex Hilbert space and  let $A: D(A) \to \sH$ be a selfadjoint operator in $\sH$. If $p (x) = \sum_{k=0}^N a_kx^k$ is a real polynomial, then it holds
$$ \sum_{k=0}^N a_kA^k= \int_\bR p(\lambda) dP^{(A)}(\lambda)$$
where the left-hand side is the operator defined on its natural domain in accordance to Def.\ref{defdomain} with $A^0:=I$ and $A^k: = A\cdots \mbox{($k$ times)}\cdots A$.
\end{proposition}
\noindent To conclude we list the three most common operator continuity notions among the seven appearing in the literature  (these can be induced form suitable seminorm topologies, e.g., see \cite{R,M}, as is well known).
\begin{definition}\label{defcontinuity} Let $\sH$ be a, respectively, real or complex Hilbert space and $\sT$ a topological space. A map $\sT \ni x \mapsto V_x \in \gB(\sH)$ is said to be:\\
{\bf (a)} {\bf uniformly continuous at} $x_0$, if $||V_x-V_{x_0}||\to 0$ for $x \to x_0$;\\
{\bf (b)} {\bf strongly continuous at} $x_0$, if $||V_xz-V_{x_0}z||\to 0$ for $x \to x_0$
and every $z \in \sH$;\\
{\bf (c)} {\bf weakly continuous at} $x_0$, if $(u|V_xz)\to (u|V_{x_0}z)$ for $x \to x_0$
and every $u,z \in \sH$;\\
{\bf (d)} {\bf uniformly continuous},  {\bf strongly continuous},  {\bf weakly continuous} if, respectively, (a), (b) or (c) is valid for every $x_0\in \sT$. 
\end{definition}
\noindent Evidently (a) implies (b) which, in turn, implies (c).

\section{Quaternionic Hilbert spaces}\label{QHS}
 $\bH := \{a1 + bi + cj + dk\:|\: a,b,c,d \in \bR\}$ denotes the real unital associative  algebra of quaternions. $i,j,k$ are the standard 
{\bf imaginary units} satisfying $i^2=j^2=k^2 = -1$ and $ij= -ji = k$, $jk=-kj = i$, $ki=-ik=j$ which give rise to the notion of associative,  distributive and non-commutative product in $\bH$ with $1$ as neutral element. $\bH$ is a division algebra, i.e., every non zero element admits a multiplicative inverse. The center of $\bH$ is $\bR$.  
 $\bH$ is assumed to be equipped with the {\bf quaternionic conjugation} $\overline{a1 + bi + cj + dk}= a1 - bi - cj -dk$. Notice that the conjugation satisfies $\overline{qq'}= \overline{q'} \overline{q}$ and $\overline{\overline{q}}=q$ for all $q,q'\in \bH$.
If $q \in \bH$, its {\bf real part} is defined as $Re \: q := \frac{1}{2}(q+\overline{q}) \in \bR$.
The quaternionic conjugation together with the Euclidean {\bf norm} $|q|:=\sqrt{q\overline{q}}$ for $q\in \bH$, makes $\bH$ a real unital $C^*$-algebra which also satisfies the {\bf composition algebra} property $|qq'|=|q|\:|q'|$. 

\begin{definition} {\em A {\bf quaternionic vector  space} is  
an additive Abelian group $(\sH, +)$ denoting the sum operation, equipped with a right-multiplication $\sH \times \bH \ni (x, q) \mapsto xq\in \sH$ such that (a) the right-multiplication  is distributive with respect to $+$,  (b) the sum of quaternions  is distributive with respect to the right-multiplication,  (c)  $(xq) q' = x(qq')$ and (d) $x1=x$  for all $x\in \sH$ and $q,q' \in \bH$.}
\end{definition}
\begin{definition}  
{\em A {\bf quaternionic Hilbert space}  is a quaternionic vector space $\sH$ equipped with a {\bf Hermitian quaternionic scalar product}, i.e., a map $\sH \times \sH \ni (x,y) \mapsto \langle x|y \rangle \in \bH$ such that  (a) $\langle x|yq+z\rangle = \langle x|y\rangle q+\langle x|z\rangle$ for every $x,y,z\in \bH$ and $q \in \bH$, (b)  $\langle x|y\rangle = \overline{\langle y|x\rangle}$ for every $x,y \in \sH$ and (c) $\langle x| x\rangle \in [0,+\infty)$ where (d) $\langle x| x\rangle=0$ implies $x=0$,    and  $\sH$ is complete with respect to the norm $||x|| = \sqrt{\langle x| x \rangle}$. }
\end{definition}
\noindent 
The standard {\bf Cauchy-Schwartz} inequality holds, $|\langle x|y \rangle| \leq ||x||\: ||y||$ for every $x,y \in \sH$ for the above defined quaternionic Hermitian scalar product.
The notion of {\bf Hilbert basis} (Def. \ref{defHB}) is the same as for real and complex Hilbert spaces and properties are the same with obvious changes. A quaternionic Hilbert space turns out to be separable as a metrical space if and only if it admits a finite or countable Hilbert basis. The notion of {\bf orthogonal subspace}  $S^\perp$ of a set $S \subset \sH$ is defined with respect to $\langle \cdot | \cdot \rangle$ (Def. \ref{deforth}) and enjoys the same standard properties as for the analog in real and complex Hilbert spaces.
 The notion of operator norm and  bounded operator are the same as for real and complex Hilbert spaces. Since the {\bf Riesz lemma} (Theorem \ref{Rlemma})  holds true also for quaternionic Hilbert spaces, the {\bf adjoint operator}
 $A^* : \sH \to \sH$ 
of a bounded quaternionic linear operator $A : \sH \to \sH$ can be defined as the unique quaternionic linear operator such that $\langle A^*y|x\rangle = \langle y|Ax\rangle$ for every pair $x,y \in \sH$. Notice that if $A : \sH \to \sH$ is quaternionic linear
and $r\in \bR$, we can define the quaternionic linear operator $rA: \sH \to \sH$ such that $rA x := (Ax)r$ for all $x \in \sH$. Replacing $r$ for $q\in \bH$ produces a non-linear map in view of non-commutativity of $\bH$. Therefore only real linear combinations of 
quaternionic linear operators are well defined.
$\gB(\sH)$ denotes the real unital $C^*$-algebra of bounded operators over $\sH$. The notion of {\bf orthogonal projector} $P: \sH \to \sH$ is defined exactly as in the real or complex Hilbert space case,
$P$ is bounded,  $PP=P$ and $P^*=P$. Orthogonal projectors $P$ are one-to-one with the class of 
closed subspaces $P(\sH)$ of $\sH$. $\cL(\sH)$ denotes the orthocomplemented complete lattice of orthogonal projectors of $\sH$. This lattice also satisfies properties  (i)-(vi) listed in Appendix \ref{Alattices}. 
Another important notion is the one of {\bf square root} of positive bounded operators. As for the real and complex case (see Prop.\ref{sqrt}), also for quaternionic Hilbert spaces, if $A$ is bounded and positive, then there exists a unique bounded positive operator $\sqrt{A}$ such that $\sqrt{A}\sqrt{A}=A$. In particular, if $A : \sH \to \sH$ is a bounded quaternionic-linear operator $|A| := \sqrt{A^*A}$ is well defined positive and self-adjoint.
For the proofs of the afore-mentioned properties and for more advanced issues, especially concerning spectral theory,  we address the reader to \cite{GMP1} and \cite{GMP2}.

\begin{remark}\label{remquatV2}{\em
In \cite{V2} and \cite{Em} the Quaternionic Hilbert space is defined through a \textit{left}-multiplication $\bH\times\sH\ni (q,u)\mapsto qu\in\sH$ and a Hermitian quaternionic scalar product $\sH\times\sH\ni (u,v)\mapsto \langle u|v\rangle\in\bH$ whose only difference resides in point (a): $\langle qx|y\rangle=q\langle x|y\rangle$ for all $x,y\in\sH$ and $q\in\bH$. To define a left-multiplication on a space with right-multiplication it suffices to define $qu:=u\overline{q}$ for all $q\in\bH$ and $u\in\sH$, while the scalar product does not need to be modified. It is immediate to see that a map $A:\sH\rightarrow \sH$ is linear, bounded, self-adjoint, idempotent and unitary with respect to the right-multiplication if and only if it has the same properties with respect to the left-multiplication. This allows us to use indifferently the results in \cite{GMP1,GMP2} and \cite{V2},\cite{Em}.}
\end{remark}

\section{Trace class operators}\label{tracelass}
We present here  some basic notions about trace-class operators for real, complex, and quaternionic Hilbert spaces.
\begin{definition}\label{DEFTCO} Let $\sH$ be a separable real, complex or quaternionic Hilbert space.  An operator $T \in \gB(\sH)$ is said to be of {\bf trace class} if $\sum_{k \in K} (e_k||T| e_k)< +\infty$ for a Hilbert basis $\{e_k\}_{k \in K} \subset \sH$. 
\end{definition}
\noindent It is possible to prove that, in view of the given definition,  
 the {\bf trace} of $T$ computed with respect to every Hilbert basis $\{v_k\}_{k \in K}$, i.e.,  
$$tr(T) = \sum_{k \in K} (v_k|Tv_k)$$
absolutely converges\footnote{This property is used as {\em definition} of trace-class operator in \cite{V2}. Unfortunately this property of trace-class operators is {\em equivalent} to the definition \ref{DEFTCO} in complex Hilbert spaces, but is a simple consequence in real Hilbert spaces. Complex structures in real Hilbert spaces are easy counterexamples. However, all theorems preved in \cite{V2} exploit the theory arising from definition \ref{DEFTCO}.}  in, respectively,  $\bR$, $\bC$ or $\bH$.  Furthermore $tr(A)$
does not depend on the chosen Hilbert basis. 
All the theory of trace-class operators (see Sect.4.4  of \cite{M}) is essentially identical in real, complex and quaternionic Hilbert spaces as is based on the theorem 
of spectral decomposition of self-adjoint compact operators (Theorems 4.17 and 4.18  in \cite{M} valid also for the real Hilbert space case and see \cite{GMP3} for the quaternionic case), the polar decomposition theorem of bounded operators and the notion of absolute-value operator $|A|$ of a bounded operator $A$ (see \cite{GMP1} and \cite{GMP2} for the quaternionic case). 
Rephrasing the proof of these statements appearing in Sect.4.4 of \cite{M}, we have in particular that the set of trace class operators  (a) is closed with respect to the $^*$-operation and (b) is a $\bR$-linear subspace of $\gB(\sH)$ in the real and quaternionic case and is a $\bC$-linear subspace of $\gB(\sH)$ if $\bC$ is complex. The map $T \mapsto tr(T)$ 
defined over the linear space of trace class operators is a respectively $\bR$, $\bC$ or $\bH$ linear functional and the following proposition is true.
\begin{proposition}
If $A, T \in \gB(\sH)$ and $T$ is trace class, then $TA$ and $AT$ are  of trace class and
$tr(AT)=tr(TA)$.
\end{proposition}

\section{Universal enveloping algebra}\label{secUEA}
\noindent A notion, very useful in  quantum physical applications,  is the  {\em universal enveloping algebra} of a Lie algebra, discussed   Ch.3 Section 2 of \cite{V}. To introduce this notion we observe that a
 real unital associative algebra can be turned 
into a (real)  Lie algebra simply by taking the natural commutator $[a,b]:=ab-ba$. However there exists also an inverse procedure allowing one
  to canonically embed a given real Lie algebra $\gg$ into a suitable real  associative unital algebra $E_\gg$ with product $\circ$ such that  $[{\bf A},{\bf B}]_\gg$ identifies with ${\bf A}\circ{\bf B}-{\bf B}\circ{\bf A}$ for any ${\bf A},{\bf B}\in\gg$. 

\begin{definition}[{\bf Universal enveloping algebra}]\label{UEA}{\em
Let $\gg$ be a real Lie algebra. The {\bf universal enveloping algebra} $E_\gg$ of $\gg$ is the quotient real  associative unital algebra,
$$
E_\gg:=T_\gg/K_\gg
$$
of  the (real associative unital) tensor algebra $T_\gg$ generated by $\gg$  
and the two-sided ideal $K_\gg$ of $T_\gg$ generated by the elements $${\bf A}\otimes{\bf B}-{\bf B}\otimes{\bf A}-[{\bf A},{\bf B}]_\gg \quad \mbox{with ${\bf A},{\bf B}\in\gg$.}$$  The  product of $E_\gg$ will be denoted by $\circ$.
}
\end{definition}
\begin{remark} {\em \noindent The enveloping algebra is clearly real  associative and unital. From now on, dealing with Lie algebras and universal enveloping algebra we will almost always omit the adjectives {\em real} and {\em associative}.}
\end{remark}
 \noindent The quotient map $\pi_\gg :T_\gg \to E_\gg$ is a unital algebra homomorphism and $E_\gg$ admits a unit obviously given by $\pi(1)$ itself. Here $1\in \bR$ where $\bR$ is viewed as a trivial subspace of $T_\gg$.  Moreover,
$$\iota_\gg := \pi|_{\gg} : \gg \to E_\gg$$ is a Lie-algebra homomorphism because
$$
\pi({\bf A})\circ \pi({\bf B})-\pi({\bf B})\circ \pi({\bf A})=\pi({\bf A}\otimes{\bf B}-{\bf B}\otimes{\bf A})=\pi([{\bf A},{\bf B}]_\gg)\quad \mbox{if ${\bf A},{\bf B}\in\gg$\:.}
$$
$E_\gg$ is a {\em natural} object because the following universality result is valid as a consequence of the universal property 
of the tensor product.
\begin{theorem}[{\bf Universal property}]\label{UP}
Let $V$ be any  (unital associative) algebra and $f:\gg\to  V$ a Lie-algebra homomorphism. Then there exists a unique unital algebra homomorphism $\tilde{f}: E_\gg \to V$ of such that $f=\tilde{f}\circ i_\gg$.
\end{theorem}
\noindent It is easy to see that $(E_\gg,\iota_g)$ is the only couple of unital associative algebra and Lie algebra homomorphism from $\gg$ to $E_g$ satisfying this property up to isomorphisms.  
To go on, suppose that $\gg$ is finite dimensional and consider a vector basis $\{{\bf X}_1,\dots,{\bf X}_n\}$ of $\gg$. The set containing the elements $1$ and all of possible finite products ${\bf X}_{i_1}\otimes\dots\otimes{\bf X}_{i_k}$ is a basis of  $T_\gg$. The objects $\pi(1),\pi({\bf X}_{i_1}\otimes\dots\otimes{\bf X}_{i_k})$ therefore span the quotient $E_\gg$, but they are not linearly independent. In order to get a basis we invoke the following (see \cite{V})

\begin{theorem}[{\bf Poincar\'{e}-Birkhoff-Witt Theorem}]\label{PBW}
Let $\gg$ be a Lie algebra of finite dimension $n$ and  $\{{\bf X}_1,\dots,{\bf X}_n\}$ a vector basis of $\gg$.
A vector basis of $E_\gg$ is made of $\pi(1)$ and all possible products
$$\pi({\bf X}_{i_1})\circ \cdots \circ \pi({\bf X}_{i_k})$$
where $k=1,2,\ldots$ and  $i_j \in \{1,\dots,n\}$ with the constraints  $i_1\le\cdots\le i_k$.
\end{theorem}
\noindent
 As a corollary, the Lie-algebra homomorphism  $i_\gg :\gg\to E_\gg$ is injective since, evidently,  the kernel of $\pi_\gg$, $K_\gg$, does not contain  
elements of $\gg \setminus \{0\}$.  Thus $\gg$ turns out to be naturally isomorphic to the  Lie subalgebra of $E_\gg$ given by $i_\gg(\gg)$.
\begin{remark} {\em  Due to the afore-mentioned canonical isomorphism, we will simply denote:

(i)  $\pi(1)$ by $1$,

(ii)  $\pi({\bf A})$ by ${\bf A}$,

(iii) $\pi({\bf A})\circ \pi({\bf B})$ by ${\bf A}\circ{\bf B}$,\\
when ${\bf A},{\bf B}\in\gg$.
\noindent In particular, the generic element ${\bf M}\in E_\gg$ can be written as
\begin{equation}\label{dev}
{\bf M}=\pi\left(c_01+\sum_{k=1}^N \sum_{j=1}^{N_k}c_{jk}{\bf A}_{j1}\otimes...\otimes{\bf A}_{jk}\right)=c_01+\sum_{k=1}^N \sum_{j=1}^{N_k}c_{jk}{\bf A}_{j1}\circ...\circ{\bf A}_{jk}
\end{equation}
for some $N, N_k\in\bN$, $c_0,c_{jk}\in\bR$ and ${\bf A}_{jm}\in\gg$.}
\end{remark}
\noindent The last notion we intend to define is the notion of {\em symmetric} element of $E_\gg$.
We start by considering the unique linear map $p : T_\gg \to T_\gg$ such that 
 $$p(1) = 1\:, \quad p({\bf A}_1\otimes\cdots\otimes{\bf A}_n) =(-1)^n{\bf A}_n\otimes\cdots\otimes {\bf A}_1 \quad \mbox{for $n=1,2,\ldots$ and ${\bf A}_j \in \gg$.}$$
Notice that $p$ is {\em involutive}, i.e., $pp=I_{T_\gg}$, so that it is a vector space automorphism,  and  also fulfills the crucial property
$$
p({\bf A}\otimes{\bf B}-{\bf B}\otimes{\bf A}-[{\bf A},{\bf B}]_\gg)={\bf B}\otimes{\bf A}-{\bf A}\otimes{\bf B}-[{\bf B},{\bf A}]_\gg\quad \mbox{for ${\bf A},{\bf B}\in\gg$.}
$$
Hence the ideal $K$ is invariant under the action of $p$ and thus there exists a unique vector space automorphism 
$$E_\gg \ni {\bf M} \mapsto  {\bf M}^+ \in  E_\gg$$ such that $\pi({\bf M})^+=\pi( p({\bf M}))$ for every ${\bf M} \in E_\gg$.
Referring to (\ref{dev}), the action of $^+$ on ${\bf M}$ is completely defined by
\begin{equation}\label{Ptilde}
\left(c_01+\sum_{k=1}^N \sum_{j=1}^{N_k}c_{jk}{\bf A}_{j1}\circ...\circ{\bf A}_{jk}\right)^+=c_01+\sum_{k=1}^N \sum_{j=1}^{N_k}c_{jk}(-1)^k{\bf A}_{jk}\circ...\circ {\bf A}_{j1}
\end{equation}
This map satisfies the following properties making it a {\em real involution} on the real unital algebra  $E_\gg$,
$$(c {\bf M})^+ = c {\bf M}\quad \mbox{and}\quad  ({\bf M} + {\bf N})^+ = {\bf M}^+ + {\bf N}^+ \quad \mbox{and}\quad  ({\bf M} \circ {\bf N})^+ = {\bf N}^+ \circ  {\bf M}^+$$
for $c \in \bR$, ${\bf M},{\bf N} \in E_\gg$.
 Summing up, $E_\gg$, equipped with the involution  $^+$, is therefore  a {\em real  unital $^*$-algebra}.
\begin{definition}\label{defsim}
{\em Let $\gg$ be a Lie algebra.
The  real involution  $E_\gg \ni {\bf M} \mapsto  {\bf M}^+ \in  E_\gg$  defined by  (\ref{Ptilde}) is called the {\bf involution} of $E_\gg$.
An element ${\bf M} \in E_\gg$ is said to be {\bf symmetric} if ${\bf M}={\bf M}^+$.}
\end{definition}

\section{Proofs of some propositions}\label{appProof}

\noindent {\bf Proof of Proposition \ref{prop2}}
\begin{proof}
All the proof is based on the theory developed in  Appendix \ref{secstatic}. Point (7) is straightforward. Let us prove items (1) and (2). (1) If $B=A+iA$, in particular $D(B)=D(A)+iD(A)$, then $C(D(B))\subset D(B)$ and $CBy=BCy$ for every $y\in D(B)$, in other words  $CB\subset BC$.  Let us prove the converse implication. Suppose that
$B : D(B) \to \sH_\bC$ is a $\bC$-linear operator. 
First consider $\sH_\bC$ as a {\em real} vector space, define the non-empty real subspace $D(A):= \{x \in \sH \:|\: x+i0 \in D(B)\}$. 
Since $CB\subset BC$, in particular $C(D(B))\subset D(B)$. By direct inspection it follows that $x+iy\in D(B)$ if and only if $x,y\in D(A)$, i.e. $D(B)=D(A)+iD(A)$. So, we can see $B$ a $\bR$-linear operator from $D(A)\times D(A)$ to $H\times H$ and, as such, it can be represented as
$$B= \left[\begin{matrix}E & F \\ G & H\end{matrix}\right]\:,$$
where $E,F,G,H : D(A) \to \sH$ are $\bR$-linear operators. Since $B$ is actually $\bC$-linear, it 
must commute with 
$$J:=\left[\begin{matrix}0 & -I\\ I &0\end{matrix}\right]\:,$$
which corresponds to $iI$ on $\sH_\bC$ viewed a proper complex vector space: $JB \subset BJ$. This inclusion, by direct inspection implies  $G=-F$ and $E=H$. If we finally impose also the constraint $CB\subset BC$, where $C=\left[\begin{matrix}I & 0\\ 0 & -I\end{matrix}\right]$, we easily obtain $F=0$ so that
$$B= \left[\begin{matrix}A & 0\\ 0 & A\end{matrix}\right]\:,$$
where $A:= E$. This is the same as saying $B= A_\bC$.\\ 
To conclude, observe that $CB\subset BC$ implies $BC^{-1}\subset C^{-1}B$. However $C=C^{-1}$ so that $CB\subset BC \subset CB$
and thus $CB=BC$.\\
(2) First recall that a subspace $ \sM \subset \sH$ is dense if and only if the subspace $\sM_\bC = \sM + i\sM$ of $\sH_\bC$ is dense. Therefore $D(A)$ is dense iff $D(A_\bC)$ is dense and thus $A^*$ and $(A_\bC)^*$ are simultaneously well defined.
Applying the definition of the domain of the adjoint we find,
$$D((A_\bC)^*)=\{x+iy\in \sH_\bC\ |\ \exists s+it\in \sH_\bC\ |\ (s+it|u+iv)=(x+iy |Au + i Av)_\bC,\ u,v \in D(A)\}\:.$$
Restricting ourselves to the case $v=0$, decomposing the inner product into real and imaginary parts we find $x,y \in D(A^*)$ and $(A_\bC)^*(x+iy)=A^*x+iA^*y$, hence $(A_\bC)^*\subset (A^*)_\bC$.  The converse inclusion
 is immediate, concluding the proof of  $(A_\bC)^*=(A^*)_\bC$.\\
The proof of (5) is an immediate consequence of (2) and the definition of $A_\bC$.
The proofs of items (3),(4),(6),(8) are direct applications of the given definitions  and the theory developed in  Appendix \ref{secstatic}. Regarding (6), the statement about the spectrum immediately arises from the definitions of 
$A_\bC$ and the definitions of  the various parts of the spectrum. The first statement in (6) 
can be established as follows. 
First, notice that $E\mapsto P^{(A)}_E$ is a PVM on $\sH$ if and only if $E\mapsto (P^{(A)}_E)_\Delta$ is a PVM on $\sH_\bC$. Moreover, with obvious notation, $u,v\in \Delta_f^{(P)}$ if and only if $u+iv\in \Delta_f^{(P_\bC)}$ for any measurable function $f:\bR\rightarrow\bR$. This easily follows from $\mu_{u+iv}^{(P_\bC)}=\mu_u^{(P)}+\mu_v^{(P)}$. So, take $f(\lambda)=\lambda$, $u,v \in D(A)$ and $x,y\in \sH$. It holds
$$(x+iy|A_\bC (u+iv))_\bC = (x|A u) - (y|A v) + i \left((y|A u) - (x|A v)  \right)\:.$$
Using the very definition of complexified operator, Thm \ref{st}, identity (\ref{intop2}), and elementary properties of the measures $\mu^{(P)}_{r,s}$,
$$(x|A u) - (y|A v) + i \left((y|A u) - (x|A v)  \right)
 = \int_{\bR} \lambda d\mu^{(P^{(A)})}_{x,u}(\lambda)
- \int_{\bR} \lambda d\mu^{(P^{(A)})}_{y,v}(\lambda)$$
$$ + i \left(\int_{\bR} \lambda d\mu^{(P^{(A)})}_{y,u}(\lambda)
- \int_{\bR} \lambda d\mu^{(P^{(A)})}_{x,v}(\lambda) \right)= 
 \int_{\bR} \lambda d\mu^{(P^{(A)}_\bC)}_{x+iy,u+iv}(\lambda)\:.$$
Summing up, we have found for $u+iv \in D(A_\bC)$ and $x+iy \in \sH_\bC$
$$(x+iy|A_\bC (u+iv))_\bC =  \int_{\bR} \lambda d\mu^{(P^{(A)}_\bC)}_{x+iy,u+iv}(\lambda)\:,$$
and thus 
$$A_\bC = \int_{\bR} \lambda d(P^{(A)})_\bC(\lambda)$$
which implies $P^{(A_\bC)}= (P^{(A)})_\bC$ by the uniqueness statement in Thm \ref{st} (a). A similar argument applies to generic measurable functions, proving the last statement.  Point (10) has a direct proof using the definition of $A_\bC$ and $\sH_\bC$.
\end{proof}

\noindent {\bf Proof of Theorem \ref{Polar}}
\begin{proof}
The proof of (a) and (b) for the complex case can be found in \cite{M}. In the rest of the proof it is useful to notice that  the bounded operator $U$ is isometric  on  $\overline{Ran(P)}$ by continuity and that, since ((e) Remark \ref{remarkclosure}) $[Ran(P)]^\perp=Ker(P)$,  $U$ also vanishes on $[Ran(P)]^\perp$.\\ Now, suppose that $\sH$ is real and let us demonstrate (a) and (b) with this hypothesis. Consider the complexifications $\sH_\bC$ and $A_\bC$. $D(A_\bC)=D(A)_\bC$ is clearly dense since $D(A)$ is dense, furthermore  $A_\bC$ is closed thanks to Proposition \ref{prop2} (3). Hence we can apply (a) and (b) for the complex case, obtaining  that $(A_\bC)^* A_\bC$ is densely defined and selfadjoint. 
Furthermore the polar decomposition $A_\bC=U'P'$, with $P'=|A_\bC|$ where  $P', U'$ satisfy all properties listed in (b). 
Notice that $(A^*A)_\bC = (A_\bC)^*A_\bC$ from (2) and (7) of Proposition \ref{prop2}, and (5) implies that $A^*A$ is densely defined and selfadjoint since $A_\bC^*A_\bC$ is densely defined and selfadjoint. $A^*A$ and $(A_\bC)^*A_\bC$ are evidently positive and so item (6) of Proposition \ref{prop2} and Prop.\ref{sqrt} guarantee that $|A_\bC| = \sqrt{A_\bC^*A_\bC}=  \sqrt{(A^*A)_\bC} =  (\sqrt{(A^*A)})_\bC  = |A|_\bC$. Define $P:=|A|$.
Of course $D(P)=D(P')|_\sH=D(A_\bC)|_\sH=D(A)$. Now we need to prove that $U'$ is the complexification of some $\bR$-linear operator on $\sH$. 
This is equivalent to demonstrate that $U'C=CU'$ as stated by Prop.\ref{prop2} (1), where $C$ is the conjugation defined in (\ref{conj}). 
Let $x+iy\in Ran(P')$, then $x+iy=P'(u+iv)$ for some $u,v\in D(P)$. So we have $U'C(x+iy)=U'CP'(u+iv)=U'P'C(u+iv)=A_\bC C(u+iv)=CA_\bC(u+iv)=CU'P'(u+iv)=CU'(x+iy)$. So, by continuity of $CU',U'C$ we get $U'C=CU'$ on $\overline{Ran(P')}$. Now, it is easy to see that $C(Ker(P'))\subset Ker(P')$ and so $U'C=CU'$ on $[Ran(P')]^\perp$ trivially. Hence there must exist $U \in \gB(\sH)$ such that $U'=U_\bC$. Putting all together we find $A_\bC=U_\bC P_\bC=(UP)_\bC$, so that 
$A=UP$ where $P\geq 0$ is selfadjoint as $P'$ is (see Prop.\ref{prop2}). This way we have proved items (i) and (ii) together with (v). The properties (iii),(iv),(vi) and (vii) can be trivially obtained from the corresponding properties of $U'$ and $P'$ exploiting Prop.\ref{propcompleX} (c) and Def.\ref{defcomplx0}.
\end{proof}
%
%

\noindent{\bf Proof of Proposition \ref{LemmaCOMM}}
\begin{proof}
(a) Remember that $\psi\in D(A)$ if and only if exists $\frac{d}{dt}\big|_0 e^{tA}\psi$ exists. The equality $Be^{tA}\psi=e^{tA}B\psi$ and the continuity of $B$ guarantees that $B\psi\in D(A)$ and $AB\psi=BA\psi$, i.e. $BA\subset AB$. Let us prove the opposite inclusion. If $\sH$ is complex the proof can be found in \cite{M} using the self adojint operator $iA$. So, suppose $\sH$ is real and take $A_\bC$ on the complexified space $\sH_\bC$. Applying the complex case to $B_\bC A_\bC\subset A_\bC B_\bC$ we get $B_\bC e^{tA_\bC}=e^{tA_\bC}B_\bC$, hence $Be^{tA}=e^{tA}B$.\\
(b) First suppose that $\sH$ is complex. The operator $A^*A$ is densely defined, positive and selfadjoint as we know.  Since $P= |A|=\sqrt{A^*A}$ is selfadjoint,  we have ((e) in Remark \ref{remarkclosure})
$Ker(|A|)^\perp=\overline{Ran(|A|)}$ and  $H=Ker(|A|)\oplus \overline{Ran(|A|)}$.
$BA\subset AB$ implies  $BAA\subset ABA$ and thus  $BAA\subset AAB$. This inclusion can be rewritten as 
 $BA^*A\subset A^*AB$ because $A^*=\pm A$.
As $A^*A$ is selfadjoint and $B$ bounded, the found inclusion extends to all measurable functions of $A^*A$:  $Bf(A^*A)\subset f(A^*A)B$ (Thm 9.35 in \cite{M}). In particular, we have
 $B|A|=B\sqrt{A^*A}\subset \sqrt{A^*A}B=|A|B$ which is the second of the pair of relations we wanted to establish.
Now, let $u\in D(|A|)=D(A)$, from the proved inclusion  we immediately  have $UB|A|u= U|A|Bu=ABu=BAu=BU|A|u$, from which we see that $UB=BU$ on $Ran(|A|)$ and thus
on $\overline{Ran(|A|)}$ by continuity. If we manage to prove that this equality holds also on $Ker(|A|)$, the proof is complete for the complex Hilbert space case because $H=Ker(|A|)\oplus \overline{Ran(|A|)}$.
Let $u\in Ker(|A|)$, then $|A|Bu=B|A|u=0$, that is $Bu\in Ker(|A|)$. Since $Ker(|A|)= Ker(U)$ (Thm \ref{Polar} (b)) and $Ker(|A|)$ is invariant under the action of $B$, it immediately follows that $UBx=BUx$ trivially for  $x\in Ker(|A|)$ as wanted, concluding the proof for the complex Hilbert space case.\\
 Now, suppose that $\sH$ is real and let $A=UP$ the polar decomposition of $A=\pm A^*$. Take $B$ as in the hypotheses and complexify everything, then we have $B_\bC A_\bC\subset A_\bC B_\bC$ on the natural domains. As we know from (c) in   Theorem \ref{Polar} $A_\bC= U_\bC P_\bC$ is  the polar decomposition of $A_\bC$, hence, using the first part of the proof we get $B_\bC U_\bC =U_\bC B_\bC$ and  $B_\bC P_\bC\subset U_\bC B_\bC$, which respectively means $(BU)_\bC=(UB)_\bC$ and
$(BP)_\bC \subset (PB)_\bC$ so that $BU=UB$ and
$BP \subset PB$ by restriction to $\sH$.\\
(c) Let first suppose that $\sH$ is complex and $A=-A^*$. In this case Theorem \ref{ST} (ii) implies that 
$e^{-tA}$ (which belongs to $\gB(\sH)$) commutes with $A$ and thus, exploiting (b),  we have $Ue^{-tA} =  e^{-tA}U$ and 
$U^*e^{tA} =  e^{tA}U^*$ taking the adjoint.  Due to Thm \ref{ST}, the limit for $t\to 0$ in both cases yields $UA \subset AU$ and $U^*A\subset AU^*$.
Remaining in the complex case, if $A^*=A$, replacing $A$ for $iA$ everywhere in our reasoning,  we again reach the same 
final result $UA \subset AU$ and $U^*A\subset AU^*$.  Now assume $A=A^*$ (otherwise everywhere replace $A$ for $iA$). As $U$ and $U^*$ are bounded, we conclude (Thm 9.35 in \cite{M}) that  $U$ and $U^*$ commute with every measurable  function of  $A$. In particular  $U|A|\subset |A|U$ and  $U^*|A|\subset |A|U^*$. Exploiting (iii) Thm 9.35 in \cite{M} once again, 
we prove that $Uf(|A|)\subset f(|A|)U$ and  $U^*f(|A|)\subset f(|A|)U^*$ for every measurable function $f: [0,+\infty) \to \bR$.
We have so far established  (c) for a complex Hilbert space $\sH$. If $\sH$ is real and $A$ (anti)-selfadjoint, $A_\bC$ fulfills 
(c) in the complex Hilbert space $\sH_\bC$. (c) Thm \ref{Polar}  and (2),(6),(7) Prop.\ref{prop2}  easily extend the result to $A$.\\
(d) We  prove that $U^*= \pm U$  if, respectively,  $A^*=\pm A$. Since $U$ is bounded, we have $(\pm U)|A|=\pm A=A^*=|A|U^*$. Take $u\in D(|A|)$, then $U^*|A|u=|A|U^*u=(\pm U)|A|u$, hence $U^*x=\pm Ux$ for $x \in Ran(|A|)$ and, by continuity, $x \in \overline{Ran(|A|)}$. Since $H=\overline{Ran(|A|)}\oplus Ker(|A|)$ we have to prove that $U^*x=\pm Ux$  holds also for
  $x \in Ker(|A|)$. Since $Ker(|A|)= Ker(U)$ by Thm \ref{Polar}, we have $Ux=0$ if $x \in Ker(|A|)$. By proving  $Ker(|A|)\subset Ker(U^*)$ we would have $U^*x=0$, establishing $U^*x=\pm Ux$ also for $x\in Ker(|A|)$ as required.  To this end, let $x\in Ker(|A|)$ and $y\in\sH$.  We have  $y=u+v$, with $u\in \overline{Ran(|A|)}$ and $v\in Ker(|A|)$. Let $|A|x_n \in  Ran(|A|)$ such that $u=\lim_{n\to\infty}|A|x_n$, then we have $(U^*x|y)=(x|Uy)=(x|Uu)=\lim_{n\to\infty}(x|U|A|x_n)=\lim_{n\to\infty}(x||A|Ux_n)=\lim_{n\to\infty}(|A|x|Ux_n)=\lim_{n\to\infty}(0|Ux_n)=0$. Since $y$ is arbitrary, we have $U^*x=0$ if $x \in Ker(|A|)$ as required, proving our thesis $U^*x= \pm Ux$ for all $x\in \sH$.\\
(e) We exploit here  Thm \ref{Polar}  several times.
Since $\sH = Ker(|A|) \oplus \overline{Ran(|A|)}$ and $U$ is isometric on $\overline{Ran(|A|)}$, if $Ker(|A|)$  (which coincides with $Ker(|A|)= Ker(U)$)   is trivial, then 
$U$ is isometric on $\sH$. It is therefore  enough proving that $Ran(U)= \sH$ to end the proof of the fact that $U$ is unitary.
We know from Thm \ref{Polar} that $Ran(U) = \overline{Ran(U)}$, but since $U=\pm U^*$, we also have  $\overline{Ran(U)} = \overline{Ran(U^*)} = Ker(U)^\perp  = Ker(|A|)^\perp = \{0\}^\perp = \sH$.\\
To conclude demonstrating the last statement of (d), observe that if $U$ is unitary and $US\subset SU$, $U^*S\subset SU^*$ simultaneously hold (in particular $U(D(S)) \subset D(S)$ and $U^*(D(S)) \subset D(S)$), we also have
 $U^*USU^*\subset U^*SUU^*$, that is   $SU^*\subset U^*S$.  The found inclusion  together with $U^*S\subset SU^*$ implies
$U^*S=SU^*$. Interchanging the r\^ole of $U$ and $U^*$, we also achieve $US=SU$. 
\end{proof}

\noindent {\bf Proof of Proposition \ref{polarCOMM}}
\begin{proof}
From $e^{sB}e^{tA}=e^{tA}e^{sB}$ and Stone's theorem, we have
$e^{sB}A \subset Ae^{sB}$. Thus (a) in Prop.\ref{LemmaCOMM} implies that both $e^{sB}|A| \subset |A|e^{sB}$ 
and $Ue^{sB}=e^{sB}U$.  Applying Stone's theorem again to the second result we have   $UB \subset BU$ and also
$U^*B\subset BU^*$ since $U=-U^*$ ((d) of Prop.\ref{LemmaCOMM}). We have so far established (i). Regarding (ii), observe that 
$UB \subset BU$ and (b) of Prop.\ref{LemmaCOMM} yield both $U|B| \subset |B|U$, which gives (ii) with the same reasoning carried out in proving (c) of Prop.\ref{LemmaCOMM}, and $UV=VU$ which gives (iii) immediately.
The last statement is a trivial consequence of the fact that $U$ is unitary if  any of $A$, $|A|$,  $U$ is injective as stated in
(d) of Prop.\ref{LemmaCOMM}.
\end{proof}

\noindent{\bf Proof of Proposition \ref{prop3i}}
\begin{proof} First of all, notice that the considered $\bR$-linear operator,  $A$, is also a $\bC$-linear operator and $D(A)$ is also a complex subspace of $\sH_J$
in view of Proposition \ref{prop2i}. (a) easily  arises by applying the definition of adjoint operator. (b) is immediate consequence of the fact that the identity map is an  isometry of metric spaces from $\sH$ to $\sH_J$.
(c) straightforwardly arises from (b). (d) is consequence of (a) and (b) and the relevant definitions. Let us prove (e).  
Let $\gB(\bR)\ni E\mapsto P^{(A)}_E$ be the PVM of $A$ on $\sH$, then, since $JA=AJ$, Theorem \ref{st} (c) (ii) guarantees that $JP^{(A)}_E=P^{(A)}_EJ$, hence $P^{(A)}$ is made of complex linear projectors and it immediately arises that it is a PVM also with respect to $\sH_J$. Moreover $\mu_x^{(P)}(E)$, and so also $\Delta_\lambda^{(P)}$ (which equals $D(A)$ on $\sH$), turns out to be equal if defined on $\sH$ or $\sH_J$ and $\mu_{x,y}^{(P),\sH_J}(E)=\mu_{x,y}^{(P)}(E)-i\mu_{x,Jy}^{(P)}(E)$, with obvious notation, if $x\in\sH$ and $y\in\Delta_\lambda^{(P)}$. So, let $x,y$ as above, then, noticing that $Jy\in \Delta_\lambda^{(P)}$, we have 
$$
(x|Ay)_J=(x|Ay)-i(x|AJy)=\int_\bR\lambda\, d\mu_{x,y}^{(P)}-i\int_\bR\lambda\, d\mu_{x,Jy}^{(P)}=\int_\bR\lambda\, d\mu_{x,y}^{(P),\sH_J}
$$ 
From Prop.\ref{propint} (a) and Th.\ref{st} $P^{(A)}$ must be also the PVM of $A$ with respect to $\sH_J$.
Since the support of $P^{(A)}$ is the spectrum of $A$ both in the real and complex Hilbert space case (Theorem \ref{st}), the two notion of spectrum coincide. Since the point spectrum is the set of eigenvalues, which are the same considering $A$ as a complex-linear or real linear operator, the two notions of point spectrum coincide as well.
Since, for selfadjoint (either real or complex) operators $\sigma_c(A) = \sigma(A)\setminus \sigma_p(A)$, the result extends to continuous spectra.
\end{proof}

\noindent {\bf Proof of the Theorem \ref{teopropvnA}}
\begin{proof} (a) If $A^*=A \in \gR$ its PVM commutes with every bounded operator commuting with $A$ for theorem \ref{st} (c)(ii), so that the PVM is in $(\gR')'= \gR$.
If the PVM $P^{(A)}$ of $A^*=A$ belongs to $\gR$, the operators of the form $\int_\bR s dP^{(A)}$ belongs to $\gR$ for every simple function $s$. Since there exists a non-decreasing sequence of simple functions $s_n$ converging pointwise to $id : \bR \ni x \mapsto \bR$ (see, e.g., \cite{R}), from the second identity of (d) of proposition \ref{propint}
and theorem \ref{st} (a), the monotone convergence theorem implies that $A \in \gR$, since the latter is closed with respect to the strong operator topology.
(b)  The proof is identical in the real and complex case see, e.g. \cite{Redei}. \\
 (c) If $\gR$ is reducible there is a non-trivial subspace invariant under the action of every element of $\gR$. The orthogonal projector $P$ onto that space is therefore an element of $\gR'$, and thus $\cL_{\gR'}(\sH)$, different form $0$ and $I$.
If there is such an element in $\cL_{\gR'}(\sH)$, the (proper) projection subspace is invariant under $\gR$ which is not irreducible consequently.\\
(d) The proof arises form the fact that $T\in \gR$ can be decomposed as $T=S+A$ where both $S$ and $iA$ are selfadjoint  elements of $\gR$
and so, thanks to (a), they are strong limit of elements belonging to the *-algebra generated by $\cL_R(\sH)$ as seen in the proof of (a) itself. \\
(e) If $T \in \gR$ is selfadjoint, its PVM belongs to $\cL_\gR(\sH)$ and thus also to
 the von Neumann algebra $\cL_\gR(\sH)''$ proving (i). Regarding  (iii),
first observe that if $J$ exists in $\gR \setminus \cL_{\gR}(\sH)''$ 
then $\cL_{\gR}(\sH)'' \subsetneq \gR$. Let us prove the converse implication and (ii) simultaneously. If $A \in \gR \setminus \cL_{\gR}(\sH)''$, then $A-A^*$ does, otherwise $A \in  \cL_{\gR}(\sH)''$ because 
the selfadjoint operator  $A+A^*$ does. Due to proposition \ref{LemmaCOMM} (b) and $\gR=\gR''$, both the factors
 of the polar decomposition of $A-A^*= J|A-A^*|$ belong to $\gR$. Since  $|A-A^*|^*= |A-A^*| \in  \cL_{\gR}(\sH)''$, it must be 
$J \not \in  \cL_{\gR}(\sH)''$.  Proposition \ref{LemmaCOMM} (d) yields $J^*=-J$. Finally, from the general properties 
of the polar decomposition and Remark \ref{essentclos} (b), we know that $-JJ=JJ^*$ is the orthogonal projector onto $Ran J$, hence belongs to $\cL_{\gR}(\sH)$ 
because is selfadjoint and a product of elements of $\gR$.  This discussion also proves   that (ii) is true because,
if $A \in \gR$, then $A \in (\cL_{\gR}(\sH) \cup \cJ_\gR)''$ since $A =
 \frac{1}{2}(A+A^*) +  \frac{1}{2}J|A-A^*|$ and we know that $\frac{1}{2}(A+A^*),\frac{1}{2}|A-A^*|\in\cL_\gR(\sH)''$ and $J\in\cJ_\gR$.
\end{proof}

\noindent{\bf Proof of Theorem \ref{teogarding}}
\begin{proof} The proof of (a) is the same for the real and the complex case and appears in Ch.10 of \cite{S0}. Equation (\ref{Ugf}) can be derived directly from the definitions, proving (b) for both the real and the complex case. The same holds for (\ref{Af}), noticing that every G\'arding vector $x$ is smooth for $U$, hence in particular every function $t\mapsto U_{\exp(t\V{A})}x$ is differentiable. This proves (c). As for (a), the proof of (d) is the same for both the two cases and can be found in Ch.10 of \cite{S0}, keeping in mind that $D_G^{(U)}$ equals the set of smooth vectors for $U$.
Now, let us prove (e). It follows directly from the universal property in Theorem \ref{UP} taking  $V$ as the {\em real} associative algebra of (either $\bR$-linear or $\bC$-linear depending on the nature of $\sH$) operators on $D_G^{(U)}$. Using item (d), a direct calculation shows that $u({\bf M})$ is symmetric 
whenever ${\bf M}= {\bf M}^+$.
Let us pass to (f). If $\sH$ is complex, Corollary 10.2.11 in Ch.10 of \cite{S0} establishes that $iu({\bf A})$ is essentially self adjoint. Since $iu({\bf A})\subset iA$ and both operators are symmetric and the second is selfadjoint, we have that 
$\overline{iu(A)} = \overline{iA} =((iA)^*)^*= iA$, which is equivalent to say that $\overline{u(A)}  = A$. If $\sH$ is real, it is convenient to pass to consider the strongly-continuous unitary representation  $U_\bC$  of $G$ on $\sH_\bC$ obtained by complexification $(U_g)_\bC$ of the operators $U_g$.  Exploiting Theorem \ref{ST} one immediately proves that, if $A$ is the anti-selfadjoint generator associated to ${\bf A} \in \gg$ by the real representation $U$, the complexified operator $A_\bC$ is the 
anti selfadjoint generator associated to ${\bf A} \in \gg$ by the complex representation $U_\bC$. 
The G\r{a}rding space $D^{(U_\bC)}_G$ of $U_\bC$ is nothing but the complexified one $(D^{(U)}_G)_\bC$ as arises from Thm \ref{DMtheorem}. 
Therefore, restricting  the operators to the G\r{a}rding domains one has 
$(A_\bC)|_{D^{(U_\bC)}_G} = (A|_{D^{(U)}_G})_\bC$, that is $u_\bC({\bf A})=u({\bf A})_\bC$, where we denoted by $u_\bC$ the 
Lie-algebra representation associated with $U_\bC$. Applying the (already proved) complex case of  (f) to $A_\bC$, we have
$A_\bC = \overline{u_\bC({\bf A})} = \overline{u({\bf A})_\bC} = (\overline{u({\bf A})})_\bC$, where the last identity arises from (3) in Proposition \ref{prop2}. We have obtained that
$A_\bC =  (\overline{u({\bf A})})_\bC$ which is equivalent to  our thesis $A = \overline{u({\bf A})}$. Finally, let us pass to the proof of (g). If $\sH$ is complex the thesis follows immediately from Theorem 10.2.6 in Ch.10 of \cite{S0}. If $\sH$ is real, taking into account (e) and the fact that $u_\bC(\V{A})=u(\V{A})_\bC$ for every $\V{A}\in\gg$, we easily get $u_\bC(\V{M})=u(\V{M})_\bC$. From this equation and Prop.\ref{prop2}, it follows $(u(\V{M})^*)_\bC=(u(\V{M})_\bC)^*=u_\bC(\V{M})^*=\overline{u_\bC(\V{M})}=\overline{u(\V{M})_\bC}=(\overline{u(\V{M})})_\bC$ which concludes the proof.
\end{proof}

\noindent{\bf Proof of Proposition \ref{alggroupcomm}}
\begin{proof}
If (i) holds, by definition of the involved domains $B(D_G^{(U)})\subset D_G^{(U)}$.  Since $B$ is bounded and $u({\bf A})$
closable, and exploiting Remark \ref{defclosop} (a), we immediately achieve (ii).
Suppose now that (ii) holds, that is  $B\overline{u({\bf A})}\subset\overline{u({\bf A})}B$, then Prop.\ref{LemmaCOMM} (a) gives $Be^{t\overline{u(\V{A})}}=e^{t\overline{u(\V{A})}}B$.
This is true both for the real and the complex Hilbert space cases.
 Since the group $G$ is connected, every element $g \in G$ can be written as the product of a finite number of one-parameter subgroup elements of $G$, hence the thesis (iii) holds true. 
Regarding the fact that  (iii) entails (i),  assume that (iii) is valid, i.e.,  $BU_g=U_gB$ for every $g\in G$. In particular, we therefore have
 $Be^{t\overline{u({\bf A})}}=BU_{\exp(t{\bf A})}=U_{\exp(t{\bf A})}B=e^{t\overline{u({\bf A})}}B$ for every ${\bf A}\in\gg$ and $t\in\bR$. Exploiting again Prop.\ref{LemmaCOMM} (a) we get (ii), namely $B\overline{u({\bf A})}\subset \overline{u({\bf A})}B$. However also (i) is valid
 because, if (iii) is satisfied,  $B\left(\int_G f(g)U_gx d\mu \right)=\int_G f(g)U_g(Bx) d\mu$, hence the G\r{a}rding domain is invariant under the action of $B$ and thus from (ii) we pass to (i).
\end{proof}

\noindent{\bf Proof of Proposition \ref{irrelie}}

\begin{proof}
Let $x\in D_N^{(U)}$ and ${\bf A}\in\gg$. Thanks to Theorem \ref{teonelson} it holds $x\in D_G^{(U)}$ and $x$ is analytic for $u({\bf A})$. Exploiting Prop.\ref{propNseries}, we have that  there exists $t_{{\bf A},x}>0$ such that
$$
U_{\exp(t{\bf A})}x=e^{t\overline{u({\bf A})}}x=\sum_{n=0}^\infty\frac{t^n}{n!}u({\bf A})^nx,\ \ |t|\le t_{{\bf A},x}\:.
$$
Moreover $D_N^{(U)}$ is invariant under the action of $u$, hence $u({\bf M})x\in D_N^{(U)}$. Then there exits $t_{{\bf A},u({\bf M})x}>0$ such that
$$
U_{\exp(t{\bf A})}u({\bf M})x=e^{t\overline{u({\bf A})}}u({\bf M})x=\sum_{n=0}^\infty\frac{t^n}{n!}u({\bf A})^nu({\bf M})x,\ \ |t|\le t_{{\bf A},u({\bf M})x}\:.
$$
Now take a positive  real  $t_x<\min\{t_{{\bf A},x},t_{{\bf A},u({\bf M})x}\}$. Using $[u({\bf M}),u({\bf A})]=0$ we have
$$
U_{\exp(t{\bf A})}u({\bf M})x=\sum_{n=0}^\infty\frac{t^n}{n!}u({\bf M})u({\bf A})^nx,\ \ |t|\le t_{x}\:.
$$
Since   $u({\bf M})$ is closable, it follows directly from the equations above and the invariance of $D_G^{(U)}$ under the action of
 $U$ that $$U_{\exp(t{\bf A})}u({\bf M})x = \sum_{n=0}^\infty\frac{t^n}{n!}u({\bf M})u({\bf A})^nx = u({\bf M})U_{\exp(t{\bf A})}x$$ for every $|t|\le 
t_x$. Actually this equality holds for every $t\in\bR$. Indeed define 
$\cZ:=\{z>0|u({\bf M})U_{\exp(t{\bf A})}x=U_{\exp(t{\bf A})}u({\bf M})x,\ |t|\le z\}$ and let $t_0:=\sup\cZ$. Suppose that $t_0<\infty$, then it is
 easy to see that the fact that  $u({\bf M})$ is closable ensures that $u({\bf M})U_{\exp(t_0{\bf A})}x=U_{\exp(t_0{\bf A})}u({\bf M})x$, hence
 $t_0\in\cZ$. We know that $y:=U_{\exp(t_0{\bf A})}x\in D_N^{(U)}$,  we can therefore  repeat the above reasoning  finding a real
 $t_y>0$ such that $u({\bf M})U_{\exp(t{\bf A})}y=U_{\exp(t{\bf A})}u({\bf M})y$ for every $|t|\le t_y$.  Noticing that  
 $\exp((t+t_0){\bf A})=\exp(t{\bf A})\exp(t_0{\bf A})$, it straightforwardly follows  that
 $u({\bf M})U_{exp(t+t_0){\bf A}}x=U_{exp(t+t_0){\bf A}}u({\bf M})x$ for $|t|\le t_y$, hence $t_0+t_y\in\cZ$, which is in contradiction with  the
 definition of $t_0$. This proves that $t_0=\infty$. As is well known from the elementary theory of Lie-group theory, since the
 $G$ is connected, every element is the product of a finite number of elements belonging to one parameter subgroups generated
 by $\gg$, so that  we have actually demonstrated  that $u({\bf M})U_g=U_gu({\bf M})$ on $D_N^{(U)}$ for every $g\in G$. This identity  
implies  $U_gu({\bf M})|_{D_N^{(U)}}=u({\bf M})|_{D_N^{(U)}}U_g$ on the natural domains thanks to the invariance of the Nelson space under the action of the 
group representation. In our hypotheses, $u({\bf M})|_{D_N^{(U)}}$ is  the restriction of a closable operators and thus  it is closable as 
well and so Remark \ref{essentclos} (f) gives $U_g\overline{u({\bf M})|_{D_N^{(U)}}}=\overline{u({\bf M})|_{D_N^{(U)}}}U_g$ for 
every $g$. Using Proposition \ref{SL2} we find $D(\overline{u({\bf M})|_{D_N^{(U)}}})=\sH$ and 
$\overline{u({\bf M})|_{D_N^{(U)}}}\in\gB(\sH)$, more precisely $\overline{u({\bf M})|_{D_N^{(U)}}}=aI+bJ$ for some $a,b\in\bR$, where 
$J=iI$ if $\sH$ is complex, and $J$ is a generic complex structure if $\sH$ is real.
Since $\overline{u({\bf M})|_{D_N^{(U)}}}\subset \overline{u({\bf M})}$, the maximality of the domain gives 
$\overline{u({\bf M})|_{D_N^{(U)}}}=\overline{u({\bf M})}$. As the  latter is selfadjoint, it follows that $b=0$ and $\overline{u({\bf M})}=aI$ 
with $a\in\bR$ ending the proof.
\end{proof}

\noindent{\bf Proof of Lemma  \ref{lemmaIR}}
\begin{proof}
Let $A\in\cL(\sH)'$, then in particular $AP_\psi=P_\psi A$ where $P_\psi$ is the orthogonal projector onto the subspace generated by  $\psi\in\sH\setminus\{0\}$. This gives $A\psi=AP_\psi\psi=P_\psi A\psi$ which means that for every $\psi\in\sH\setminus\{0\}$ it holds $A\psi=\lambda_\psi \psi$  for some (and unique) $\lambda_\psi\in\bR$.  If $\dim H =1$ the proof ends, otherwise  if $\phi\in\sH\setminus\{0\}$ is linearly independent from a given $\psi\in \sH\setminus\{0\}$, with the same argument it holds $\lambda_{\phi+\psi}(\phi+\psi)=A(\phi+\psi)=\lambda_\phi\phi+\lambda_\psi\psi$, that is $(\lambda_{\phi+\psi}-\lambda_\phi)\phi=(\lambda_\psi-\lambda_{\phi+\psi})\psi $.  By linear independence it must be $\lambda_{\phi+\psi}-\lambda_\phi=\lambda_\psi-\lambda_{\phi+\psi}=0$, that is $\lambda_\phi=\lambda_\psi$. Finally, complete the initial $\psi$ -- assumed to have unit norm -- to a Hilbert  basis $\{e_i\}_{i\in\cI}$  of $\sH$ so that $Ae_i=\lambda_\psi e_i$ for every $i\in\cI$. By linearity and continuity  of $A$ we have $A=\lambda_\psi I\in\bR I$.
\end{proof}

\noindent{\bf Proof of Lemma \ref{quatimp1}}
\begin{proof}
If $A,B\in\gR$ and $A+JB=0$, then $0=-K(A+JB)K=A-JB$. The first part of the thesis follows immediately. To conclude we  prove that $\gR_J:=\gR+J\gR$ is a complex von Neumann algebra whose commutant is trivial, this implies the second part of the thesis because $\gR_J=\gR_J''= \{cI \:|\: c \in\bC\}' = \gB(\sH_J)$.  $\gR_J$ is evidently a unital $^*$-subalgebra of $\gB(\sH_J)$, hence we only need to prove that it is closed with respect to the strong operator topology and that its commutant is made up of complex scalars. Suppose that $A_n+JB_n\to T\in\gB(\sH_J)$ strongly. Since all the operators considered are also real linear and the norms in $H$ and $H_J$ coincide, the same strong limit holds in $\gB(\sH)$. By continuity of $K$ we have that $A_nK+JB_nK\rightarrow TK$ and $A_nK-JB_nK=KA_n+KJB_n\rightarrow KT$ strongly. From this it easily follows that $A_n\rightarrow A$ and $B_n\rightarrow B$ for some real linear operators $A,B\in \gB(\sH)$. Since $\gR$ is strongly closed we have $A,B\in\gR$ so that  $A_n+JB_n\rightarrow A+JB\in\gR_J$. 
We have established that $\gR+J\gR$ is a complex von Neumann algebra.
Finally, take $T\in\gR_J'$, in particular we have $[T,A]=0$ for every $A\in\gR\subset\gR_J$.  Since $T$ is also real linear, it must be $T=aI+bJ+cK+dJK$ for some $a,b,c,d\in\bR$. Since it also holds $[T,J]=0$, it must be $T=aI+bJ$. In other words 
$T = (a+ib)I$ which is equivalent to say that $\gR_J$ has trivial commutant concluding the proof.
\end{proof}

\noindent{\bf Proof of Lemma \ref{commutanquaternionic}}
\begin{proof}
Let $A\in\cL(\sH)'$, then, reasoning as in the proof of Lemma
\ref{lemmaIR} we conclude that for every $\psi\in\sH\setminus\{0\}$ it holds $A\psi=\psi\lambda_\psi$ for some and unique $\lambda_\psi\in\bH$. Again, if $\phi\in\sH\{0\}$ is linearly independent from a given $\psi\in\sH\setminus\{0\}$, we find $\lambda_\phi=\lambda_\psi$. Next step consists in proving that $\lambda_\psi\in\bR$. Let $p\in\bH\setminus\{0\}$. Clearly, if $\psi,\phi$ are linearly independent then so are $\psi,\phi p$, hence $\lambda_\phi=\lambda_\psi=\lambda_{\phi p}.$ Now, we have $(\phi p)\lambda_\phi=(\phi p)\lambda_{\phi p}=A(\phi p)=(A\phi)p=(\phi\lambda_\phi)p$ from which it immediately follows $p\lambda_\phi=\lambda_\phi p$. Being $p$ generic, $\lambda_\phi$ must be real. The conclusion follows as in the proof of Lemma \ref{lemmaIR}.

\end{proof}

\end{document}